\DeclareRobustCommand*\cal{\relax\mathcal}

\documentclass[11pt]{article}
\usepackage{times}
\usepackage{amsmath,amssymb,amsthm,amsfonts,latexsym,bbm,xspace,graphicx,float,mathtools,dsfont,bm}
\usepackage{hyperref}
\hypersetup{colorlinks,citecolor=blue}
\usepackage[dvips, paper=letterpaper, top=1in, bottom=1in, left=1in, right=1in, nohead, includefoot, footskip=.25in]{geometry}

\usepackage{color}
\usepackage{comment}
\usepackage{algorithm}
\usepackage{algorithmic}

\usepackage{appendix}
\usepackage[nomessages]{fp}
\usepackage{multirow}

\usepackage{enumitem}
\usepackage{booktabs}

\usepackage{flushend}
\usepackage{enumitem}

\usepackage{mathtools}

\newtheorem{theorem}{Theorem}
\newtheorem{corollary}{Corollary}
\newtheorem{lemma}{Lemma}

\newtheorem{definition}{Definition}

\newtheorem{observation}{Observation}
\newtheorem{example}{Example}
\newtheorem{remark}{Remark}

\newcommand{\hG}{\widehat{G}}
\newcommand{\hz}{\hat{z}}
\newcommand{\hx}{\hat{x}}
\newcommand{\hy}{\hat{y}}
\newcommand{\cO}{\mathcal{O}}

\newcommand{\FF}{\mathcal{F}}
\newcommand{\VV}{\mathcal{V}}

\newcommand{\dD}{\mathcal{D}}

\newtheorem{informaltheorem}{Informal Theorem}

\newcommand{\rev}{\textsc{Rev}}

\newcommand{\opt}{\textsc{OPT}}

\def\mout{\ensuremath{\MM'}}
\def\Bs{\ensuremath{\textbf{s}}}
\def\Br{\ensuremath{\textbf{r}}}

\def\Bz{\ensuremath{{z}}}
\def\Balpha{{\alpha}}

\newcommand{\diff}{\textsc{Diff}}
\newcommand{\dist}{\text{dist}}
\newcommand{\wass}{d_w}

\newcommand{\revsecond}{\textsc{Rev-second}}

\newenvironment{prevproof}[2]{\noindent {\em {Proof of {#1}~\ref{#2}:}}}{$\Box$\vskip \belowdisplayskip}

\newcommand{\R}{\ensuremath{\mathbb{R}}} 

\newcommand{\poly}{{\rm poly}}

\newcommand{\yangnote}[1]{{\color{blue}{#1}}}
\newcommand{\mingfeinote}[1]{{\color{magenta}{#1}}}

\newcommand{\notshow}[1]{{}}

\newcommand{\hw}{\hat{\omega}}

\DeclareMathOperator*{\E}{\mathbb{E}}

\def \TT  {{\cal T}}
\def \rj{r^{(j)}}
\def \sk{s^{(k)}}
\def \PP  {{\cal P}}

\def \AA  {{\cal A}}
\def \OO  {{\cal O}}

\def \MM  {\mathcal{M}}
\def \WW  {\mathcal{W}}

\def\supp{\text{supp}}

\def \DD  {{\mathcal{D}}}

\DeclareMathOperator{\argmax}{argmax}
\DeclareMathOperator{\argmin}{argmin}

\definecolor{MyGray}{rgb}{0.8,0.8,0.8}

\begin{document}

\title{An Efficient $\varepsilon$-BIC to BIC Transformation and Its Application to Black-Box Reduction in Revenue Maximization}

\author{ Yang Cai\footnote{Supported by a Sloan Foundation Research Fellowship and the NSF Award CCF-1942583 (CAREER) .} \\Yale University, USA\\yang.cai@yale.edu
 \and Argyris Oikonomou\\Yale University, USA\\ argyris.oikonomou@yale.edu
  \and Grigoris Velegkas\\Yale University, USA\\ grigoris.velegkas@yale.edu
  \and Mingfei Zhao\\ Yale University, USA\\ mingfei.zhao@yale.edu
}


\maketitle

\begin{abstract}

We consider the black-box reduction from multi-dimensional revenue maximization to virtual welfare maximization. Cai et al.~\cite{CaiDW12a, CaiDW12b,CaiDW13a,CaiDW13b} show a polynomial-time approximation-preserving reduction, however, the mechanism produced by their reduction is only approximately Bayesian incentive compatible ($\varepsilon$-BIC). We provide two new \textbf{polynomial time} transformations that convert any $\varepsilon$-BIC mechanism to an exactly BIC mechanism with only a negligible revenue loss.
\begin{itemize}
    \item Our first transformation applies to any mechanism design setting with \emph{downward-closed outcome space} and only requires \emph{sample access} to the agents' type distributions.
    \item Our second transformation applies to the \emph{fully general outcome space},  removing the downward-closed assumption, but requires \emph{full access} to the agents' type distributions.
\end{itemize}
Both transformations only require query access to the original $\varepsilon$-BIC mechanism. Other $\varepsilon$-BIC to BIC transformations for revenue exist in the literature~\cite{DaskalakisW12, RubinsteinW15, CaiZ17} but all require exponential time to run in both of the settings we consider. As an application of our transformations,  we improve the reduction by Cai et al.~\cite{CaiDW12a, CaiDW12b,CaiDW13a,CaiDW13b} to generate an exactly BIC mechanism.

\end{abstract}

\thispagestyle{empty}
\thispagestyle{empty}
\addtocounter{page}{-1}
\newpage

\maketitle
\section{Introduction}\label{sec:intro}

Mechanism design is the study of optimization algorithms with the additional constraint of incentive compatibility. A central theme of algorithmic mechanism design is thus to understand how much this extra constraint hinders our ability to optimize a certain objective efficiently. In the best scenario, one may hope to establish an equivalence between a mechanism design problem and an algorithm design problem, manifested via a \emph{black-box reduction} that converts any algorithm to an incentive compatible mechanism. In this paper, we study the black-box reduction of a central problem in mechanism design: \emph{multi-dimensional revenue maximization}. 




The problem description is simple: an auctioneer is selling a collection of items to  one or more strategic bidders. We follow the standard Bayesian assumption, that is, each bidder's type is drawn independently from a distribution known to all other bidders and the auctioneer. The auctioneer's goal is to design a \emph{Bayesian incentive compatible} (BIC) mechanism that maximizes the \emph{expected revenue}. 

In the special case of single-item auction, Myerson provides an elegant characterization of the optimal mechanism. Indeed, Myerson's solution can be viewed as a black-box reduction from revenue maximization to the algorithmic problem of (virtual) welfare maximization~\cite{Myerson81}.  
{However, whether the black-box reduction can be extended to multi-dimensional settings remained open after Myerson's result. Recently, a line of work by Cai et al.~\cite{CaiDW12a, CaiDW12b,CaiDW13a,CaiDW13b} showed that} there is \emph{a polynomial-time approximation-preserving black-box reduction from multi-dimensional revenue maximization to the algorithmic question of (virtual) welfare optimization}. 
However, this result still has the following two caveats: (i) the revenue of the mechanism is only guaranteed to be within an additive $\varepsilon$ of the optimum; and (ii) the mechanism is only \emph{approximately} Bayesian incentive compatible. Thus, an immediate open problem following their result is whether these two compromises are inevitable. In this paper, we show that approximately Bayesian incentive compatibility is unnecessary through our first main result: 

\medskip \noindent \hspace{0.7cm}\begin{minipage}{0.92\textwidth}
\begin{enumerate}
\item[{\bf Result I:}] There is a polynomial-time approximation-preserving black-box reduction from multi-dimensional revenue maximization to the algorithmic question of (virtual) welfare optimization that generates an \textbf{exactly Bayesian incentive compatible} mechanism.
\end{enumerate}
\end{minipage}

\bigskip \textbf{Result I} is enabled by \textbf{a new polynomial time $\varepsilon$-BIC to BIC transformation {for revenue}}, which is our second main result:

\bigskip \noindent \hspace{0.7cm}\begin{minipage}{0.92\textwidth}
\begin{enumerate}
\item[{\bf Result II:}] There is a polynomial-time $\varepsilon$-BIC to BIC transformation that converts any approximately Bayesian incentive compatible mechanism to an exactly Bayesian incentive compatible mechanism with a negligible revenue loss {for any downward-closed environment} with only \emph{sample access} to the agents' type distributions. \end{enumerate}
\end{minipage}

\medskip The transformation is applicable to any downward-closed mechanism design setting.\footnote{Roughly speaking, the setting is downward-closed if the agents have the choice to not participate in the mechanism. See Section~\ref{sec:prelim} for the formal definition.} We believe the transformation is of independent interest and would have numerous applications in mechanism design. Indeed, our  black-box reduction follows straightforwardly from applying the transformation to the mechanism of Cai et al.~\cite{CaiDW12a, CaiDW12b,CaiDW13a,CaiDW13b}. Furthermore, if we are given \emph{full access} to the type distributions of the buyers, we can extend our transformation to the fully general outcome space, removing the downward-closedness assumption. Here is our third main result:

\medskip \noindent \hspace{0.7cm}\begin{minipage}{0.92\textwidth}
\begin{enumerate}
\item[{\bf Result III:}] There exists a polynomial-time $\varepsilon$-BIC to BIC transformation with oracle access to any approximately Bayesian incentive compatible mechanism and full access to the agents' distributions, outputs an exactly Bayesian incentive compatible mechanism for any \emph{general outcome space} with a negligible revenue loss. \end{enumerate}
\end{minipage}

\vspace{.1in}
 Note that other $\varepsilon$-BIC to BIC transformations {for revenue} have been proposed in the literature~\cite{DaskalakisW12, RubinsteinW15, CaiZ17}, however, all of the existing transformations require solving a $\# P$-hard problem repeatedly~\cite{hartline2015bayesian} and therefore cannot be made computationally efficient.



\subsection{Our Results and Techniques}\label{sec:results and techniques}
We first fix some notations to facilitate our discussion of the results. We consider a general mechanism design environment where there is a set of feasible outcomes denoted by $\cO$.
There are $n$ agents, and each agent $i$ has a type $t_i$ drawn from distribution $\DD_i$ independently. We use $\TT_i$ to denote the support of $\DD_i$, and for every $t_i\in \TT_i$,  $v_i(t_i,\cdot)$ is a valuation function that maps every outcome to a real number in [0,1]. A mechanism $\MM$ consists of an allocation rule $x(\cdot): \bigtimes_{i\in[n]} \TT_i\mapsto \Delta(\cO)$ and a payment rule $p(\cdot): \bigtimes_{i\in[n]} \TT_i\mapsto \mathbb{R}^n$. We slightly abuse notation to define $v_i(t_i,x(b))\equiv \E_{o\sim x(b)}[v_i(t_i,o)]$.  If we have query access to $\MM$, then on any query bid profile $b=(b_1,\ldots, b_n)$, we receive an outcome $o\sim x(b)$ and payments $p_1(b),\ldots, p_n(b)$.

The outcome space $\mathcal{O}$ is called \emph{downward-closed} if each $o\in \mathcal{O}$ can be written as a vector $o=(o_1,...,o_n)$ where $o_i$ is the outcome for agent $i$. And for every $o=(o_1,...,o_n)\in \mathcal{O}$, any $o'=(o_1',...,o_n')$ with $o_i'=o_i$ or $o_i'=\perp$ for every $i$ is also in $\cO$. Here $\perp$ is a null outcome  available to each agent $i$, which represents the option of not participating in the mechanism. 

Equiped with the notations, we are ready to discuss our $\varepsilon$-BIC to BIC transformations.

\begin{informaltheorem}[$\varepsilon$-BIC to BIC transformation on downward-closed outcome space]\label{informalthm:eps-BIC to BIC}
Given sample access to a collection of distributions $\left( \DD_i \right)_{i\in[n]}$ {on a downward-closed outcome space}, and query access to an $\varepsilon$-BIC and individually rational (IR) mechanism $\MM=(x,p)$ with respect to $\bigtimes_{i\in[n]} \DD_i$. We can construct another mechanism $\MM'$ that is exactly BIC and IR with respect to $\bigtimes_{i\in[n]} \DD_i$, and its revenue is at most {$O(n\sqrt{\varepsilon})$} worse than the revenue of $\MM$. Moreover, for any bid profile $b=(b_1,\ldots, b_n)$, $\MM'$ computes an outcome $o\in \cO$ and payments $p_1(b),\ldots, p_n(b)$ {in expected running} time {$\poly\left(\sum_{i\in[n]} |\TT_i|,1/\varepsilon\right)$} and makes {in expectation} at most {$\poly\left(\sum_{i\in[n]} |\TT_i|,1/\varepsilon\right)$}~queries~to~$\MM$.
\end{informaltheorem}

\begin{informaltheorem}[$\varepsilon$-BIC to BIC transformation on general outcome space]\label{informalthm:eps-BIC to BIC general outcome}
Given full access to the collection of distributions $\left( \DD_i \right)_{i\in[n]}$ on a general outcome space such that $\left|\supp(\DD_i)\right| \leq m$ for every $i$, and query access to an $\varepsilon$-BIC and individually rational (IR) mechanism $\MM$ with respect to $\bigtimes_{i\in[n]} \DD_i$.
We can construct a mechanism $\MM'$ that is exactly BIC and IR with respect to $\bigtimes_{i\in[n]} \DD_i$.
Moreover its revenue is within an additive {$O(n m \varepsilon)$} of the revenue of $\MM$.
Furthermore, the running time of the constructed mechanism is {$\poly\left(n,m,1/\varepsilon\right)$} and the mechanism makes at most {$\poly\left(n,m,1/\varepsilon\right)$}~queries~to~$\MM$.
\end{informaltheorem}

Previous transformations can produce an $\MM'$ with similar guarantees in the downward-closed setting but require $\poly(\prod_{i\in[n]} |\TT_i|)$ time to run~\cite{DaskalakisW12, RubinsteinW15, CaiZ17}. In the special case, where there exists symmetry in the agents' type distributions, the transformation can be improved to run in time $\poly(\sum_{i\in[n]} |\TT_i|)$, as the interim allocation probabilities and payments of the mechanism $\MM$ can be computed efficiently  via a polynomial-size LP.\footnote{If the interim allocation probabilities and payments of $\MM$ are given, the edge weights in the Replica-Surrogate matching can be computed efficiently. See the next paragraph for more details.} Our result achieves the $\poly(\sum_{i\in[n]} |\TT_i|)$ running~time~without~the~symmetry~assumption. 

\subsubsection{Our Result for the Downward-Closed Outcome Space}
We first discuss our result for the downward-closed outcome space. To illustrate our new ideas, we first briefly review the constructions in the literature. In the heart of all the previous constructions lies the problem called \emph{replica-surrogate matching}. 

\paragraph{Replica-Surrogate Matching}  For each agent $i$, form a bipartite graph $G_i$. The left hand side nodes are called replicas, which are types sampled {i.i.d.} from $\DD_i$.  In particular, the true type $t_i$ of agent $i$ is one of the replicas. On the right hand side, the nodes are called surrogates, which are also types sampled from $\DD_i$. The edge between a replica {with type} $t^{(j)}$ and a surrogate {with type} $t^{(k)}$ is assigned weight 
$w_{jk}\equiv \E_{t_{-i}\sim \DD_{-i}}[v_i(t^{(j)},x(t^{(k)},t_{-i})) -p_i(t^{(k)},t_{-i})]$~\footnote{The true weight $w_{jk}\equiv \E_{t_{-i}\sim \DD_{-i}}\left[v_i\left(t^{(j)},x\left(t^{(k)},t_{-i}\right)\right) -(1-\eta)\cdot p_i\left(t^{(k)},t_{-i}\right)\right]$ is computed using a discounted price, but we can ignore the difference for now.}, 
 which is the interim utility of agent $i$ when her true type is $t^{(j)}$ but reports $t^{(k)}$ to $\MM$. Compute the maximum weight matching on $G_i$. The true type $t_i$ selects a surrogate using the matching to compete~in~$\MM$. {Agent $i$ competes in $\MM$ using the type of the surrogate she is matched to in the maximum weight matching.}

The intuition is that since $\MM$ is not BIC, the true type $t_i$ may prefer the outcome and payment from reporting some different type. The matching is set up to allow the true type $t_i$ to pick a more favorable type to compete in $\MM$ for it. But why wouldn't the agent misreport in the matching? After all, the edge weights depend on the agent's report. As it turns out, to guarantee incentive compatibility, one needs to find a matching with a \emph{maximal-in-range} algorithm. Namely, the matched surrogate is selected to maximize the agent's induced utility less some cost that only depends on the outcome. It is not hard to verify that the maximum weight matching is indeed maximal-in-range, and therefore the agent has no incentive to~lie. 

But why does the maximum weight matching take exponential time to find? The problem is that we are not given the edge weights. For each edge, we can only sample from a distribution whose mean is the weight of the edge: Sample $t_{-i}$ from $\DD_{-i}$ and compute 
$v_i(t^{(j)},x(t^{(k)},t_{-i})) -p_i(t^{(k)},t_{-i})$.
Even if we assume that we know the distributions $(\DD_i)_{i\in[n]}$, it still takes $\poly( \prod_{j\neq i} |\TT_j|)$ time to compute the weight of a single edge exactly, which is already an exponential on $\sum_{i\in[n]} |\TT_i|$. But why can't we first estimate the edge weights with samples and find the maximum matching using the estimated weights? The issue is that no matter how many samples we take, the empirical mean will be off by some estimation error. The maximal-in-range property is so fragile that even a tiny bit of estimation error can cause the algorithm to violate the property, making the whole mechanism not incentive compatible. See Example~\ref{example:sample edges} below for a more detailed explanation.


\vspace{-.15in}
\paragraph{Black-box Reduction for Welfare Maximization}To overcome the difficulty, we turn to another important problem in mechanism design, black-box reduction for welfare maximization, for inspiration. A line of beautiful results~\cite{HartlineL10, BeiH11, hartline2015bayesian,  dughmi2017bernoulli} initiated by Hartline and Lucier shows that the mechanism design problem of welfare maximization in the Bayesian setting can be black-box reduced to the algorithmic problem of welfare maximization. The replica-surrogate matching is again the central piece in the reduction. Indeed, the idea of replica-surrogate matching was first proposed by Hartline et al.~\cite{HartlineKM11, hartline2015bayesian}, and later introduced 
 by Daskalakis and Weinberg~\cite{DaskalakisW12} to the study of $\varepsilon$-BIC to BIC transformation {for revenue}. The main difference of the two scenarios is the way the edge weights are defined. For welfare maximization, the edge weight between a replica $t^{(j)}$ and a surrogate $t^{(k)}$ is 
 $v_{jk}\equiv \E_{t_{-i}\sim \DD_{-i}}[v_i(t^{(j)},x(t^{(k)},t_{-i}))]$, namely,  the \emph{interim value} for agent $i$ when her true type is $t^{(j)}$ but reports $t^{(k)}$ to $\MM$. 
  We will refer to the one with interim utilities as edge weights the \emph{U-replica-surrogate matching} and the one with interim values as edge weights the \emph{V-replica-surrogate matching}. The main reason that we would like to distinguish the two settings is as follows: in a V-replica-surrogate matching all edge weights are nonnegative, while in a U-replica-surrogate matching the edge weights may be negative. The importance of the presence of negative edges will become clear soon.
  Obviously, it also takes exponential time to compute the exact maximum weight V-replica-surrogate matching due to the same reason discussed above.





A result by Dughmi et al.~\cite{dughmi2017bernoulli} shows how to circumvent this barrier for welfare maximization. Their solution has the following two main components: (i) a polynomial time maximal-in-range algorithm to solve the \emph{maximum entropy regularized perfect matching} problem; (ii) the \emph{fast exponential Bernoulli race}, a new Bernoulli factory~\footnote{A Bernoulli factory is an algorithm that with sample access to a $p$-coin to simulate a $f(p)$-coin. In Section~\ref{appx:bernoulli} we give a brief introduction to Bernoulli factories. We also refer the readers to \cite{keane1994bernoulli,nacu2005fast} and the references therein for more details.}, that allows them to execute the algorithm in (i) exactly with only sample access to distributions whose means are the edge weights. They use the algorithm to find a maximum entropy regularized V-replica-surrogate matching, and argue that this matching has approximately maximum weight, which allows them to conclude that their new mechanism loses at most a negligible fraction of the welfare.

\vspace{-.15in}
\paragraph{Why is the mechanism by Dughmi et al.~\cite{dughmi2017bernoulli} unsuitable?}
 The reason turns out to be subtle. As the U-replica-surrogate matching contains negative edges and the algorithm by Dughmi et al.~\cite{dughmi2017bernoulli} always returns a perfect matching, some agent types may receive negative utilities from the matching. To guarantee individually rationality, the mechanism must compensate these types. However, due the incentive compatibility constraint, the mechanism must also compensate other agent types. One might think that the overall compensation can be shown to be negligible. Unfortunately, in the following example, we show that the overall compensation can in fact dramatically damage the revenue and may even drive the revenue to $0$.


\begin{example}
\label{example:perfect_revenue}
Consider the following instance with a single agent and outcome space $\cO=\{\perp,o\}$. The agent has two possible types $H$ and $L$ with probability $1-\sigma$ and $\sigma$ respectively, where $\sigma\in (0,1)$ is sufficiently small. The agent's valuation is: $v(H,o)=1$, $v(H,\perp)=v(L,\perp)=v(L,o)=0$. The given mechanism $\MM$ chooses outcome $o$ and charges $1$ if the agent reports $H$, and chooses outcome $\perp$ and {gives the agent $\varepsilon$ if the agent reports $L$. Clearly, $\MM$ is $\varepsilon$-BIC and IR. $\rev(\MM)=1-\sigma-\sigma\varepsilon$}. {Note that in the U-Replica-Surrogate matching, the edge between replica with type L and surrogate with type H has negative weight $-1$.

Denote $\ell$ the number of surrogates sampled from the above distribution. 
Let $\MM'$ be the constructed mechanism that always selects a perfect replica-surrogate matching.
Denote $p(\cdot)$ the payment function of $\MM'$. 
Then $p(L)\leq 0$ since $\MM'$ is IR. Consider the following two scenarios when the agent has true type $H$. In the first scenario she reports truthfully her type $H$ and in the second scenario she reports $L$. With probability $\beta=(1-\sigma)^{\ell}$, none of the surrogates has type $L$. In both scenarios the buyer must be matched to a surrogate with type $H$ in the perfect replica-surrogate matching and has value 1 for the result outcome $o$. Thus the difference of the agent's expected value between the two scenarios is at most $1-\beta$. Since $\MM'$ is BIC, we must have $p(H)\leq p(L)+1-\beta\leq 1-\beta$. Therefore $\rev(\MM')\leq (1-\sigma)\cdot (1-\beta)$. For any fixed $\varepsilon$ and $\ell$, when $\sigma\to 0$, the $\beta\to 1$ and revenue loss goes to 1. Note that our mechanism guarantees revenue loss at most $c\sqrt{\varepsilon}$ for any type distribution, where $c$ is an absolute constant.

}


\end{example}

 
 {
 
 Let us take a closer look at what happens in the  example above. With high probability, the agent with low type $L$ is matched to a surrogate with high type $H$ in the perfect matching and has negative utility. Thus to satisfy individual rationality, the mechanism has to compensate her. To guarantee incentive compatibility, the mechanism will also have to compensate the agent when her true type is $H$. However, the total compensation is so large that it essentially drives the revenue of the constructed mechanism to $0$.  Our main challenge is to enforce both individually rationality and incentive compatibility in the presence of negative edges without sacrificing much of the revenue.
 
 }

 
 

As {Example~\ref{example:perfect_revenue} implies, to preserve revenue, it is crucial to avoid matching a replica with a negative weight edge with high probability in the U-replica-surrogate matching. Again the exact weight can not be computed efficiently. }One may try to remove the negative edges using samples. However, removing edges based on the empirical means from samples could easily violate the maximal-in-range property. {See Example~\ref{example:sample edges}.} 



\begin{example}\label{example:sample edges}

Let $N$ be the number of samples that the algorithm uses to calculate the empirical expectation. Choose $\sigma>0$ such that $\frac{\sigma}{1-2\sigma}<\frac{1}{N}$. Consider the following example with 1 node on each side. There are two instances. For the first instance, the random variable $\FF^{(1)}$ attached to this edge is $1$ w.p. $2\sigma$, and $-\frac{\sigma}{1-2\sigma}$ with probability $1-2\sigma$. The edge weight $\omega^{(1)}=\sigma$. For the second instance, $\FF^{(2)}$ is $\sigma$ w.p. 1 and $\omega^{(2)}=\sigma$.

For the above example, both instances have the same edge weight and any maximal-in-range allocation will always output the same matching. However in the first instance, with probability $(1-2\sigma)^N<1$ the empirical expectation is negative and the two nodes are not matched. While in the second instance the algorithm will always match the two nodes. Thus the output matching is not maximal-in-range. It is well-known that if the allocation is maximal-in-range, there must exist a payment rule such that the agent is incentive-compatible. Thus the algorithm will violate the incentive-compatibility when applied to the replica-surrogate matching.

\end{example}

\vspace{-.15in}
\paragraph{Our Solution for the Downward-Closed Outcome Space.} 
Our transformation on downward-closed environments is directly inspired  by~\cite{dughmi2017bernoulli} but differs in several major ways. Our plan is to design an algorithm, for general graphs with arbitrary weights, that satisfies the following two properties:
\begin{enumerate}
    \item The algorithm produces a distribution of matchings whose expected
    weight is close to the maximum weight matching.
    \item For any LHS node, the expected weight of the edge matched to it is not too negative.
\end{enumerate}

Note that for graphs with positive weights, the algorithm by Dughmi et al.~\cite{dughmi2017bernoulli} satisfies both properties. When the edges are negative, the first property is not satisfied by their algorithm in general. Even if we only consider U-replica-surrogate matching, their algorithm still violates the second property. Interestingly, we provide a reduction from the case of arbitrary edge weights to the case with only positive edge weights. Indeed, our reduction can be succinctly summarized by the following formula, if an edge has weight $w_{jk}$, set the new weight by applying the $\delta$-softplus function to $w_{jk}$\footnote{The function $\log\left( \exp\left(x\right)+1\right)$ is known as the soft plus function.}:
$\zeta_\delta(w_{jk})=\delta\cdot \log\left( \exp\left(w_{jk}/\delta\right)+1\right)$,
 where $\delta>0$ is a parameter of our algorithm. Note that for any value of $w_{jk}$, $\zeta_\delta(w_{jk})$ is always nonnegative! Moreover, the maximum entropy regularized matching on weights $(\zeta_\delta(w_{jk}))_{j, k}$ can be shown to be close to the maximum weight matching on $(w_{jk})_{j, k}$, and the second property also holds due to nice features of the algorithm and the softplus function. So it seems that we only need to run the algorithm from~\cite{dughmi2017bernoulli} on the new weights $(\zeta_\delta(w_{jk}))_{j, k}$.  An astute reader may have already realized that being able to run the algorithm on $(w_{jk})_{j, k}$ does not imply that one can run the algorithm on $(\zeta_\delta(w_{jk}))_{j, k}$, as we can only sample from distributions whose means are $(w_{jk})_{j, k}$ but not $(\zeta_\delta(w_{jk}))_{j, k}$. One idea is to construct a Bernoulli factory to simulate a $\zeta_\delta(w_{jk})$-coin using a $w_{jk}$-coin. To the best of our knowledge, no such construction exists. We take a different approach and make use of a crucial property of the algorithm from~\cite{dughmi2017bernoulli}. Namely, if we run their algorithm with the same parameter $\delta$, the algorithm only needs to sample from the softmax function over the weights. More specifically, with weights $(w_{jk})_{j,k}$, it suffices to have the ability to sample an edge $(j,k)$ with  probability exactly $\exp({w_{jk}/\delta})\over \sum_{k'} \exp({ w_{jk'}/\delta})$. Despite the fact that we cannot directly sample from distributions with means $(\zeta_\delta(w_{jk}))_{j, k}$, we can indeed sample edge $(j,k)$ with exactly the right probability,~as 
$${\exp\left(\zeta_\delta(w_{jk})/\delta\right)\over \sum_{k'} \exp\left( \zeta_\delta(w_{jk'})/\delta\right)}={\exp\left(w_{jk}/\delta\right)+1\over \sum_{k'} \left(\exp({ w_{jk'}/\delta})+1\right)},$$
 which can be sampled efficiently using the fast exponential Bernoulli race given only sample access to distributions with means equal to the original edge weights $(w_{jk})_{j,k}$.

Our second contribution is to show that an approximately maximum U-replica-surrogate matching suffices to guarantee only a small loss in revenue. Previous results~\cite{DaskalakisW12, RubinsteinW15, CaiZ17} only prove the statement for the exactly maximum matching. We provide a more delicate analysis that allows us to extend the statement to approximately maximum matchings. Finally, as the agent may receive negative utility from certain surrogates, we sometimes need to subsidize the agent to ensure individual rationality. Due to the second property of our algorithm, we can argue that the total subsidy is small~compared~to~the~revenue. 
We emphasize that the U-replica-surrogate matching found in our mechanism may not be perfect. Thus the assumption $\mathcal{O}$ being downward-closed is necessary, in order to allow the agent whose true type is unmatched in the matching to receive a $\perp$ outcome. See Mechanism~\ref{alg:mout} in Section~\ref{sec:mechanism} for more details.

\subsubsection{Our Result for the General Outcome Space}
Our result for the general outcome space is based on the \emph{regularized replica-surrogate fractional assignment mechanism} by Hartline et al.~\cite{hartline2015bayesian}. They use this mechanism to provide a black-box reduction for welfare maximization, but this mechanism only applies to discrete type space and requires full access to all agents' type distributions. We refer the readers to Section~\ref{sec:appx_rrsf} for details of their mechanism. 

The main barrier for applying their approach to transform an $\varepsilon$-BIC mechanism to a BIC mechanism is that their mechanism does not provide any guarantees on the revenue. More specifically, the prices of their mechanism are determined by a set of optimal dual variables of their problem. Although any set of optimal dual variables can guarantee the mechanism to be incentive compatible and individually rational, which is sufficient for welfare maximization, some choices could result in substantial revenue loss, making them unsuitable to preserve revenue. In fact, Example~\ref{example:negative. payments} illustrates that for some optimal dual variables, the revenue loss due to negative prices can be very high.  Our main contribution 
is to show an efficient algorithm to find a set of optimal dual variables such that the induced prices only cause negligible loss in the revenue. 

\subsection{Application to Multi-dimensional Revenue Maximization}
We next apply the $\varepsilon$-BIC to BIC transformation to obtain our black-box reduction for revenue maximization. We first introduce the  problem formally.

\smallskip\noindent  \framebox{
\begin{minipage}[c]{0.98\textwidth}
\textbf{Multi-Dimensional Revenue Maximization (MRM):} Given as input $n$ type distributions $\mathcal{D}_1,\ldots,\mathcal{D}_n$ and a set of feasible outcomes $\cO$, 
 output a BIC and IR mechanism $\MM$ who chooses outcomes from $\mathcal{O}$ with probability $1$ and whose expected revenue is optimal relative to any other, possibly randomized, BIC, and IR mechanism with respect to $\mathcal{D} = \bigtimes_{i\in[n]} \mathcal{D}_i$.
\end{minipage}}

To state our black-box reduction, we introduce the virtual welfare {optimization} problem.

\smallskip\noindent  \framebox{
\begin{minipage}[c]{0.98\textwidth}
\textbf{Virtual Welfare Optimization (VWO)~{\cite{CaiDW13b}}:} Given as input $n$ functions $C_i(\cdot): \TT_i\mapsto \mathbb{R}$ and a set of feasible outcomes $\mathcal{O}$, output an outcome $o \in \argmax_{x\in \cO} \sum_i\sum_{t_i\in \TT_i}C_i(t_i)\cdot v_i(t_i,x)$. 
$C_i(\cdot)$ is considered as the weight function that depends on the agent's real type. We refer to the sum $\sum_{t_i\in \TT_i}C_i(t_i)\cdot v_i(t_i,x)$ as agent $i$'s virtual value for outcome $x$.
\end{minipage}}

\begin{informaltheorem}
When the outcome space is downward-closed, given sample access to distribution $\bigtimes_{i=1}^n \DD_i$ and oracle access to an $\alpha$-approximation Algorithm~$G$ for VWO,
we can construct an exactly BIC and IR mechanism $\MM = (x,p)$ with respect to $\bigtimes_{i\in[n]} \DD_i$, that has expected revenue
$\alpha\cdot OPT - O\left(n \sqrt{\varepsilon} \right)$, where $OPT$ is the optimal revenue over all BIC and IR mechanisms with respect to $\bigtimes_{i\in[n]} \DD_i$. The running time is $\poly\left(\sum_{i\in[n]}{|\TT_i|},\frac{1}{\varepsilon},b,rt_G\left(\poly\left( \sum_{i\in[n]}{|\TT_i|},\frac{1}{\varepsilon},b \right)\right)\right)$, where $rt_G(\cdot)$ is the running time of $G$, and $b$ is an upper bound on the bit complexity of $v_i(t_i,o)$ for any agent $i$, any type $t_i$, and any outcome $o$.


\end{informaltheorem}

Note that a similar result holds for a general outcome space, given that we have \emph{full access} to distribution $\bigtimes_{i=1}^n \DD_i$. The modified running time and revenue loss can be found in Section~\ref{sec:fractional}.

\subsection{Further Related Work}
Multi-dimensional revenue maximization has recently received lots of attention from computer scientists. Significant progress has been on the computational front~\cite{ChawlaHK07,ChawlaHMS10,Alaei11,CaiD11b,AlaeiFHHM12,CaiDW12a,CaiDW12b,CaiH13,CaiDW13b,AlaeiFHH13,BhalgatGM13,DaskalakisDW15,kothari2019approximation}. On the structural front, a family of simple mechanisms, i.e., variants of sequential posted price and two-part tariff mechanisms, have been shown to achieve constant factor approximations of the optimal revenue in quite general settings~\cite{BabaioffILW14,Yao15, RubinsteinW15, CaiDW16,ChawlaM16, CaiZ17}.  
$\varepsilon$-BIC to BIC transformation {for revenue} has been an instrumental tool in obtaining both the computational and structural results~\cite{DaskalakisW12, RubinsteinW15, CaiZ17, kothari2019approximation}. 

There has also been significant interest in understanding the sample complexity for learning an almost revenue-optimal auction in multi-item settings. Last year, Gonczarowski and Weinberg~\cite{GonczarowskiW18} show that one can learn an almost revenue-optimal $\varepsilon$-BIC mechanism using $\poly(n,m, 1/\varepsilon)$ samples under the item-independence assumption, where $n$ is the number of bidders and $m$ is the number of items. Brustle et al.~\cite{BrustleCD19} generalize the result to settings where the item values are drawn from correlated but structured distributions that can be modeled by either Markov random fields or Bayesian Networks. The mechanism they produce is still $\varepsilon$-BIC. Our transformation can certainly convert these mechanisms from~\cite{GonczarowskiW18, BrustleCD19} into exactly BIC mechanisms, and the transformation requires $\poly\left(\sum_{i\in[n]} |\TT_i|,1/\varepsilon\right)$ many samples. Unfortunately, each $|\TT_i|$ is already exponential in $m$ in their settings. The dependence on $|\TT_i|$ is unavoidable for us, as our goal is to provide a transformation that is applicable to a general mechanism design setting. Nonetheless, the techniques we develop in this paper may be combined with special structure of the distribution to provide more sample-efficient $\varepsilon$-BIC to BIC transformations.

Recently, Gergatsouli et al. \cite{GergatsouliLT19}
proved that for the case where we have a single buyer with an additive valuation over $m$ independent items and the set of outcomes is downward-closed, an exponential (in $m$) query complexity is necessary for any black-box reduction for welfare maximization. It will certainly be interesting to see whether a similar lower bound on the sample complexity exists. 

{
\subsection{Organization of the Paper}
In Section~\ref{sec:prelim}, we provide the notations we use throughout the paper.
In Section~\ref{sec:reduction}, we mention the tools from the literature that are important for our constructions. In Section~\ref{sec:arbitrary_weight_matching}, we provide our new algorithm that solves the entropy regularized matching problem for graphs with arbitrary edge weights {assuming the outcome space is downward-closed}. In particular, we show the arbitrary edge weight case can be reduced to the  nonnegative edge weight case.
In Section~\ref{sec:mechanism}, we describe our $\varepsilon$-BIC to BIC transformation {when the outcome space is downward-closed}. In Section~\ref{sec:apps}, we show how to  use our $\varepsilon$-BIC to BIC transformation to improve the black-box reduction for multi-dimensional revenue maximization. 
In Section~\ref{sec:fractional} we state our result for the general outcome space.
Note that in Sections~\ref{sec:arbitrary_weight_matching} through \ref{sec:apps} we focus on downward-closed outcome spaces, whereas in Section~\ref{sec:fractional} we shift our attention to general outcome spaces.

}

\section{Preliminaries}\label{sec:prelim}

We specify a general mechanism design setting by the tuple $(n, \mathcal{V},  \mathcal{D}, v, \mathcal{O})$. There are $n$ agents participating in the mechanism. Denote $\mathcal{O}$ the set of all possible outcomes. We consider two types of outcome space.
\begin{enumerate}
    \item \textbf{Downward-Closed Outcome Space:} Each $o\in \mathcal{O}$ can be written as a vector $o=(o_1,...,o_n)$ where $o_i$ is the outcome for agent $i$. We also assume that a null outcome $\perp$ is available to each agent $i$. One can think of $\perp$ as the option of not participating in the mechanism. Throughout the paper, when we say the outcome space $\mathcal{O}$ is \emph{downward-closed}, then for every $o=(o_1,...,o_n)\in \mathcal{O}$, any $o'=(o_1',...,o_n')$ with $o_i'=o_i$ or $o_i'=\perp$ for every $i$ is also in $\cO$. An example of the downward-closed outcome space is the combinatorial auction, where the outcome set contains all possible ways to allocate items to agents, and the null outcome represents allocating nothing to the agent. One setting that does not have a downward-closed outcome space is building a public project.
        \item \textbf{General Outcome Space:} $\OO$ is an arbitrary set. This is the most general outcome space, and it can capture settings such as building a public project.
\end{enumerate} 

Each agent $i$ has a type $t_i$ from type space $\mathcal{V}_i$, which is drawn independently from some distribution $\mathcal{D}_i$. We use $\mathcal{T}_i\subseteq \mathcal{V}_i \text{ or } \supp(\DD_i)$ to denote the support of $\mathcal{D}_i$. We use $\DD$ to denote $\bigtimes_{i\in[n]}\DD_i$. In the paper we consider discrete type spaces, we assume that every $|\mathcal{T}_i|\leq T$ for some finite $T$. {Note that our results for the downward-closed outcome space can easily be extended to the continuous case using similar techniques as in~\cite{dughmi2017bernoulli}, while our results for general outcome space requires the type space to be discrete.} For every $t_i\in \VV_i$,  $v_i(t_i,\cdot)$ is a valuation function that maps every outcome to a real number in $[0,1]$. In a downward-closed outcome space, for all agent $i$ and type $t_i$, $v_i(t_i, o)=0$ if $o_i=\perp$. Every agent is risk-neutral and has quasi-linear utility.


For any mechanism $\MM$, denote $\rev(\MM,\DD)=\E_{t\sim \DD}\left[\sum_{i\in[n]}p_i(t)\right]$ the expected revenue of $\MM$. We use $\rev(\MM)$ for short when the agents' distributions and valuation functions are clear. We use the standard definitions of BIC, $\varepsilon$-BIC, IR, and $\varepsilon$-IR:

\vspace{.1in}
Bayesian Incentive Compatible (BIC): 
$$\E_{t_{-i}\sim \DD_{-i}}\left[v_i\left(t_i,x\left(t_i,t_{-i}\right)\right) -p_i\left(t_i,t_{-i}\right)\right]\geq \E_{t_{-i}\sim \DD_{-i}}\left[v_i\left(t_i,x\left(t'_i,t_{-i}\right)\right) -p_i\left(t'_i,t_{-i}\right)\right],~~~~\forall i\in[n], t_i,t'_i\in\TT_i.$$

Individual Rational (IR): $$\E_{t_{-i}\sim \DD_{-i}}\left[v_i\left(t_i,x\left(t_i,t_{-i}\right)\right) -p_i\left(t_i,t_{-i}\right)\right]\geq 0,~~~~\forall i\in[n], t_i\in\TT_i.$$

$\varepsilon$-BIC: 
$$\E_{t_{-i}\sim \DD_{-i}}\left[v_i\left(t_i,x\left(t_i,t_{-i}\right)\right) -p_i\left(t_i,t_{-i}\right)\right]\geq \E_{t_{-i}\sim \DD_{-i}}\left[v_i\left(t_i,x\left(t'_i,t_{-i}\right)\right) -p_i\left(t'_i,t_{-i}\right)\right]-\varepsilon,~~~~\forall i\in[n], t_i,t'_i\in\TT_i.$$

$\varepsilon$-IR:  $$\E_{t_{-i}\sim \DD_{-i}}\left[v_i\left(t_i,x\left(t_i,t_{-i}\right)\right) -p_i\left(t_i,t_{-i}\right)\right]\geq -\varepsilon,~~~~\forall i\in[n], t_i\in\TT_i.$$

\paragraph{Coupling between Type Distributions:}In order to measure the difference between the two distributions, we will introduce the following definition. Fix every agent $i$. A \emph{coupling} $c_i(\cdot,\cdot)$ of distribution $\mathcal{D}'_i$ and $\mathcal{D}_i$ is a joint distribution on the probability space $\mathcal{T}'_i\times \mathcal{T}_i$ such that the marginal of $c_i$ coincide with $\mathcal{D}'_i$ and $\mathcal{D}_i$. In the paper we slightly abuse the notation, denoting $c_i(b)$ a random variable that is distributed according to the conditional distribution of type $t_i$ over $\mathcal{T}_i$ when $t_i'=b$. According to the definition of the coupling, when $t_i'\sim \mathcal{D}_i'$, $c_i(t_i')\sim \mathcal{D}_i$.

We say $v_i$ is \emph{non-increasing} w.r.t. the coupling $c_i$ if for all $t_i\in \TT_i'$, outcome  $o\in\cO$, and every realized type $c_i(t_i)$, $v_i(t_i,o)\geq v_i(c_i(t_i),o)$. Intuitively, the coupling always maps a ``higher'' type to a ``lower'' type. Such coupling is common, for example in a combinatorial auction, rounding agent $i$'s value for each bundle of items down to the closest multiples of $\delta$ can be viewed as such a coupling. 

\paragraph{Wasserstein Distance:} For any $t_i,t_i'\in \mathcal{V}_i$, let $\dist_i(t_i,t_i')=\max_{o\in \cO}|v_i(t_i,o)-v_i(t_i',o)|$. The $\ell_{\infty}$-\emph{Wasserstein Distance} between distribution $\mathcal{D}_i$ and $\mathcal{D}_i'$ w.r.t. $\dist_i$ is defined as the smallest expected distance among all couplings. Formally,

$$\wass(\mathcal{D}_i,\mathcal{D}_i')=\min_{c_i(\cdot,\cdot)}\int \dist_i(t_i,t_i')dc_i(t_i,t_i')$$

Finally, we use $\log (\cdot)$ to denote the natural logarithm and $\Delta^{\ell}$ to denote the set of all distributions~over~$\ell$~elements.

\begin{definition}[Gibbs Distribution]\label{def:Gibbs dist}
	For any integer $\ell$, define the Gibbs distribution $z\in \Delta^\ell$ over $\ell$ states with temperature $\beta$ as $z_i=\frac{\exp(E_i/\beta)}{\sum_{i'\in[\ell]}\exp(E_{i'}/\beta)}$ for all $i\in[\ell]$, where $E_i$ is the energy of element $i$.
\end{definition}

\begin{definition}[Maximal-in-Range Algorithms]
An algorithm is maximal-in-range, if for every $j\in[d\ell]$, there exists a cost function $c(\cdot)$, which may depend on $\omega_{-j}$, such that the allocation $z_j\in \argmax_{z'\in \FF} \sum_{j,k}z_{jk}'\cdot\omega_{jk}-c(z')$ for any $\omega_j$, where $\FF$ is a set of all feasible allocations. 
\end{definition}

\subsection{A Brief Introduction to Bernoulli Factories}\label{appx:bernoulli}

Suppose we are given a coin with bias $\mu$, can we construct another coin with bias $f(\mu)$ using the original coin? If the answer is yes, then how many flips do we need from the original coin to simulate the new coin? A framework that tackles this problem is called Bernoulli Factories.
We refer the reader to \cite{latuszynski2010bernoulli} for a survey on this topic.

\begin{definition}[Keane and O'Brien \cite{keane1994bernoulli}]

Given some function $f: (0,1) \mapsto (0,1)$ and black-box access to independent samples of a Bernoulli random variable with bias $p$, the Bernoulli factory problem is to generate a sample from a Bernoulli distribution with bias $f(p)$.
\end{definition}


A useful generalization of the previous model is the following\footnote{The model is called \emph{Expectations from Samples} in~\cite{dughmi2017bernoulli}.}: given sample access to distributions $\dD_1,\dD_2,\ldots,\dD_m$ with expectations {$\mu_1,\mu_2,\ldots,\mu_m\in (0,1)$},
and a function $f:(0,1)^m \rightarrow \Delta(X)$,
where $X$ is a set of feasible outcomes and $\Delta(X)$ is a set of probability distributions over these outcomes, how can we want generate a sample from $f(\mu_1,\ldots,\mu_m)$?


\notshow{
Below we state an important result from \cite{DughmiHKN17}, which is used in our reduction. 
The basic version of the algorithm (Algorithm~2 \cite{DughmiHKN17}) requires $\lambda n e^{\lambda(1-\mu_{\max})}$ samples on expectation to produce the result.
With some Bernoulli factory tricks (Algorithm~3 \cite{DughmiHKN17}), we can reduce the expected number of samples to $\lambda^4 n^2 \log(\lambda n)$.

\begin{theorem}[Theorem~3.2 ,Theorem~3.4 \cite{DughmiHKN17}]
\label{thm:exp_bern_race}
 Given a parameter $\lambda>0$ and sample access to $n$ Bernoulli distributions with bias $p_1,p_2,\ldots,p_n$,
 we can sample from the \emph{Exponential Weights} distribution, i.e. the distribution $
\dD$ where $\Pr_{I \sim \dD}[I = i] = \frac{e^{\lambda p_i}}{\sum_{i}e^{\lambda p_j}}$, using  $\lambda m e^{\lambda(1-p_{\max})}$ samples on expectation. For $\lambda > 4$ this bound is improved to $\lambda^4 n^2 \log(\lambda n)$ samples on expectation.
\end{theorem}
}

Below we state an important result from \cite{dughmi2017bernoulli}, which we use in this paper. It proposes an algorithm called \emph{Fast Exponential Bernoulli Race} with $X=[m]$. For every $\lambda>0$, it produce a sample from the Gibbs distribution with temperature $\frac{1}{\lambda}$ and energy $\mu_i$ for each outcome $i$, given only sample access to distributions $\dD_1,\dD_2,\ldots,\dD_m$. 

\begin{theorem}\cite{dughmi2017bernoulli}
\label{thm:exp_bern_race}
Given any parameter $\lambda>0$ and sample access to distributions $\FF_1,\FF_2,\ldots,\FF_m$ with expectations $\mu_1,\mu_2,...\mu_m\in (0,1)$, there exists an algorithm that can sample from a Gibbs distribution in $\Delta^m$, where $$z_i={\exp(\lambda \mu_i)\over \sum_{j\in [m]}\exp (\lambda \mu_{j})},$$  using $O(\lambda^4 m^2 \log(\lambda m))$ samples in expectation.
\end{theorem}




\notshow{
\begin{corollary}
\label{lem:sample-exp-weight} 
\mingfeinote{Fix $j$. }For any fixed $\Balpha \in [0,\mathcal{R}']^\ell$ and weights $\omega_{jk} \in [-\mathcal{R}, \mathcal{R}]$, a sample from the distribution 

$$z_i={\exp(\lambda \mu_i)\over \sum_{j\in [m]}\exp (\lambda \mu_{j})},$$

\mingfeinote{$\hat{z}_j$ in Lemma~\ref{obs:exp weight every replica}} can be drawn with $\left(\frac{2\mathcal{R} + \mathcal{R}'}{\delta}\right)^4 \ell^2 \log\left(\ell\cdot\frac{2\mathcal{R}+\mathcal{R}'}{\delta} \right)$ samples in expectation.
\end{corollary}

\begin{proof}
We initially consider the exponential weights distribution with weights:

\begin{eqnarray*}
\Omega_{jk} =
\frac{\omega_{jk}-\alpha_k + \mathcal{R}' + \mathcal{R}}{\mathcal{R}' + 2\mathcal{R}}, &j\in[\ell] \mingfeinote{k\in [\ell]?}
\end{eqnarray*}

We note that since $-\mathcal{R} \leq \omega_{jk} \leq \mathcal{R}$ and $\bm{\alpha}_k \leq \mathcal{R}'$,
then:

\begin{eqnarray*}
    0 \leq \frac{\omega_{jk}-\alpha_k + \mathcal{R}' + \mathcal{R}}{\mathcal{R}' + 2\mathcal{R}} \leq 1
\end{eqnarray*}

We can use Theorem~\ref{thm:exp_bern_race} with $\lambda=\frac{\mathcal{R}'+2\mathcal{R}}{\delta}$ and get exponential weights distribution with weights:

\begin{eqnarray*}
\Omega'_{jk} =
\frac{\omega_{jk}-\alpha_k + \mathcal{R}' + \mathcal{R}}{\delta}, &j\in[\ell] \mingfeinote{k\in [\ell]?}
\end{eqnarray*}

using $\left(\frac{2\mathcal{R} + \mathcal{R}'}{\delta}\right)^4 \ell^2 \log\left(\ell \frac{2\mathcal{R}+\mathcal{R}'}{\delta} \right)$ samples on expectation.\footnote{if $0 <\lambda \leq 4$ we need $O(m)$ }

\mingfeinote{The lemma is proved by noticing that $\Omega_{jk}=\frac{\mathcal{R}'+2\mathcal{R}}{\delta}\cdot \Omega_{jk}'$ for $k$, and thus the above two distributions are the same.}
\end{proof}
}

The \textit{fast exponential Bernoulli race} \cite{dughmi2017bernoulli} is a randomized algorithm that allows us to sample from the Gibbs distribution. We use the following result in the rest of our paper.

\begin{lemma}\label{lem:sample-exp-weight}[Fast Exponential Bernoulli Race]
For any integer $m$, any $\delta>0$, and any $(\alpha_k)_{k\in [m]} \in [0,h]^m$, given sample access to distributions $\FF_1,\ldots, \FF_m$ with expectations $w_1,\ldots, w_m \in [-1,1]$,  a sample from the following Gibbs distribution in $\Delta^m$:

 $$z_k={\exp((w_k-\alpha_k)/\delta)\over \sum_{j\in [m]}\exp ((w_j-\alpha_j)/\delta)},$$
 can be drawn with $\left(\frac{4 + h}{\delta}\right)^4 m^2 \log\left(\frac{(4+h)m}{\delta} \right)$ samples from $(\FF_k)_{k\in[m]}$ in expectation.
\end{lemma}

\begin{prevproof}{Lemma}{lem:sample-exp-weight}
First, notice that $(z_k)_{k\in[m]}$ can also be represented as the Gibbs distribution with temperature {$\frac{\delta}{h+4}$} and energy $\Omega_{k} =
\frac{\omega_{k}-\alpha_k + h + 2}{h + 4}, k\in [m].$
Note that since $-1 \leq w_{k} \leq 1$ and ${\alpha}_k \leq h$,
then $0 <\frac{w_k-\alpha_k + h + 2}{h + 4} <1$. 
Thus, by Theorem~\ref{thm:exp_bern_race} with $\lambda=\frac{h+4}{\delta}$,  we can generate a sample according to $(z_k)_{k\in[m]}$ with $\left(\frac{4 + h}{\delta}\right)^4 m^2\log\left( \frac{(4 + h)m}{\delta} \right)$ samples in expectation.
%
%
\end{prevproof}





\section{Tools from the Literature}\label{sec:reduction}

\subsection{Replica-Surrogate Matching}

Now we provide a detailed description of the replica-surrogate matching mechanism used in~\cite{DaskalakisW12, RubinsteinW15, CaiZ17}. For each agent $i$, the mechanism generates a number of replicas and surrogates from $\DD_i$, and maps the agent's type $t_i$ to one of the surrogates via a maximum weight replica-surrogate matching, and charges the agent the corresponding VCG payment. Then let the matched surrogate participate in the mechanism for the agent. Formally, suppose we are given query access to a mechanism $\MM=(x,p)$, we construct a new mechanism $\MM'$ using the following two-phase procedure:

\paragraph{Phase 1: Surrogate Selection} For each agent $i$,
\begin{enumerate}
    \item Given her reported type $t_i\in \mathcal{D}_i$, create $\ell-1$ replicas sampled i.i.d. from $\mathcal{D}_i$ and $\ell$ surrogates sampled i.i.d. from $\mathcal{D}_i$. The value of $\ell$ is specified in Lemma~\ref{cor:revenue of basic M_0}.
    \item Construct a weighted bipartite graph between replicas (and agent $i$'s true type $t_i$) and surrogates. The weight between the $j$-th replica $\rj$ and the $k$-th surrogate $\sk$ is the interim value of agent $i$ when her true type is $\rj$ but reported $\sk$ to $\MM$ less the interim payment for reporting $\sk$ multiplied by $(1-\eta)$:
    \begin{equation}\label{equ:def W}
    W_i(\rj,\sk)= \E_{t_{-i}\sim \DD_{-i}}\left[v_i(\rj,x(\sk, t_{-i}))\right]-(1-\eta)\cdot \E_{t_{-i}\sim \DD_{-i}}\left[p_i(\sk,t_{-i})\right].
    \end{equation}
    \item Treat $W_i(\rj,\sk)$ as the value of replica $\rj$ for being matched to surrogate $\sk$. Run the VCG mechanism among the replicas, that is, compute the maximum weight matching w.r.t. edge weight $W_i(\cdot,\cdot)$ and the corresponding VCG payments.  If a replica (or type $t_i$) is unmatched in the maximum matching, match it to a random unmatched surrogate.
\end{enumerate}

\paragraph{Phase 2: Surrogate Competition}
Let $s_i$ be the surrogate matched with the agent $i$'s true type $t_i$. Run mechanism $\MM$ under input $s=(s_1,\ldots, s_n)$. Let $o=(o_1,\ldots, o_n)$ be a the outcome generated by $x(s)$. If agent $i$ is matched in the maximum matching, her outcome is $o_i$ and her expected payment is 
$(1-\eta)\cdot p_i(s)$ plus the VCG payment for winning surrogate $s_i$ in the first phase; Otherwise the agent gets the null outcome $\perp$ and pays $0$.

In Figure~\ref{fig:replica-surrogate} we illustrate the replica-surrogate scheme.

\begin{figure}[h!]
  \centering
  \includegraphics[width=0.5\textwidth]{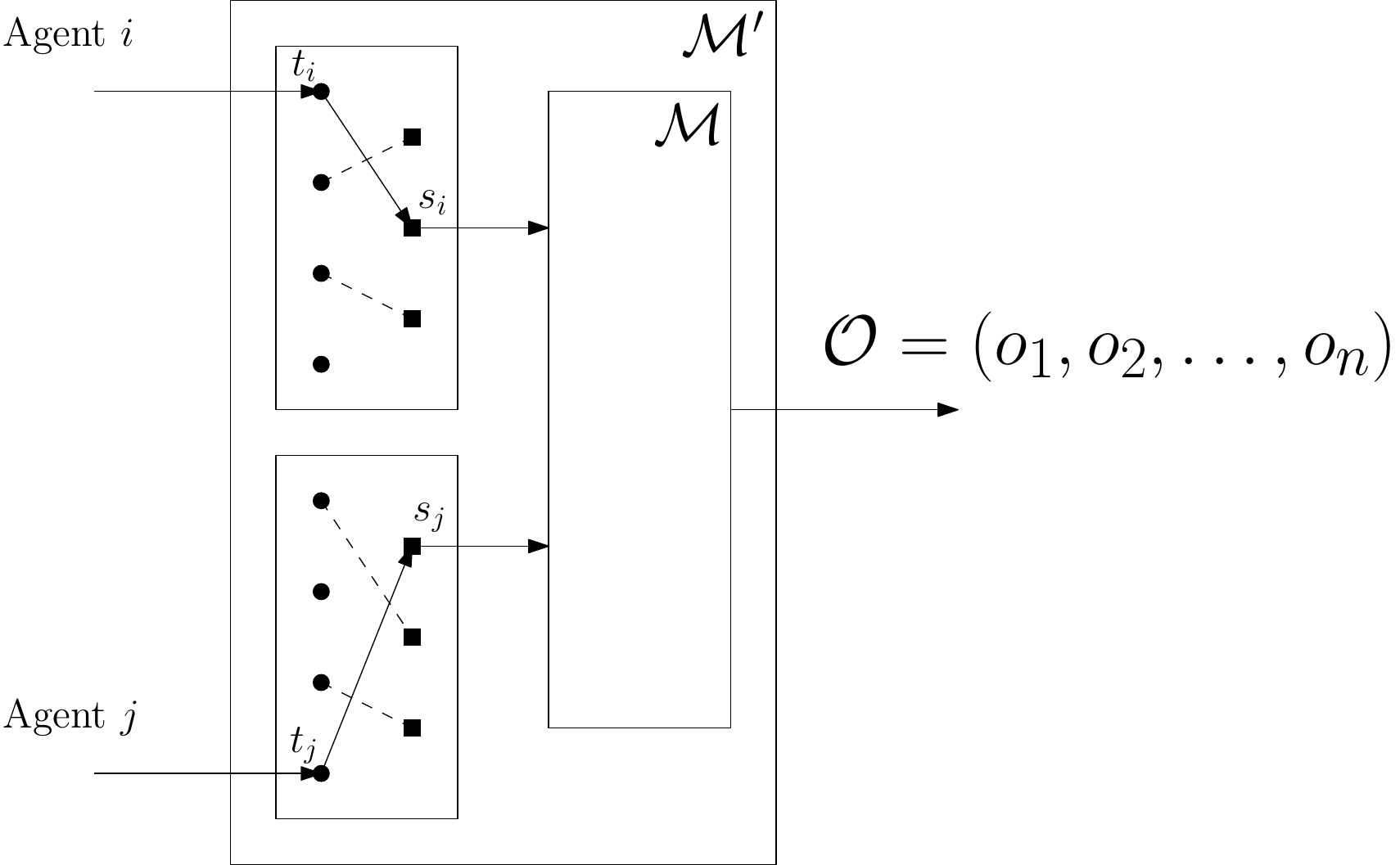}
  \caption{With $\bullet$ we denote the replicas and with $\blacksquare$ the surrogates}
  \label{fig:replica-surrogate}
\end{figure}

\begin{lemma}\cite{HartlineKM11,BeiH11,DaskalakisW12,RubinsteinW15, CaiZ17}
\label{lemma:IDEAL_BIC_IR}
$\MM'$ is BIC and IR.
\end{lemma}


\begin{proof}
We prove this in two parts,
similarly to \cite{CaiZ17}.
First we argue that the distribution of the surrogate $s_i$ that represents the agent, when the agent reports truthfully,
is $\dD_i$.
Since we have a perfect matching,
an equivalent way of thinking about the process is to draw $\ell$ replicas, produce the perfect matching (the VCG matching plus the uniform matching between the unmatched replicas and surrogates) and then pick one replica uniformly at random to be the agent. These two processes produce the same joint distribution between replicas, surrogates and the agents $i$. So we can just argue about the second process of sampling. Since the agent is chosen uniformly at random between the replicas in the second process, the surrogate $s_i$ that represents the agent, will also be chosen uniformly at random between all the surrogates. Thus, the distribution of $s_i$ is $\dD_i$.

We need to argue that for every agent $i$ reporting truthfully is a best response, if every other agent is truthful. In the VCG mechanism, agent $i$ faces a competition with the replicas to win a surrogate.  If agent $i$ has type $t_i$, then her value for winning a surrogate with type $s_i$ in the VCG mechanism is exactly the edge weight $$W_i(t_i,s_i)= \E_{t_{-i}\sim \DD_{-i}}\left[v_i(t_i,x(s_i, t_{-i}))\right]-(1-\eta)\cdot \E_{t_{-i}\sim \DD_{-i}}\left[p_i(s_i,t_{-i})\right].$$ Clearly, if agent $i$ reports truthfully, the weights on all incident edges between her and all the surrogates will be exactly her value for winning those surrogates. Since agent $i$ is in a VCG mechanism to compete for a surrogate, reporting the true edge weights is a dominant strategy for her, therefore reporting truthfully is also a best response for her assuming the other agents are truthful. It is critical that the other agents are reporting truthfully, otherwise agent $i$'s value for winning a surrogate with type $s_i$ may be different from the weight on the corresponding edge.

\end{proof}

Moreover, when $\ell$ is sufficiently large, the revenue of $\MM'$ is close to the revenue of $\MM$.
\notshow{

\begin{corollary} \label{cor:revenue of basic M_0}
If $\MM$ is an $\varepsilon$-BIC and IR mechanism w.r.t. $\DD$, then for any $\eta>0$, $\rev(\MM',\DD)\geq(1-\eta)\rev(M)-n\mathcal{R}\cdot\sqrt{\frac{T}{\ell}}-n\varepsilon/\eta$.
\end{corollary}
}

\begin{lemma} \cite{DaskalakisW12, RubinsteinW15, CaiZ17}\label{cor:revenue of basic M_0}
If $\MM$ is an $\varepsilon$-BIC and IR mechanism w.r.t. $\DD$, then for any $\eta\in(0,1)$ and any $\ell>\frac{T}{\varepsilon^2}$, $\rev(\MM',\DD)\geq(1-\eta)\rev(\MM,\DD)-\Theta(n\varepsilon)/\eta$.\end{lemma}

{Lemma~\ref{cor:revenue of basic M_0} also} follows from a special case of Lemma~\ref{lem:payment K to V} when $\Delta=0$ and $d=1$. The main takeaway of this {Lemma} is that the above mechanism $\MM'$ indeed satisfies the requirement of an $\varepsilon$-BIC to BIC transformation. However, as we discussed in Section~\ref{sec:results and techniques}, the mechanism runs in exponential time.

\subsection{Online Entropy Regularized Matching}\label{subsec:online matching}



Now we describe the \emph{online entropy regularized matching algorithm} developed by Dughmi et al.~\cite{dughmi2017bernoulli}. The original application is to find approximately maximum replica-surrogate matching in welfare maximization, but
the algorithm is general and can be applied to any \emph{$d$-to-$1$ bipartite matching with positive edge~weights}.

\vspace{-.15in}
\paragraph{$d$-to-$1$ Matching} For every integer $\ell$, $d$, consider the complete bipartite graph between $d\ell$ left hand side nodes (called LHS-nodes) and $\ell$ right hand side nodes (called RHS-nodes). Let $\omega_{jk}$ be the edge weight between LHS-node $j$ and RHS-node $k$ for $j\in [d\ell], k\in [\ell]$. For ease of notation, let $\omega_j=(\omega_{jk})_{k\in[\ell]}$, $\omega=(\omega_j)_{j\in [d\ell]}$, and $\omega_{-j}=(\omega_{j'})_{j'\not=j}$. A matching is called a $d$-to-$1$ matching if every LHS-node is matched to at most one RHS-node, and every RHS-node is matched to at most $d$ LHS-nodes. A $d$-to-$1$ matching is called \emph{perfect} if every LHS-node is matched to one RHS-node, and every RHS-node is matched to~exactly~$d$~LHS-nodes.

In this section, we focus on the case where all edge weights $\omega$ are nonnegative, and we refer to this case as the \emph{nonnegative weight $d$-to-$1$  matching}. In Section~\ref{sec:arbitrary_weight_matching}, we generalize the results to arbitrary weights.


\notshow{

In the program we use the constraint $\sum_{k}z_{jk}=1$.:
\vspace{-.1in}
\begin{equation}\label{LP:max matching}
\begin{split}
\max \displaystyle&\sum_{j,k} z_{jk}\cdot\omega_{jk} \\
\text{subject to} \displaystyle &\sum_j z_{jk}\leq d,~~~~~~\forall k\in [\ell]\\
                                        & \sum_k z_{jk}=1,~~~~~~\forall j\in [d\ell]\\
                                        &  z_{jk}\in [0,1], ~~~~~~\forall j\in [d\ell], \forall k\in [\ell].\end{split}
\end{equation}

\vspace{-.05in}

}
The optimal $d$-to-$1$ matching is simply a maximum weight bipartite matching problem. The challenge is that the weights are not given. For every edge $(j,k)$, we only have sample access to {a distribution $\FF_{jk}$} whose expectation is $\omega_{jk}$. To the best of our knowledge, none of the algorithms for finding a maximum weight bipartite matching can be implemented exactly with such sample access to the edge weights. Moreover, as we require the replica-surrogate matching mechanism to be incentive compatible, the algorithm should be maximal-in-range. Therefore, finding the maximum weight matching using the empirical means is also not an option, as it violates the maximal-in-range property (see the discussion in Section~\ref{sec:results and techniques}). 

Dughmi et al.~\cite{dughmi2017bernoulli} provide a polynomial time maximal-in-range algorithm (Algorithm~\ref{alg:online matching}) to compute an approximately maximum weight perfect $d$-to-$1$ matching. 
The key idea is to find a ``\emph{soft maximum weight matching}'' instead of the maximum weight matching by adding an entropy function as a regularizer to the total weight. We summarize the guarantees of Algorithm~\ref{alg:online matching} in Theorem~\ref{thm:dughmi result}. We refer the readers to~\cite{dughmi2017bernoulli} for intuition behind Algorithm~\ref{alg:online matching}. However, to understand this paper, readers can simply treat Theorem~\ref{thm:dughmi result}  as a black box that guarantees that Algorithm~\ref{alg:online matching} is maximal-in-range, and finds approximately maximum expected weight $d$-to-1 matching, with only sample access to the distributions.

\begin{definition}
Given parameter $\delta>0$, the (offline) entropy regularized matching program (P) is:
\vspace{-.05in}
\begin{equation}\label{MP:regularized matching}
\begin{split}
\max &~~\textstyle\sum_{j,k} z_{jk}\cdot\omega_{jk}-\delta\sum_{j,k}z_{jk}\log(z_{jk}) \\
\text{subject to} &~~\textstyle\sum_j z_{jk}\leq d,~~~~~~\forall k\in [\ell]\\
                                         &~~\textstyle\sum_k z_{jk}=1,~~~~~~\forall j\in [d\ell] \\
                                          &~~z_{jk}\in [0,1], ~~~~~~\forall j\in [d\ell], \forall k\in [\ell].\end{split}
\end{equation}
\end{definition}

\vspace{-.05in}
Lagrangify the constraints $\sum_j z_{jk}\leq d, \forall  k\in[\ell] $. The Lagrangian dual of $(P)$ is:
$$\textstyle L(\Bz,\Balpha)=\sum_{j,k}z_{jk}\omega_{jk}-\delta\sum_{j,k}z_{jk}\log(z_{jk})-\sum_k\alpha_k(d-\sum_j z_{jk}).$$

The following lemma follows from the first-order condition: for any dual variables $\Balpha$, the optimal solution for the Lagrangian is given by a collection of Gibbs distribution $z^*=(z^*_j)_{j\in[d\ell]}$. 

\begin{lemma}\label{obs:exp weight}\cite{dughmi2017bernoulli}
For every dual variables $\Balpha\in[0,h]^\ell$, the optimal solution $z^*$ maximizing the Lagrangian $L(z,\Balpha)$ subject to constraints  $\sum_k z^*_{jk}=1,\forall j\in [d\ell]$ is:
$z_{jk}^*=\frac{\exp\left(\frac{\omega_{jk}-\alpha_k}{\delta}\right)}{\sum_{k'\in[\ell]}\exp\left(\frac{\omega_{jk'}-\alpha_{k'}}{\delta}\right)},~\forall j\in[d\ell],k\in[\ell]$. 

\noindent If for every edge $(j,k)$, we are given sample access to a distribution $\FF_{jk}$ whose mean is $\omega_{jk}\in [0,1]$, we can use the fast exponential Bernoulli race~\cite{dughmi2017bernoulli} to  sample from the Gibbs distribution $z^*_j$ for all $j\in [d \ell]$. In particular, each sample from distribution $z^*_j=(z^*_{j1},\ldots, z^*_{j\ell})$ only requires in expectation $poly(h,\ell,1/\delta)$ many samples from $(\FF_{jk})_{k}$ (Lemma~\ref{lem:sample-exp-weight}).

\end{lemma}


\noindent If the optimal dual variables $\Balpha^*$ are known, by complementary slackness, the corresponding $z^*$ in Lemma~\ref{obs:exp weight} is the optimal solution of $(P)$.  The gap between the expected weight of $z^*$ and the maximum weight is at most the value of the maximum entropy $\delta\cdot d\ell\log\ell$, so we can simply use the matching sampled according to the distribution $z^*$. However, as the optimal dual is unknown, the wrong dual variables $\Balpha$ may cause a loss of $\sum_k\alpha_k(d-\sum_j z_{jk})$, which may be too large when $z$ is not computed based on the optimal dual variables. To resolve this difficulty, Dughmi et al.~\cite{dughmi2017bernoulli} introduce the second key idea -- \emph{Online Entropy Regularized Matching algorithm} (Algorithm~\ref{alg:online matching}). The online algorithm gradually learns a set of dual variables close to the optimum $\Balpha^*$. When the algorithm terminates, it is guaranteed to find a close to optimal solution to  program $(P)$.
From Lemma~\ref{obs:exp weight every replica}, the algorithm is also maximal-in-range for any choice of the parameters $\delta,\eta', \gamma$.

\begin{lemma}\label{obs:exp weight every replica}\cite{dughmi2017bernoulli}
For every $j$, $\alpha^{(j)}$ and parameter $\gamma$, the Gibbs distribution $\hz_j$ (specified in step 4) is maximal-in-range, 
as $$\hz_j\in \argmax_{z'\in \Delta^{|K|}} \sum_{k\in K}z'_{jk}\omega_{jk}-\delta\sum_{k\in K}z'_{jk}\log(z'_{jk})-\sum_{k\in K}\gamma\alpha_k^{(j)}\cdot z'_{jk}~\footnote{Notice that $\alpha^{(j)}$ only depends on the weights incident to the LHS-nodes $1$ to $j-1$.}.$$
 
\end{lemma}

\begin{algorithm}[h]
\begin{algorithmic}[1]
\REQUIRE Sample access to the distribution $\FF_{jk}$ whose expectation is $\omega_{jk}$, for every $j\in[d\ell],k\in[\ell]$.
    

    \FOR{$j \in [d\ell]$}
        \STATE Let $d_k^{(j-1)}$ be the number of LHS-nodes matched to RHS-node $k$ in the current matching and $K=\{k:d_k^{(j-1)}<d\}$. 
        \STATE Set $\Balpha^{(j)}$ according to the Gibbs distribution with energy $d_k^{(j-1)}$ for RHS-node $k\in K$ and temperature $1/\eta'$, and $\alpha_k^{(j)}=0$ for all $k\not\in K$.
        \STATE Match LHS-node $j$ to a RHS-node $k\in K$ according to the Gibbs distribution $\hz_j$ over RHS-nodes in $K$, where the temperature is $\delta$ and the energy for matching to a RHS-node $k\in K$ is $(\omega_{jk}-\gamma\alpha_k^{(j)})$. A sample from $\hz_j$ can be generated via the fast exponential Bernoulli race 
        with {$ \poly(\gamma,\ell, 1/\delta)$} sample from $(\FF_{jk})_{k}$ in expectation (See Lemma~\ref{lem:sample-exp-weight} for details). 
    \ENDFOR
\end{algorithmic}
\caption{{\sf Online Entropy Regularized Matching with Non-negative Edge Weights (with parameters $\delta, \eta',\gamma$)}}
\label{alg:online matching}

\end{algorithm}
\vspace{-.15in}


\notshow{

\begin{lemma}\label{lem:informal-sample-exp-weight}\cite{dughmi2017bernoulli} 
For any fixed $\Balpha \in [0,\mathcal{R}]^\ell$ and weights $\omega_{jk} \in [0, \mathcal{R}]$, a sample from the distribution $z^*$ in Lemma~\ref{obs:exp weight} can be drawn with $poly(\mathcal{R},\ell,1/\delta)$ 
samples in expectation.

\end{lemma}
}



\begin{theorem}\cite{dughmi2017bernoulli}\label{thm:dughmi result}
When $\omega_{jk}\in [0,1]$~\footnote{The theorem applies to any bounded edge weights $\omega_{jk}\in [0,\mathcal{R}]$. 
 For simplicity we normalize the edge weights to lie between~$[0,1]$.} for all $j,k$, Algorithm~\ref{alg:online matching} satisfies the following properties:
\begin{enumerate}[itemsep=1pt]
    \item For any choice of the parameters, it always returns a \textbf{perfect} $d$-to-1 matching.
    \item For any choice of the parameters, the algorithm is \textbf{maximal-in-range}. The expected running time and sample complexity of Algorithm~\ref{alg:online matching} is $\poly(d,\ell,\gamma,1/\delta)$.      
    \item  For every $\delta,\eta'>0$, if $d\geq \ell\log\ell/\eta'^2$ and $\gamma\in \left[{\opt(P)\over d}, {O(1)\cdot\opt(P)\over d}\right]$, where $\opt(P)$ is the optimum of program $(P)$, the expected value (over the randomness of the Algorithm~\ref{alg:online matching}) of $\sum_{j\in[d\ell], k\in[\ell]} \hz_{jk}\omega_{jk}-\delta\sum_{j,k}\hz_{jk}\log(\hz_{jk})$ is at least $(1-O(\eta'))\cdot \opt(P)$. 
    
    Moreover, for every $\psi\in (0,1)$, if we set $\delta=\Theta(\frac{\psi}{\log\ell}), \eta'=\Theta(\psi)$, and $d$ and $\gamma$ satisfy the conditions above, then the expected total weight of the matching output by the algorithm is at most  $O(d\ell\psi)$ less than the maximum weight matching. 
\end{enumerate}
\end{theorem}

{The only part of Algorithm~\ref{alg:online matching} does not specified  is how to choose a $\gamma$ that is a constant factor approximation to $\opt(P)\over d$. Dughmi et al.~\cite{dughmi2017bernoulli} show a polynomial time randomized algorithm that produces a $\gamma$ that falls into $\left[{\opt(P)\over d}, {O(1)\cdot\opt(P)\over d}\right]$ with high probability, which suffices to find a close to optimum V-replica-surrogate matching. Please see Appendix~\ref{appx:offline optimal approximation} for details.} 

\section{$d$-to-$1$ Matching with Arbitrary Edge Weights}\label{sec:arbitrary_weight_matching}

To obtain an {approximately revenue-preserving} $\varepsilon$-BIC to BIC transformation, we need to find a near-optimal U-surrogate-replica matching, where edge weights may be negative. Motivated by this application, we provide a generalization of Theorem~\ref{thm:dughmi result} to general  $d$-to-$1$ matchings with arbitrary edge weights. 
We design a new algorithm (Algorithm~\ref{alg:new online matching}) with guarantees summarized in Theorem~\ref{thm:arbitrary-weight-algorithm}.


In Example~\ref{example:perfect_revenue} we point out the issue of directly applying Algorithm~\ref{alg:online matching} to the general $d$-to-$1$ matching problem. A tempting way to fix the issue may be to first remove all edges with negative weights then run Algorithm~\ref{alg:online matching}. With only sample access to $\FF_{jk}$, one way to achieve this is to remove edges with negative empirical means. In fact, with a sufficiently large number of samples, with high probability, all edges with strictly positive weights will remain and all edges with strictly negative weights will be removed. 
 However, with non-zero probability,  some edges will either be kept or removed incorrectly causing the algorithm to violate the maximal-in-range property. See Example~\ref{example:sample edges} for a concrete construction. 

An alternative way is to relax the constraint $\sum_k z_{jk}=1$ to $\sum_k z_{jk}\leq1$, so the algorithm no longer needs to find a perfect matching. However, Lemma~\ref{obs:exp weight} fails to hold as the optimal solution is no longer a Gibbs distribution and it is unclear how to sample efficiently from it with only sample access to $\FF_{jk}$.\footnote{The issue is that $\sum_k z^*_{jk}$ may be strictly less than $1$ and has a complex expression. It is not clear whether we can sample efficiently from $z^*_j$ with only sample access to $(\FF_{jk})_{jk}$. Moreover, even if we can sample from the distribution, the guarantees in Theorem~\ref{thm:dughmi result} may no longer hold.} A similar attempt is to add a slack variable $y$ to $(P)$, modifying the constraint $\sum_k z_{jk}=1$ to $\sum_k z_{jk}+y=1$. It is equivalent to adding one dummy RHS-node, with weight $0$ on every incident edge. Now for every dual variable, the optimal solution for the Lagrangian follows from a Gibbs distribution. However, the program differs from $(P)$, in particular the new dummy RHS-node has no capacity constraint, and as a result there is no dual variable that corresponds to this dummy node. It is not clear how to modify Algorithm~\ref{alg:online matching} to accommodate the new dummy node and to produce a close to maximum matching.


\subsection{Reduction from Arbitrary Weights to Non-Negative Weights}\label{subsec:reduction-arbi-weights}

In this section, we provide a reduction from the $d$-to-$1$ matching with arbitrary edge weight case to the non-negative  edge weight case. 

\begin{definition}
For arbitrary edge weights $(\omega_{jk})_{jk}$ and parameter $\delta>0$, define the $\delta$-softplus function:
$$\zeta_\delta(\omega_{jk})=\delta\cdot \log(\exp(\omega_{jk}/\delta)+1)$$
Consider the entropy regularized matching program $(P')$ w.r.t. weights $(\zeta_\delta(\omega_{jk}))_{jk}$:
\begin{equation}\label{MP:regularized zeta matching}
\begin{split}
\max &~~\textstyle G(z)=\sum_{j,k} z_{jk}\cdot\zeta_\delta(\omega_{jk})-\delta\cdot\textstyle\sum_{j,k}z_{jk}\log(z_{jk}) \\
\text{subject to} &~~\textstyle\sum_j z_{jk}\leq d,~~~~~~\forall k\in [\ell]\\
                                         &~~\textstyle\sum_k z_{jk}=1,~~~~~~\forall j\in [d\ell] \\
                                          &~~z_{jk}\in [0,1], ~~~~~~\forall j\in [d\ell], \forall k\in [\ell].\end{split}
\end{equation}

\end{definition}

Note that $\zeta_\delta(x)>0$ for any $x$, so the program $(P')$ is exactly a $d$-to-$1$ matching with positive edge weights. {We prove that the optimum of $(P')$ is close to the weight of the maximum weight $d$-to-$1$ matching (See Lemma~\ref{lem:P' and P''} and the proof of Theorem~\ref{thm:arbitrary-weight-algorithm}). 


Thus in the rest of this section we will consider approximating the optimum of $(P')$.}
Let $\hz$ be the solution produced by Algorithm~\ref{alg:online matching} on $(P')$. 
Program $(P')$ is the same as $(P)$ if we substitute the weight $\omega_{jk}$ for each LHS-node $r^{(j)}$ and RHS-node $s^{(k)}$ with $\zeta_\delta(\omega_{jk})$.
Recall that our main goal is to avoid being matched with negative edges too often. Now for every RHS-node, we construct a dummy $0$-RHS-node with weight 0 for all edges incident to it. Let the meta-RHS-node consists of the real RHS-node and the corresponding $0$-RHS-node. The weight between the LHS-node $j$ and the meta-RHS-node $k$ is defined as $\zeta_\delta(\omega_{jk})$. We will explain later why the weights are chosen in this way.

Think of the procedure that first executes Algorithm~\ref{alg:online matching} to find a matching between LHS-nodes and meta-RHS-nodes. As a second step, when a LHS-node $j$ is matched to some meta-RHS-node $k$, we further decide how to match it to the real RHS-node or the 0-RHS-node, according to the following ``softmax'' program between weight $\omega_{jk}$ and $0$ (see Figure~\ref{fig:surrogate differences} for an illustration):\footnote{Note that it's also the entropy regularized matching program between a single LHS-node $j$ and two RHS-nodes (real RHS-node $k$ and 0-RHS-node $k$).} 

\vspace{-.2in}
\begin{equation}
\begin{split}
\max &~~\textstyle x_{jk}\omega_{jk} -\delta\cdot x_{jk}\log(x_{jk})-\delta\cdot y_{jk}\log(y_{jk}) \\
\text{subject to} &~~\textstyle x_{jk} +y_{jk} = 1\\
 &~~x_{jk},y_{jk}\in [0,1].\end{split}
\end{equation}

Let $(x_{jk}^*,y_{jk}^*)$ be the optimal solution. One can easily verify that the optimum of the softmax program is equal to $\zeta_\delta(\omega_{jk})$. The two-step procedure finds a $d$-to-$1$ matching in the original graph (by removing all edges matched to $0$-RHS-nodes). Moreover, when the LHS-node $j$ is matched to the meta-RHS-node $k$, its expected weight $x_{jk}^*\omega_{jk}$ is at most $O(\delta)$ less than $\zeta_\delta(\omega_{jk})$. Thus by Theorem~\ref{thm:dughmi result}, the two-step procedure that executes Algorithm~\ref{alg:online matching} w.r.t. $(P')$ indeed finds an approximately-optimal $d$-to-$1$ matching w.r.t. weights $\zeta_\delta(\omega_{jk})$ and no LHS-node is matched to an edge with too negative weight.

The main issue with the above two-step procedure is that, to execute Algorithm~\ref{alg:online matching} w.r.t. $(P')$, we will have to sample from a distribution with mean $\zeta_\delta(\omega_{jk})$ with only sample access to {the distribution $\FF_{jk}$ whose mean is} $\omega_{jk}$. To the best of our knowledge, no algorithm exists to sample exactly from such a distribution.

\begin{figure}[h!]
  \centering
 \includegraphics[width=0.6\textwidth]{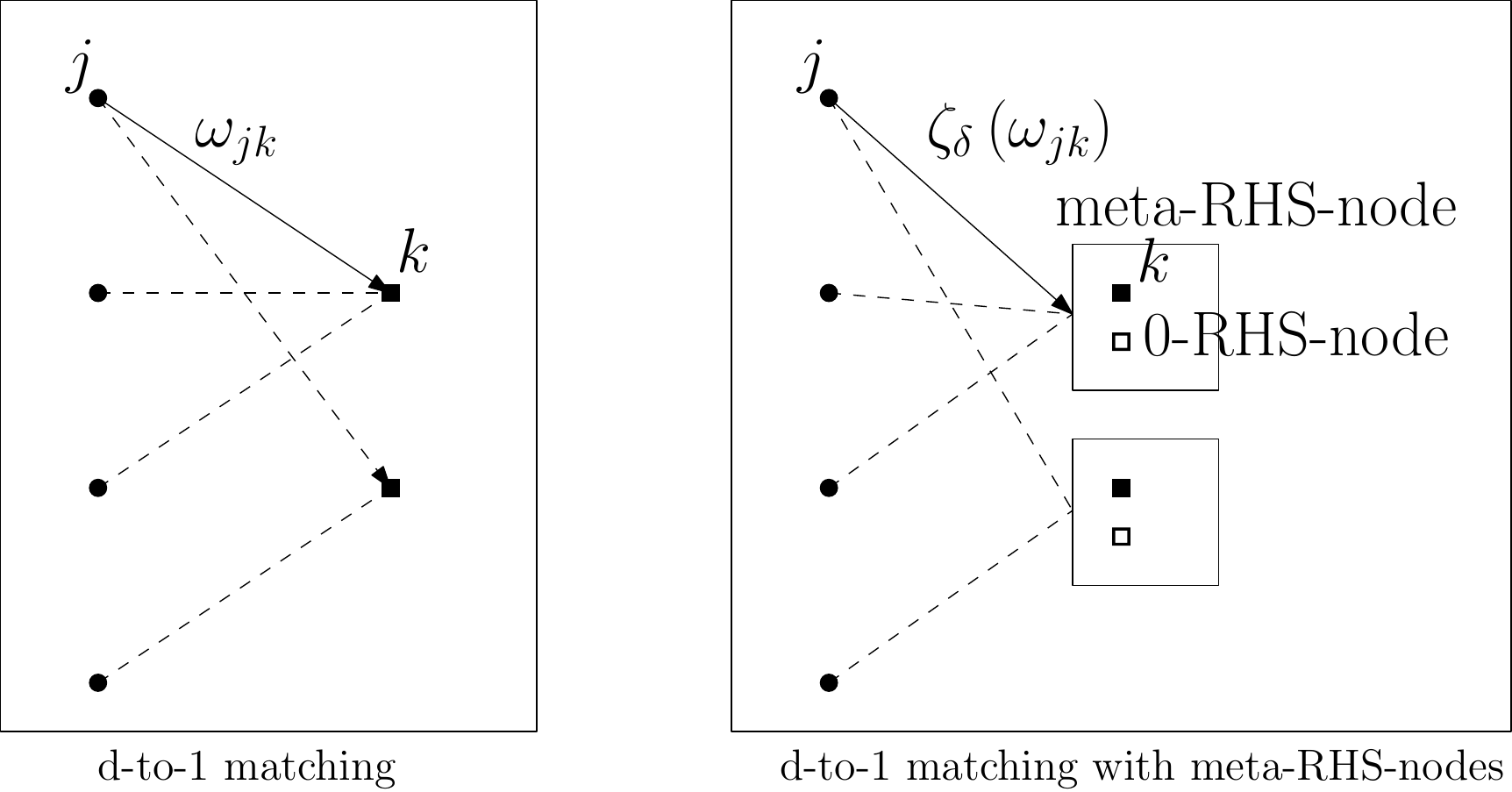}
  \caption{With $\bullet$ we denote the LHS-nodes, with $\blacksquare$ RHS-nodes, with $\square$ we denote the $0$-RHS-nodes and with a big rectangle that encloses a $\blacksquare$ and a $\square$ we denote the meta-RHS-nodes.}
  \label{fig:surrogate differences}
\end{figure}

\notshow{
To address these two issues, we introduce an auxiliary convex program $(P'')$.

\begin{definition}\label{def:augmented entropy regularized matching}
For any parameter $\delta>0$, we define the following auxiliary convex program $(P'')$:
\begin{align*}
\max~~~&\textstyle\sum_{j,k} x_{jk}\omega_{jk}-\delta\cdot \textstyle\sum_{j,k}(x_{jk}\log(x_{jk})+y_{jk}\log(y_{jk}))\\
\text{s.t.}~~~~&\textstyle\sum_j (x_{jk}+y_{jk})\leq d, ~~~~\forall k\in [\ell]\\
&\textstyle\sum_k (x_{jk}+y_{jk})=1,~~~~\forall j\in [d\ell]\\
&x_{jk},y_{jk}\in [0,1], ~~~~\forall j,k\\
\end{align*}
\end{definition}
\vspace{-.15in}
We prove in Lemma~\ref{lem:P' and P''} that the optimum of $(P'')$ is exactly the same as the optimum of $(P')$. The proof of Lemma~\ref{lem:P' and P''} is postponed to Appendix~\ref{sec:appx_proof_online_matching_2}. 
\begin{lemma}\label{lem:P' and P''}
For all $j\in[d\ell]$ and $k\in[\ell]$, if $\frac{x_{jk}}{y_{jk}} = \exp(\omega_{jk}/\delta), x_{jk}+y_{jk} = z_{jk}$, then
\begin{equation}\label{equ:lem:P' and P''}
z_{jk}\cdot\zeta_\delta(\omega_{jk})-\delta\cdot z_{jk}\log (z_{jk})=x_{jk}\cdot \omega_{jk}-\delta\cdot x_{jk}\log(x_{jk}) - \delta\cdot y_{jk}\log(y_{jk}).    
\end{equation}
This implies that the optimal objective values of $(P')$ and $(P'')$ are equal.
\end{lemma}

}

We present our algorithm (Algorithm~\ref{alg:new online matching}) that solves this issue. The key conceptual idea is to merge the two steps into one. It directly matches LHS-nodes to either real RHS-nodes or 0-RHS-nodes. The reason that this is possible is because the distribution from the combined procedure is again a Gibbs distribution, which allows us to use the fast exponential Bernoulli race to sample directly from it {(Observation~\ref{obs:dual_var})}.

\notshow{
We show that  $$G(\hz)=F(\hx,\hy)=\sum_{j,k} \hx_{jk}\omega_{jk}-\delta\sum_{j,k}(\hx_{jk}\log(\hx_{jk})+\hy_{jk}\log(\hy_{jk})),$$ where $(\hx_{jk},\hy_{jk})_{jk}$ is the solution produced by Algorithm~\ref{alg:new online matching}. In other words, running Algorithm~\ref{alg:online matching} on $(P')$ is equivalent to running Algorithm~\ref{alg:new online matching} on $(P'')$. 
}

\begin{algorithm}[h]
\begin{algorithmic}[1]
\REQUIRE Sample access to $\FF_{jk}$ whose mean is $\omega_{jk}$, for every $j,k$.
\STATE For each RHS-node $k$, add a $0$-RHS-node to the bipartite graph with edge weight $0$ to every LHS-node. We refer to the $k$-th original RHS-node the $k$-th normal-RHS-node.
    \FOR{$j \in [d\ell]$}
        \STATE Let $d_k^{(j-1)}$ be the number of LHS-nodes matched to either the $k$-th normal-RHS-node or the $k$-th $0$-RHS-node in the current matching and $K=\{k:d_k^{(j-1)}<d\}$. 
        \STATE 
        Set $\Balpha^{(j)}$ according to the Gibbs distribution over RHS-nodes in $K$, where the energy for any RHS-node $k\in K$ is $d_k^{(j-1)}$ and the temperature is $1/\eta'$. Set  $\alpha_k^{(j)}=0$ for all $k\not\in K$.
        \STATE Match LHS-node $j$ to a normal RHS-node (or a $0$-RHS-node)  $k\in K$ according to the Gibbs distribution  over the $2|K|$ RHS-nodes in $K$, where the temperature is $\delta$ and the energy for matching to a normal-RHS-node $k$ is  $(\omega_{jk}-\gamma\alpha_k^{(j)})$ and the energy for matching to a $0$-RHS-node $k\in K$ is $(-\gamma\alpha_k^{(j)})$. More specifically, match to the normal-RHS-node $k$ with probability $\hx_{jk}={\exp((\omega_{jk}-\gamma\alpha_k^{(j)})/\delta)\over \sum_{k'\in K}(\exp((\omega_{jk'}-\gamma\alpha_{k'}^{(j)})/\delta)+\exp((-\gamma\alpha_{k'}^{(j)})/\delta))}$ and match to the $0$-RHS-node $k$ with probability $\hy_{jk}={\exp((-\gamma\alpha_k^{(j)})/\delta)\over \sum_{k'\in K}(\exp((\omega_{jk'}-\gamma\alpha_{k'}^{(j)})/\delta)+\exp((-\gamma\alpha_{k'}^{(j)})/\delta))}$. 
        We can generate a sample from $(\hx_j,\hy_j)$  via the fast exponential Bernoulli race with {$ \poly(\gamma,\ell, 1/\delta)$} sample from $(\FF_{jk})_{k}$ in expectation (See Lemma~\ref{lem:sample-exp-weight}).

        
        \ENDFOR
\end{algorithmic}
\caption{{\sf Online Entropy Regularized Matching {with Arbitrary Edge Weights} (with parameters $\delta, \eta', \gamma$)}}
\label{alg:new online matching}
\end{algorithm}

{


We prove a coupling between executing Algorithm~\ref{alg:new online matching} over weights $(\omega_{jk})_{jk}$ and executing Algorithm~\ref{alg:online matching} over weights  $(\zeta_\delta(\omega_{jk}))_{jk}$, with the same parameters $\delta,\eta', \gamma$. Note that both executions are online procedures. Thus the distribution of matching the current LHS-node depends on the previous matching. We carefully prove that the dual variables $\Balpha^{(j)}$ and the remaining capacities $(d_k^{(j)})_k$ are the same for every round $j$. Our result is summarized in Theorem~\ref{thm:arbitrary-weight-algorithm}.
}


\notshow{
\begin{lemma}\label{lem:math}
For any $j,k$, if two numbers satisfy that $\frac{x_{jk}}{y_{jk}} = \exp(\omega_{jk}/\delta)$ and $x_{jk}+y_{jk} = z_{jk}$, then 
$z_{jk}\omega_{jk}' - \delta\cdot z_{jk}\log z_{jk} = x_{jk}\omega_{jk} - \delta(x_{jk}\log x_{jk} + y_{jk}\log y_{jk})$.
\end{lemma}
}

\notshow{
\mingfeinote{

Now we would like to show that the per-LHS-node weight of the online matching is close to that of the optimal offline matching by applying the result of~\cite{dughmi2017bernoulli}. Lemma~\ref{lem:bern_fact} shows this fact. We point out that choosing the scaling parameter $\gamma$ as a constant factor estimation of the optimal objective of program $(P')$ divided by $d$ is necessary to guarantee this property. Lemma~\ref{lem:find gamma} in Appendix~\ref{sec:appx_proof_online_matching} shows an algorithm to find such a $\gamma$ with high probability.

\begin{lemma}[\cite{DughmiHKN17}]
\label{lem:bern_fact}
For any $\psi>0$, running the Online Entropy Regularized Matching mechanism on program $(P')$ with parameters $\delta,\eta',\gamma$,
the per-LHS-node weight of the online matching is less than that of the offline optimal matching w.r.t. $(P')$ by at most an additive $\psi$,
when $\delta =\frac{\psi}{3\log(\ell)}$. The running time is $poly(\ell, 1/\psi)$.
\end{lemma}

}
}

\begin{theorem}\label{thm:arbitrary-weight-algorithm}
When $\omega_{jk}\in [-1,1]$~\footnote{Again the theorem applies to any bounded edge weights in $ [-\mathcal{R},\mathcal{R}]$. For simplicity we  normalize the edge weights to lie in $[-1,1]$.} for all $j,k$, Algorithm~\ref{alg:new online matching} satisfies the following properties:
\begin{enumerate}[itemsep=1pt]
    \item For any choice of the parameters, dropping all the edges incident to any $0$-RHS-nodes in the matching, the algorithm produces a feasible $d$-to-$1$ matching (not necessarily perfect).
    \item For any choice of the parameters, the algorithm is \textbf{maximal-in-range}. The expected running time and sample complexity is $\poly(d,\ell,1/\delta,\gamma)$.
    \item  For every $\delta,\eta'>0$, if $d\geq \ell\log\ell/\eta'^2$ and $\gamma\in \left[{\opt(P')\over d}, {O(1)\cdot\opt(P')\over d}\right]$, where $\opt(P')$ is the optimum of program $(P')$, then the expected value (over the randomness of the Algorithm~\ref{alg:new online matching}) of $\sum_{j\in[d\ell], k\in[\ell]} \hx_{jk}\omega_{jk}-\delta\sum_{j,k}\hx_{jk}\log(\hx_{jk})-\delta\sum_{j,k}\hy_{jk}\log(\hy_{jk})$ is at least $(1-O(\eta'))\cdot \opt(P')$. 
    
         Moreover, for every $\psi\in (0,1)$, if we set $\delta=\Theta(\frac{\psi}{\log\ell}), \eta'=\Theta(\psi)$, and $d$ and $\gamma$ satisfy the conditions above, then the expected value of $\sum_{j\in[d\ell], k\in[\ell]} \hx_{jk}\omega_{jk}$, the expected total weight of the matching output by the algorithm (dropping all the edges incident to any $0$-RHS-nodes in the matching), has weight at most  $O(d\ell\psi)$ less than the maximum weight matching. 
    \item{For every LHS-node $j$, the expected weight of the edge that matches $j$ is never too much smaller than $0$. Formally, for every $\delta$ and every $j\in [d\ell]$, $\sum_{k\in [\ell]}\hat{x}_{jk}\omega_{jk}\geq -\delta\cdot \log(2\ell)$.}

\end{enumerate}
\end{theorem}

\begin{remark}
Property (4) is relevant only when the edge weights may be negative. As discussed in the introduction, this is a crucial property to preserve the revenue in the transformation. Directly applying Algorithm~\ref{alg:online matching} from~\cite{dughmi2017bernoulli} is insufficient to guarantee this property as shown in Example~\ref{example:perfect_revenue}. 
\end{remark}

To prove Theorem~\ref{thm:arbitrary-weight-algorithm}, we consider the following auxiliary problem.

\begin{definition}\label{def:augmented entropy regularized matching}
For any parameter $\delta>0$, we define the following auxiliary convex program $(P'')$:
\begin{align*}
\max~~~&\textstyle F(x,y)=\sum_{j,k} x_{jk}\omega_{jk}-\delta\cdot \textstyle\sum_{j,k}(x_{jk}\log(x_{jk})+y_{jk}\log(y_{jk}))\\
\text{s.t.}~~~~&\textstyle\sum_j (x_{jk}+y_{jk})\leq d, ~~~~\forall k\in [\ell]\\
&\textstyle\sum_k (x_{jk}+y_{jk})=1,~~~~\forall j\in [d\ell]\\
&x_{jk},y_{jk}\in [0,1], ~~~~\forall j,k
\end{align*}
Let $(\hx_{jk},\hy_{jk})_{jk}$ be the solution produced by Algorithm~\ref{alg:new online matching}. 

\end{definition}

\begin{observation}\label{obs:exp weight with dummy-every replica}
For every $j$, $\alpha^{(j)}$ and parameter $\gamma$, match $j$ according to the Gibbs distribution $(\hx_{j},\hy_{j})$ to the available $2|K|$ RHS-nodes in $K$, 
$$\hx_{jk}=\frac{\exp\left((\omega_{jk}-\gamma\alpha_k^{(j)})/\delta\right)}{\sum_{k\in K}\left (\exp((\omega_{jk}-\gamma\alpha_k^{(j)})/\delta)+\exp(-\gamma\alpha_k^{(j)}/\delta)\right)},
 \hy_{jk}=\frac{\exp\left((-\gamma\alpha_k^{(j)})/\delta\right)}{\sum_{k\in K}\left(\exp((\omega_{jk}-\gamma\alpha_k^{(j)})/\delta)+\exp(-\gamma\alpha_k^{(j)}/ \delta)\right)}$$
maximizes
$$\sum_{k\in K}x_{jk}\omega_{jk}-\delta\sum_{k\in K}x_{jk}\log(x_{jk})-\delta\sum_{k\in K}y_{jk}\log(y_{jk})-\sum_{k\in K}\gamma\alpha_k^{(j)}\cdot (x_{jk}+y_{jk}),$$
subject to the constraint $\sum_{k}(x_{jk}+y_{jk}) = 1$.
\end{observation}


\begin{observation}
\label{obs:dual_var}
For every dual variables $\bm{\alpha}\in[0,h]^{\ell}$ the optimal solution $\bm{x}^*, \bm{y}^*$ maximizing the Lagrangian $L((\bm{x},\bm{y}), \alpha)$ of program $(P'')$ subject to the constraints $\sum_{k}(x_{jk}+y_{jk}) = 1, \forall j \in [d\ell]$ is

\begin{eqnarray*}
    x_{jk}^* = \frac{\exp(\frac{\omega_{jk}-\alpha_k}{\delta})}{\sum_{k'}\left(\exp(\frac{\omega_{jk'}-\alpha_{k'}}{\delta}) + \exp(\frac{-\alpha_{k'}}{\delta})\right)}, & \forall j \in [d\ell], \forall k \in [\ell] \\
    y_{jk}^* = \frac{\exp(\frac{-\alpha_k}{\delta})}{\sum_{k'}\left(\exp(\frac{\omega_{jk'}-\alpha_{k'}}{\delta}) + \exp(\frac{-\alpha_{k'}}{\delta})\right)}, & \forall j \in [d\ell],  \forall k \in [\ell]
\end{eqnarray*}

Hence
\begin{eqnarray*}
    \frac{x^*_{jk}}{y^*_{jk}} = \exp(\omega_{jk}/\delta), & \forall j,k
\end{eqnarray*}
\end{observation}




We prove in Lemma~\ref{lem:P' and P''} that the optimum of $(P'')$ is exactly the same as the optimum of $(P')$. 
\begin{lemma}\label{lem:P' and P''}
For all $j\in[d\ell]$ and $k\in[\ell]$, if $\frac{x_{jk}}{y_{jk}} = \exp(\omega_{jk}/\delta), x_{jk}+y_{jk} = z_{jk}$, then
\begin{equation}\label{equ:lem:P' and P''}
z_{jk}\cdot\zeta_\delta(\omega_{jk})-\delta\cdot z_{jk}\log (z_{jk})=x_{jk}\cdot \omega_{jk}-\delta\cdot x_{jk}\log(x_{jk}) - \delta\cdot y_{jk}\log(y_{jk}).    
\end{equation}
This implies that the optimal objective values of $(P')$ and $(P'')$ are equal.
\end{lemma}

\begin{prevproof}{Lemma}{lem:P' and P''}
\notshow{For every $j,k$, we observe that $\frac{x_{jk}}{z_{jk}} = \frac{exp(\omega_{jk}/\delta)}{1+ exp(\omega_{jk}/\delta)}$, hence $z_{jk} = \frac{x_{jk}(1+exp(\omega_{jk}/\delta))}{exp(\omega_{jk}/\delta)}$.
We simplify the equality that we need to prove to Equation~\eqref{eq:o0}.
\begin{eqnarray}
& z_{jk}\cdot\zeta_\delta(\omega_{jk})-\delta\cdot z_{jk}\log (z_{jk})=x_{jk}\cdot \omega_{jk}-\delta\cdot x_{jk}\log(x_{jk}) - \delta\cdot y_{jk}\log(y_{jk}) \nonumber \\
\Longleftrightarrow &
\zeta_\delta(\omega_{jk}) - \delta \log(z_{jk}) = \frac{x_{jk}}{z_{jk}} \omega_{jk} -
\delta \frac{x_{jk}}{z_{jk}} \log(x_{jk}) - \delta \frac{y_{jk}}{z_{jk}} \log(y_{jk}) \label{eq:o0}
\end{eqnarray}

Since:
\begin{equation}\label{eq:o1}
\zeta_\delta(\omega_{jk}) =  \delta \log(\exp(\omega_{jk}/\delta) + 1) 
\end{equation}
\begin{eqnarray}\label{eq:o2}
	\delta \log(z_{jk}) &=  \delta \log\left( \frac{x_{jk} (1+exp(\omega_{jk}/\delta))}{exp(\omega_{jk}/\delta)} \right) = \delta \log(x_{jk}) - \delta \log(\exp(\omega_{jk}/\delta)) + \delta \log(1+ \exp(\omega_{jk}/\delta)) \nonumber \\
&= \delta \log(x_{jk}) - \omega_{jk} + \zeta_\delta(\omega_{jk})
\end{eqnarray}
So the LHS of Equation~\eqref{eq:o0} equals to $\omega_{jk}-\delta\log(x_{jk})$. How about the RHS of Equation~\eqref{eq:o0}?

\begin{align}
& \frac{x_{jk}}{z_{jk}} \omega_{jk} -
\delta \frac{x_{jk}}{z_{jk}} \log(x_{jk}) - \delta \frac{y_{jk}}{z_{jk}} \log(y_{jk}) \nonumber  \\
=  & 
\frac{\exp(\omega_{jk}/\delta)}{1+\exp(\omega_{jk}/\delta)}\omega_{jk} 
- \delta \frac{\exp(\omega_{jk}/\delta)}{1+ \exp(\omega_{jk}/\delta)}\log(x_{jk}) 
- \delta \frac{1}{1+ \exp(\omega_{jk}/\delta)}\log\left( \frac{x_{jk}}{\exp(\omega_{jk}/\delta)} \right) \nonumber \\
=  & 
\frac{\exp(\omega_{jk}/\delta)}{1+\exp(\omega_{jk}/\delta)}\omega_{jk} 
- \delta \frac{\exp(\omega_{jk}/\delta)}{1+ \exp(\omega_{jk}/\delta)}\log(x_{jk}) 
- \delta \frac{1}{1+ \exp(\omega_{jk}/\delta)}\log\left( x_{jk}\right)
+ \delta \frac{1}{1+ \exp(\omega_{jk}/\delta)}\log\left(\exp(\omega_{jk}/\delta) \right) \nonumber \\
=  & 
\frac{\exp(\omega_{jk}/\delta)}{1+\exp(\omega_{jk}/\delta)}\omega_{jk} 
- \delta \frac{\exp(\omega_{jk}/\delta)}{1+ \exp(\omega_{jk}/\delta)}\log(x_{jk}) 
- \delta \frac{1}{1+ \exp(\omega_{jk}/\delta)}\log\left( x_{jk}\right)
+ \frac{1}{1+ \exp(\omega_{jk}/\delta)}\omega_{jk}\nonumber \\
=  & 
\frac{1+ \exp(\omega_{jk}/\delta)}{1+\exp(\omega_{jk}/\delta)}\omega_{jk} 
- \delta \frac{1+ \exp(\omega_{jk}/\delta)}{1+ \exp(\omega_{jk}/\delta)}\log(x_{jk}) \nonumber\\
= & \omega_{jk} - \delta \log(x_{jk}) \label{eq:o3}
\end{align}

Hence, Equation~\eqref{eq:o0}) holds.}

For every $j,k$, recall $\zeta_\delta(\omega_{jk})=\delta\cdot \log(\exp(\omega_{jk}/\delta)+1)$. Observe that $z_{jk} = (1+\exp(\omega_{jk}/\delta))y_{jk}$. We have 
$$\text{LHS}=z_{jk}\cdot\zeta_\delta(\omega_{jk})-\delta\cdot z_{jk}\log (z_{jk})=z_{jk}\cdot\delta\cdot\log(\frac{\exp(\omega_{jk}/\delta)+1}{z_{jk}})=-(\exp(\omega_{jk}/\delta)+1)y_{jk}\cdot\delta\log(y_{jk})$$

\begin{align*}
\text{RHS}&=-\delta\cdot y_{jk}\log(y_{jk})+x_{jk}\cdot (\omega_{jk}-\delta\log(x_{jk}))\\
&=-\delta\cdot y_{jk}\log(y_{jk})+\exp(\omega_{jk}/\delta)y_{jk}\cdot (\omega_{jk}-\delta\log (\exp(\omega_{jk}/\delta))-\delta\log y_{jk})\\
&=-\delta\cdot y_{jk}\log(y_{jk})\cdot (1+\exp(\omega_{jk}/\delta))
\end{align*}

Hence, Equation~\eqref{equ:lem:P' and P''} holds. Since the optimal values $x_{jk}^*,y_{jk}^*$ satisfy the requirements by Observation~\ref{obs:dual_var},
we have that the optimum of  $(P')$ is at least as large as the optimum of $(P'')$. On the other hand, let $z^*$ be the optimal solution of  $(P')$,
we can choose $x^*_{jk}$ and $y^*_{jk}$ so that $x^*_{jk}+y^*_{jk}=z^*_{jk}$ and ${x^*_{jk}\over y^*_{jk}}=\exp(\omega_{jk}/\delta)$. Clearly, $(x^*_{jk},y^*_{jk})_{jk}$ is a feasible solution to $(P'')$, therefore the optimum of  $(P')$ is at most as large as the optimum of $(P'')$. Combining. the two claims, we prove that $(P')$ and $(P'')$ have the same optimal objective values.
\end{prevproof}

\notshow{

\begin{prevproof}{lemma}{lem:math}
Follow from proof of Claim~\ref{clm:P' and P''}.
\end{prevproof}

}

\begin{lemma}\label{lem:entropy error}
With parameter $\delta\geq 0$, let $(x^*,y^*)$ be the optimal solution of $(P'')$. The optimum of $(P'')$, $\sum_{j,k} x^*_{jk}\omega_{jk}-\delta\cdot \sum_{j,k}(x^*_{jk}\log(x^*_{jk})+y^*_{jk}\log(y^*_{jk}))$, is no smaller than the 
weight of the maximum weight matching.
\end{lemma}
\begin{proof}
Let $x'$ be the maximum weight matching. It is not hard to see that we can construct a $0-1$ vector $y'$ so that $(x',y')$ is a feasible solution of $(P'')$. As both $x'$ and $y'$ only take values in $0$ or $1$, the entropy term $-\sum_{j,k}x'_{jk}\log(x'_{jk})-\sum_{j,k}y'_{jk}\log(y'_{jk})=0$. 
Hence, the optimum of $(P'')$ is at least as large as the weight of the maximum weight matching $\sum_{jk}x'_{jk}\omega_{jk}$.
\end{proof}

\begin{prevproof}{Theorem}{thm:arbitrary-weight-algorithm}
As the algorithm always produces a matching that respects the constraints of $(P'')$, the first  property clearly holds. As the set of available RHS-nodes $K$ and the dual variables $\Balpha^{(j)}$ only depend on the first $j-1$ LHS-nodes but not the LHS-node $j$, the maximal-in-range property follows from Observation~\ref{obs:exp weight with dummy-every replica}. The algorithm runs in $d\ell$ rounds, step 3 and 4 both take $O(\ell)$ time. Step 5 takes expected time $\poly(\gamma,\ell,1/\delta)$ and $\poly(\gamma,\ell,1/\delta)$-many samples from distributions $(\FF_{jk})_k$ to complete. Hence, the running time and sample complexity as stated in the second property.

If we execute Algorithm~\ref{alg:online matching} on a $d$-to-$1$ matching with weights $(\zeta_\delta(\omega_{jk}))_{jk}$ and Algorithm~\ref{alg:new online matching} over weights $(\omega_{jk})_{jk}$ with the same parameters $\delta,\eta', \gamma$, we can couple the two executions so that the dual variables $\Balpha^{(j)}$ and the remaining capacities $(d_k^{(j)})_k$ are the same for every $j$. We introduce the new notation $K^{(j)}$ which is exactly the set of available RHS-nodes $K$ in step 2 of both algorithm in round $j$. Note that $K^{(j)}$ is deterministically determined by $(d_k^{(j-1)})_k$. 
If $\Balpha^{(j)}$ and $K^{(j)}$ are the same in both algorithms for every $j$, then  $\hx_{jk}+\hy_{jk}=\hz_{jk}$ for every $j\in[d\ell]$ and $k\in K^{(j)}$. To verify this, simply observe that $$\hz_{jk}=\exp( \frac{\zeta_\delta(\omega_{jk})-\gamma\alpha_k^{(j)}}{\delta})=(\exp(\frac{\omega_{jk}}{\delta})+1)\cdot\exp(\frac{-\gamma\alpha_k^{(j)}}{\delta})=\exp(\frac{\omega_{jk}-\gamma\alpha_k^{(j)}}{\delta})+\exp(\frac{-\gamma\alpha_k^{(j)}}{\delta})=\hx_{jk}+\hy_{jk}.$$ 

How does the coupling work? We construct it by induction. In the base case where $j=1$, clearly everything is the same in both algorithms. Suppose the dual variables $\Balpha^{(1)},\ldots, \Balpha^{(j)}$ and the remaining capacities $(d_k^{(1)})_k,\ldots, (d_k^{(j)})_k$ are all the same for the first $j$ rounds, we argue that we can couple the two executions in round $j+1$ so that $\Balpha^{(j+1)}$ and $(d_k^{(j+1)})_k$ remain the same in both algorithms. First, the set $K^{(j+1)}$ is the same, which implies that the dual variables $\Balpha^{(j+1)}$ are also the same. Next, Algorithm~\ref{alg:online matching} samples a RHS-node $k$ according to distribution $\hz_{j+1}$ and Algorithm~\ref{alg:new online matching} samples a RHS-node according to distribution $(\hx_{j+1},\hy_{j+1})$. Note that $\hx_{(j+1)k}+\hy_{(j+1)k}=\hz_{(j+1)k}$, so wherever Algorithm~\ref{alg:online matching} matches the LHS-node $j+1$ to a RHS-node $k$ we match the LHS-node $j+1$ to the normal RHS-node $k$ with probability $\hx_{(j+1)k}\over \hz_{(j+1)k}$ and to the $0$-RHS-node with probability $\hy_{(j+1)k}\over \hz_{(j+1)k}$. Clearly, this coupling makes sure the new remaining capacities $(d_k^{(j+1)})_k$ also remain the same. Combining the coupling with Lemma~\ref{lem:P' and P''}, we conclude that 
$$G(\hz)=\sum_{j,k}\hz_{jk}\cdot\zeta_\delta(\omega_{jk})-\delta\cdot\sum_{j,k}\hz_{jk}\log(\hz_{jk})=\sum_{j,k} \hx_{jk}\omega_{jk}-\delta\cdot \sum_{j,k}(\hx_{jk}\log(\hx_{jk})+\hy_{jk}\log(\hy_{jk}))=F(\hx,\hy).$$

By Theorem~\ref{thm:dughmi result}, the expected value of $G(\hz)$ is a $(1-O(\eta'))$ multiplicative approximation to $\opt(P')$, if we choose the parameters according to the third property of the statement. Therefore, the expected value of $F(\hx,\hy)$ is a $(1-O(\eta'))$ multiplicative approximation to $\opt(P')$. Since the optimum of $(P'')$, $\opt(P'')$, is the same as $\opt(P')$ (Lemma~\ref{lem:P' and P''}),  the expected value of $F(\hx,\hy)$ is also a $(1-O(\eta'))$ multiplicative approximation to $\opt(P'')$. Now, invoke Lemma~\ref{lem:entropy error}, we know that the expected value of $F(\hx,\hy)$ is at least  a $(1-O(\eta'))$ multiplicative approximation to the weight of the maximum weight matching, which we denote as $\opt$. Note that the entropy term $-\delta\cdot \left(\sum_{j,k}\hx_{jk}\log(\hx_{jk})+\sum_{j,k}\hy_{jk}\log(\hy_{jk})\right)$ is non-negative and at most $\delta d\ell\log(2\ell)$, hence the expected weight of the matching produced by Algorithm~\ref{alg:new online matching}, the expected value of $\sum_{j,k} \hx_{jk}\omega_{jk}$, is at least $(1-O(\eta'))\cdot \opt-\delta d\ell\log(2\ell)$.

If we choose $\delta=\Theta(\frac{\psi}{\log\ell})$, $\eta'=\Theta(\psi)$, then  $\delta d\ell\log(2\ell)=\Theta(d\ell\psi)$ and $O(\eta')\cdot \opt = O(d\ell \psi)$ as $\opt\leq d\ell$. Thus, the expected weight of the matching produced by Algorithm~\ref{alg:new online matching} is within an additive error of $\Theta(d\ell\psi)$ from the weight of the maximum weight matching. This completes our proof for the third property.


{
Now we are going to prove the final bullet of the theorem.
We denote the entropy of $(x,y)$ as:
$$H(x,y) =- \left( \sum_{k}x_k \log(x_k) +\sum_k y_k \log(y_k)\right)$$
By Observation~\ref{obs:exp weight with dummy-every replica},
$$(\hat{x}_{j},\hat{y}_{j})=\argmax_{(x_{j},y_{j})}\sum_{k}x_{jk}\cdot \omega_{jk}+\delta\cdot H(x_{j},y_{j})-\sum_k\gamma\alpha_k^{(j)}\cdot (x_{j k}+y_{j k})$$

Therefore the objective of solution $(0,\hat{x}_{j} + \hat{y}_{j})$ is a lower bound on the objective of solution $(\hat{x}_{j},\hat{y}_{j})$:

\begin{align*}
	& \sum_{k}\hat{x}_{j k}\cdot \omega_{jk}+\delta\cdot H(\hat{x}_j,\hat{y}_j)-\sum_k\gamma\alpha_k^{(j)}\cdot (\hat{x}_{j k}+\hat{y}_{j k})\\
	&\geq \delta\cdot H(0,\hat{x}_j+\hat{y}_j) - \sum_k\gamma\alpha_k^{(j)}\cdot (\hat{x}_{j k}+\hat{y}_{j k})
\end{align*}

Since $H(0,\hat{x}_j+\hat{y}_j) \geq 0$,
we can conclude that:

\begin{align*}
	\sum_{k}\hat{x}_{j k}\cdot \omega_{jk} \geq -&\delta\cdot H(\hat{x}_j,\hat{y}_j)
	\geq - \delta \log(2 \ell)
\end{align*}

}
\end{prevproof}

\section{$\varepsilon$-BIC to BIC Transformation}\label{sec:mechanism}




In this section, we present our  $\varepsilon$-BIC to BIC transformation. In Theorem~\ref{thm:general reduction_2}, we prove a more general statement where the given mechanism $\MM$ is $\varepsilon$-BIC with respect to $\DD=\bigtimes_{i\in[n]}\DD_i$, while we construct an exactly BIC mechanism $\MM'$ with respect to a different distribution $\DD'=\bigtimes_{i\in[n]}\DD'_i$. If $\DD=\DD'$, the problem is the $\varepsilon$-BIC to BIC transformation problem.  We show that the revenue of $\MM'$ under $\DD'$ decreases gracefully with respect to the Wasserstein Distance of the two distributions. For every $i$, we denote $\wass(\mathcal{D}_i,\mathcal{D}_i')$ the $\ell_{\infty}$-Wasserstein Distance of distributions $\mathcal{D}_i$, $\mathcal{D}_i'$. We slightly abuse notation and let  $\wass(\mathcal{D},\mathcal{D}')=\sum_{i=1}^n \wass(\mathcal{D}_i,\mathcal{D}_i')$. 

Our mechanism works in the following way. After some agent reports her type, we sample $d\ell - 1$ i.i.d. replicas, $\ell$ i.i.d. ``real'' surrogates and add $\ell$ 0-surrogates, for some appropriately chosen parameters $d, \ell$. The true type is inserted into a random position in the replicas. We define the weight between a replica $r^{(j)}$ and a ``real'' surrogate $s^{(i)}$ to be (almost) the expected utility of type $r^{(j)}$ if she reported $s^{(i)}$ to the original mechanism $\MM$. The weights between replicas and 0-surrogates are 0. Then, we run Algorithm~\ref{alg:new online matching} in order to get a $d$-to-1 matching between the replicas and the surrogates. To ensure that this matching is truthful, we impose appropriately selected payments to the agent. This is the first phase of $\MM'$. Now, suppose that the true agent is matched to surrogate $s^{(k)}$. In the second phase of the mechanism, we let $s^{(k)}$ participate in $\MM$ along with the matched surrogates of the other agents. If $s^{(k)}$ is a 0-surrogate, then the agent gets the outcome $\perp$ and pays nothing. Otherwise, she gets the outcome that the surrogate gets and pays the same price discounted by a factor of $(1-\eta)$. We remark that a downward-closed outcome space $\mathcal{O}$ is necessary here to allow a $\perp$ outcome. 

A formal description of our mechanism is shown in Mechanism~\ref{alg:mout}. In step 3, the parameter $\gamma$ is estimated using an approach similar to Dughmi et al.~\cite{dughmi2017bernoulli}. See Lemma~\ref{lem:find gamma} in Appendix~\ref{appx:offline optimal approximation} for more details. 

\floatname{algorithm}{Mechanism}
\begin{algorithm}[t]
\begin{algorithmic}[1]
\REQUIRE Query access to an IR mechanism $\MM=(x,p)$ w.r.t. $\DD=\bigtimes_{i\in[n]} \DD_i$; sample access to the type distribution $\mathcal{D}_i$ and $\mathcal{D}_i'$ for every $i\in [n]$; Parameters $\eta, \eta',\delta,\ell,$ and $d\geq 32 \frac{\log(8 \eta'^{-1})}{\delta^2 \ell \log^2(\ell)}$.\\
\textbf{Phase 1: Surrogate Selection}
\FOR{$i \in [n]$}
    \STATE Sample $\ell$ surrogates i.i.d. from $\mathcal{D}_i$. We use $\Bs$ to denote all surrogates. 
    \STATE {Estimate $\gamma$ with parameters $\eta'$ and $\delta$ using the algorithm in Lemma~\ref{lem:find gamma}.}
    \STATE Agent $i$ reports her type $t_i$. Create $d\ell-1$ replicas sampled i.i.d. from $\mathcal{D}_i'$ and insert $t_i$ into the replicas at a uniformly random position. We use $\Br$ to denote all the $d\ell$ replicas.
    \STATE For each normal surrogate $k$, also create a $0$-surrogate with a special type $\diamond$ . Create a bipartite graph $G_i$ between the $d\ell$ replicas and $2\ell$ surrogates. Define the weight between the $j$-th replica $\rj$ ($t_i$ is also a replica) and the $k$-th normal surrogate $\sk$ using     $$W_i(\rj,\sk)= \E_{t_{-i}\sim \DD_{-i}}\left[v_i(\rj,x(\sk, t_{-i}))\right]-(1-\eta)\cdot \E_{t_{-i}\sim \DD_{-i}}\left[p_i(\sk,t_{-i})\right].$$ A $0$-surrogate has edge weight $0$ to every replica, that is $W_i(\rj,\diamond)=0$ for all $j$.
        \STATE Run Algorithm~\ref{alg:new online matching} on $G_i$ with parameters $\delta$, $\eta'$, and $\gamma$. For any edge between a replica $\rj$ and a surrogate $\sk$, we can sample the edge weight by first sampling $t_{-i}$ from $\DD_{-i}$, then query $\MM$ on input $(\sk,t_{-i})$, and compute $v_i(\rj,x(\sk,t_{-i}))-(1-\eta)\cdot p_i(\sk,t_{-i}).$
    \STATE Suppose the reported type $t_i$ of agent $i$ is matched to the $k$-th normal surrogate or the $k$-th $0$-surrogate. Let $s_i$ be the type of the $k$-th normal surrogate.
    \STATE Sample $\lambda$ from $U[0,1]$ and charge the agent $q_i(t_i,\lambda)$, which is her payment for Phase 1. $q_i(t_i,\lambda)$ is computed via a modified implicit payment (Defintion~\ref{def:implicit payment}).
\ENDFOR\\
 \textbf{Phase 2: Surrogate Competition}
\STATE Run mechanism $\MM$ on input $s=(s_1,...,s_n)$. Let $o=(o_1,\ldots,o_n)$ be a random outcome sampled from $x(s)$. If agent $i$ is matched to a normal surrogate in Phase 1, her outcome is $o_i$ and her payment for Phase 2 is 
$(1-\eta)\cdot p_i(s)$; otherwise the agent gets the outcome $\perp$ and pays $0$ for Phase 2.
\end{algorithmic}
\caption{{\sf $\varepsilon$-BIC to BIC Transformation (Mechanism $\mout$})}
\label{alg:mout}
\end{algorithm}

How do we compute the payment of Phase  1? Note that if any agent $i'\in[n]$ reports truthfully, then the surrogate $s_{i'}$ who participates for agent $i'$ in Phase 2~\footnote{Agent $i'$ may be matched to a $0$-surrogate, then $s_{i'}$ is the type of the corresponding normal surrogate.}  is exactly drawn from distribution $\DD_{i'}$. Therefore, if all the other agents report truthfully, agent $i$'s value for winning a normal surrogate $s$ is exactly $W_i(t_i, s)$ and $0$ otherwise. In other words, Mechanism~\ref{alg:mout} is equivalent to a competition among replicas to win surrogates, and the edge weight between a replicas and a surrogate is exactly the replica's value for the surrogate. To show that Mechanism~\ref{alg:mout} is BIC, it suffices to prove that the payment of Phase 1 incentivizes the replicas to submit their true edge weights. As Algorithm~\ref{alg:new online matching} is maximal-in-range, such payment rule indeed exists. 

{If the true type is the $j$-th replica, and the reported type $t_i$ induces edge weights $(\omega_{jk})_{k\in[\ell]}$, charge the agent \begin{equation}\label{eq:explicit payment}
 	\delta\sum_{k\in K}x_{jk}\log(x_{jk})+\delta\sum_{k\in K}y_{jk}\log(y_{jk})+\sum_{k\in K}\gamma\alpha_k^{(j)}\cdot (x_{jk}+y_{jk}), 
 \end{equation}
 where $x_{jk}=\frac{\exp\left((\omega_{jk}-\gamma\alpha_k^{(j)})/\delta\right)}{\sum_{k\in K}\left (\exp((\omega_{jk}-\gamma\alpha_k^{(j)})/\delta)+\exp(-\gamma\alpha_k^{(j)}/\delta)\right)}$, $
 y_{jk}=\frac{\exp\left((-\gamma\alpha_k^{(j)})/\delta\right)}{\sum_{k\in K}\left(\exp((\omega_{jk}-\gamma\alpha_k^{(j)})/\delta)+\exp(-\gamma\alpha_k^{(j)}/ \delta)\right)}$, and $\alpha^{(j)}$ is the set of dual variables in the $j$-th iteration of Algorithm~\ref{alg:new online matching}. Observation~\ref{obs:exp weight with dummy-every replica} implies that the payment rule is BIC.}
However, direct implementation of the payment requires knowing the edge weights which we only have sample access to. We use a procedure called the implicit payment computation~\cite{archer2004approximate, HartlineL10, babaioff2013multi,babaioff2015truthful, dughmi2017bernoulli} to circumvent this difficulty.

\begin{definition}[Implicit Payment Computation]\label{def:implicit payment}
	For any fixed  parameters $\delta$,$\eta$, $\eta'$ and $\gamma$, let $(\omega_{jk})_{jk}$ be the edge weights on a $[d\ell]\times[2\ell]$ size bipartite graph, we use $\AA_{j}(\omega)$ to denote $(\hx_{j1},\ldots,\hx_{j\ell},\hy_{j1},$\\$\ldots, \hy_{j\ell})$,  the allocation of the $j$-th LHS-node/replica to the surrogates computed by Algorithm~\ref{alg:new online matching} on the bipartite graph. 	
	Now, fix $\Br$ and $\Bs$, we use $u_i(t_i,(x,y))$ to denote $\sum_{k\in [\ell]} x_{k}\cdot W_i(t_i,\sk)$. Suppose agent $i$'s reported type $t_i$ is in position $\pi$, that is, $r^{(\pi)}=t_i$. To compute price  $q_i(t_i,\lambda)$, let surrogate $s'$ be the surrogate sampled from $\AA_{\pi}(W)$ by Algorithm~\ref{alg:new online matching} in step 6,  where $W$ is the collection of edge weights in graph $G_i$ as defined in step 5 of Mechanism~\ref{alg:mout}, and we sample a surrogate $s''$ from $\AA_{\pi}\left(\lambda W_{\pi},W_{-\pi}\right)$, where $W_{\pi}$ contains all weights of the edges incident to the $\pi$-th replica, and $\lambda W_{\pi}$ is simply multiplying each weight in $W_{\pi}$ by $\lambda$. Then we sample $t_{-i}$ from $\DD_{-i}$, the price $q_i(t_i,\lambda)$ is 
	$$weight_i(t_i, s',t_{-i}))-weight_i(t_i, s'',t_{-i})-\sqrt{\delta}(\log (2\ell) +1),$$ 
	where $weight_i(t_i,s,t_{-i})=v_i(t_i,x(s,t_{-i}))-(1-\eta)\cdot p_i(s,t_{-i})$ if $s\neq \diamond$, otherwise $weight_i(t_i, s,t_{-i})=0$.

	In expectation over $s'$, $s''$ and $t_{-i}$, 	$$\E\left[ q_i(t_i,\lambda)\right] = u_i\left(t_i,\AA_{\pi}(W)\right)- u_i\left(t_i,\AA_{\pi}\left(\lambda W_{\pi},W_{-\pi}\right)\right)-\sqrt{\delta}(\log (2\ell) +1),$$
	if we also take expectation over $\lambda$, 
	$$\E_{\lambda\sim U[0,1]}[q_i(t_i,\lambda)]=u_i\left(t_i,\AA_{\pi}(W)\right)-\int_0^1 u_i\left(t_i,\AA_{\pi}\left(\lambda W_{\pi},W_{-\pi}\right)\right)d \lambda-\sqrt{\delta}(\log(2 \ell) +1).~\footnote{The difference between $\E_{\lambda\sim U[0,1]}[q_i(t_i,\lambda)]$ and Equation~\eqref{eq:explicit payment} is indeed a fixed constant, hence our mechanism is BIC.}$$

\end{definition}


With Definition~\ref{def:implicit payment}, our mechanism is fully specified. We proceed to prove that the mechanism is BIC and IR. Our transformation is quite robust. Even if the original mechanism $\MM$ is not $\varepsilon$-BIC or the $\gamma$ estimated in step 3 is not a constant factor approximation of ${\opt({\omega}(\Br)) \over d}$, the mechanism is still BIC and IR. The proof for truthfulness is similar to the one in Dughmi et al.~\cite{dughmi2017bernoulli}. However, as our edge weights may be negative, it is  more challenging to establish the individually rationality compared to Dughmi et al.~\cite{dughmi2017bernoulli}. To make sure the mechanism is IR, we sometimes need to use negative payments to subsidize the agents, and at the same time guarantee that the total subsidy is negligible compared to the overall revenue. Note that this is also different from the previous $\varepsilon$-BIC to BIC transformations~\cite{DaskalakisW12, RubinsteinW15, CaiZ17}, as they essentially use the VCG mechanism to match surrogates to replicas, their mechanisms are clearly individually rational and use non-negative payments. 
The proof of Lemma~\ref{lem:M' BIC} is postponed to Appendix~\ref{sec:appx_sec_5}.

\begin{lemma}\label{lem:M' BIC}
For any choice of the parameters $\ell$, $d$, $\eta$, $\eta'$, $\delta$ and any IR mechanism $\MM$, $\MM'$ is a BIC and IR mechanism w.r.t. $\DD'$. In particular, we do not require $\MM$ to be $\varepsilon$-BIC. Moreover, each agent $i$'s expected Phase 1 payment $\E\left[ q_i(t_i,\lambda)\right]$ is at least  $ -\sqrt{\delta}(\log (2\ell) +1)$. Finally,  on any input bid $b=(b_1,\ldots,b_n)$, $\MM'$ computes the outcome in expected running time $\poly(d,\ell,1/\eta',1/\delta)$ and makes {in expectation} at most {$\poly(d,\ell,1/\eta',1/\delta)$}~queries~to~$\MM$.

\end{lemma}

We are now ready to present our main result for this section (Theorem~\ref{thm:general reduction_2}). 
To prove the theorem, it suffices to lower bound the revenue of $\MM'$ from the second phase due to  Lemma~\ref{lem:M' BIC}. In the previous transformations~\cite{DaskalakisW12, RubinsteinW15, CaiZ17}, the mechanism computes an exact maximum weight replica-surrogate matching, which allows them to bound the revenue from the second phase directly. Our mechanism only computes an approximately maximum weight replica-surrogate matching. As a result, we need to use a more delicate analysis to lower bound the revenue from Phase 2. We refer the readers to Appendix~\ref{sec:rev M'} for more details about bounding the revenue loss and the complete proof of Theorem~\ref{thm:general reduction_2}.

\begin{theorem}\label{thm:general reduction_2}
Let $\OO$ be a downward-closed outcome space. Given sample access to distributions $\mathcal{D}=\bigtimes_{i\in[n]}\DD_i$ and $\mathcal{D'}=\bigtimes_{i\in[n]}\DD'_i$, and query access to an $\varepsilon$-BIC and IR mechanism $\MM$ w.r.t. distribution $\mathcal{D}$. We can construct an exactly BIC and IR mechanism $\MM'$ w.r.t. distribution $\mathcal{D'}$, such that
\begin{equation}\label{eq:gap in revenue}
	  \rev(\MM',\DD')\geq \rev(\MM,\DD) -  O(n\sqrt{\varepsilon})-O\left(\sqrt{{n\cdot d_w(\dD,\dD' )}}\right).
\end{equation}
On any input bid $b=(b_1,\ldots,b_n)$, $\MM'$ computes the outcome and payments  in expected running time $\poly(n,T',1/\varepsilon)$ and makes {in expectation} at most {$\poly\left(n,T',1/\varepsilon\right)$}~queries~to~$\MM$, where $\TT_i'$ is the support of $\DD_i'$ and $T'=\max_{i\in [n]} |\TT_i'|$.

Furthermore, for any coupling $c_i (\cdot)$ between $\mathcal{D}_i$ and $\mathcal{D}_i'$ such that $v_i$ is non-increasing w.r.t. $c_i(\cdot)$~\footnote{Roughly speaking, $v_i$ is non-increasing w.r.t. a coupling $c_i$ if the coupling always couples a ``higher'' type to a ``lower'' type. Namely, for all $t_i$, outcome  $o\in\cO$, if the coupling produces type $t_i$ and $c_i(t_i)$, then $v_i(t_i,o)\geq v_i(c_i(t_i),o)$.} 
, the error bound can be improved as follows:


\begin{equation}\label{eq:improved non increasing coupling}
    \rev(\MM',\DD')\geq \rev(\MM,\DD)-n\sqrt{\varepsilon}-O\left(n\eta+{n\varepsilon\over \eta}\right)-{\sum_{i\in [n]}\E_{t\sim \DD'}\left[\E_{c_i(t_i)}\left[v_i\left(t_i,x'(t)\right)-v_i\left(c_i(t_i),x'(t)\right)\right]\right]\over \eta},
\end{equation}

where $x'(\cdot)$ is the allocation rule of $\MM'$ and $\eta$ can be chosen to be an arbitrary constant in $(0,1)$.
\end{theorem}



Inequality~\eqref{eq:gap in revenue} is our main result, and provides a strong guarantee in very general settings. Even though the difference between Inequality~\eqref{eq:improved non increasing coupling} and~\eqref{eq:gap in revenue} seems small,  we like to point out that the difference can be substantial sometimes and there were indeed cases where one needed a sharper version similar to Inequality~\eqref{eq:improved non increasing coupling}. In particular, one common application of bounds similar to Inequality~\eqref{eq:improved non increasing coupling} is when the coupling simply rounds values down. For example, the main results in~\cite{CaiZ17,kothari2019approximation} heavily rely on  inequalities similar to Inequality~\eqref{eq:improved non increasing coupling}, and these results may not be possible if only an Inequality~\eqref{eq:gap in revenue} type bounds are used.

When $\mathcal{D}=\mathcal{D'}$, $\wass(\mathcal{D},\mathcal{D'})=0$, the following corollary states the $\varepsilon$-BIC to BIC transformation.

\begin{corollary}\label{col:small wass distance}
If $\mathcal{D}=\mathcal{D'}$, $\rev(\MM',\DD) \geq \rev(\MM,\DD) -{O}\left(n\sqrt{\varepsilon} \right)$.
\end{corollary}

Another useful corollary is when we choose $\MM$ to be the optimal BIC, IR mechanism for $\DD$, we conclude that the optimal revenue under $\DD'$ is not far away from the optimal revenue under $\DD$.
\begin{corollary}\label{cor:opt is close with small wass distance is }
If $\wass(\DD_i,\DD'_i)\leq \kappa$ for all $i\in[n]$, let $\opt(\DD)$ and $\opt(\DD')$ be the optimal revenue achievable by any BIC and IR mechanism w.r.t. $\DD$ and $\DD'$ respectively. Then $\left\lvert\opt (\DD)-\opt(\DD')\right\rvert\leq O(n\cdot\sqrt{\kappa})$.
\end{corollary}



\section{Black-box Reduction for Multi-Dimensional Revenue Maximization}\label{sec:apps}

In this section, we apply Theorem~\ref{thm:general reduction_2} to the multi-dimensional revenue maximization problem. 
Cai et al. \cite{CaiDW13b} provide a reduction from MRM to VWO.
More formally:

\begin{theorem}[Rephrased from Theorem~2 of Cai et al.~\cite{CaiDW13b}]
\label{thm:rev-max}
Let $\OO$ be a general outcome space. Given the bidders' type distributions $\DD=\bigtimes_i \DD_i$. Let $b$ be an upper bound on the bit complexity of $v_i(t_i,o)$ and $\Pr(t_i)$ for any agent $i$, any type $t_i$, and any outcome $o$, and $\opt$~be the optimal revenue achievable by any BIC and IR mechanisms. We further assume that  types are normalized, that is, for each agent $i$, type $t_i$ and outcome $o$,
$v_i(t_i,o)\in [0,1]$.

 Given oracle access to an $\alpha$-approximation algorithm $G$ for VWO
with running time $rt_G(x)$, where $x$ is the bit complexity of the input,
there is an algorithm that terminates in $\poly\left(n, T,\frac{1}{\varepsilon},b,rt_G\left(\poly\left( n, T,\frac{1}{\varepsilon},b \right)\right)\right)$ time, and outputs a mechanism with expected revenue $\alpha OPT - \varepsilon$ that is $\varepsilon$-BIC with probability at least $1-\exp(-n/\varepsilon)$. Recall that $T=\max_{i\in [n]} |\TT_i|$. On any input bid, the mechanism computes the outcome and payments in expected running time $\poly\left(n, T,\frac{1}{\varepsilon},b,rt_G\left(\poly\left( n, T,\frac{1}{\varepsilon},b \right)\right)\right)$. 

\end{theorem}

We can apply Theorem~\ref{thm:general reduction_2} to the final mechanism produced by Theorem~\ref{thm:rev-max} and obtain an exactly BIC mechanism with almost the same revenue.

\begin{theorem}
\label{thm:app2}
Let $\OO$ be a downward-closed outcome space.
Given the bidders' type distributions $\DD=\bigtimes_i \DD_i$. Let $b$ be an upper bound on the bit complexity of $v_i(t_i,o)$ and $\Pr(t_i)$ for any agent $i$, any type $t_i$, and any outcome $o$, and $\opt$~be the optimal revenue achievable by any BIC and IR mechanisms. We further assume that  types are normalized, that is, for each agent $i$, type $t_i$ and outcome $o$,
$v_i(t_i,o)\in [0,1]$.

 Given oracle access to an $\alpha$-approximation algorithm $G$ for VWO
with running time $rt_G(x)$, where $x$ is the bit complexity of the input,
there is an algorithm that terminates in $\poly\left(n, T,\frac{1}{\varepsilon},b,rt_G\left(\poly\left( n, T,\frac{1}{\varepsilon},b \right)\right)\right)$ time, and outputs an \textbf{exactly BIC} and IR mechanism with expected revenue 
$$\rev(\MM,\DD) \geq \alpha\cdot OPT - O\left(n  \sqrt{\varepsilon} \right),$$

where $T=\max_{i\in [n]} |\TT_i|$. On any input bid, $\MM$ computes the outcome and payments in expected running time $\poly\left(n, T,\frac{1}{\varepsilon},b,rt_G\left(\poly\left( n, T,\frac{1}{\varepsilon},b \right)\right)\right)$. 

\end{theorem}

\begin{prevproof}{Theorem}{thm:app2}
When the mechanism computed by Theorem~\ref{thm:rev-max} is $\varepsilon$-BIC, our transformation converts it into a BIC mechanism with at most $O(n\sqrt{\varepsilon})$ less revenue. The important property of our transformation as stated in Lemma~\ref{lem:M' BIC} is that even if the initial mechanism is not $\varepsilon$-BIC, our transformation still produces an exactly BIC mechanism. In this case, we can still treat the given mechanism as $1$-BIC and IR, and use the corresponding revenue guarantees provided by Theorem~\ref{thm:general reduction_2}. Since the probability that the mechanism computed by Theorem~\ref{thm:rev-max} is not $\varepsilon$-BIC is  exponentially small, we can absorb the loss from this exponentially small event in the error term $O(n\sqrt{\varepsilon})$. The time complexity follows from Theorem~\ref{thm:general reduction_2} and~\ref{thm:rev-max}.
\end{prevproof}


Since our Theorem~\ref{thm:general reduction_2} allows us to construct a close to optimal mechanism $\MM'$ w.r.t. the type distribution $\DD_i'$, if $\DD'$ is not too far away from the distribution $\DD$ that $\MM$ is designed , we can approximate the optimal revenue even when we only have sample access to the bidders' type distributions. A byproduct of this result is that the running time of our algorithm no longer depends on the bit complexity of the probability that a particular type shows up.

\notshow{
Before we state the theorem, we informally provide the basic terminology \cite{CaiDW13b}.
Let $\FF$ be the set of feasible outcomes,
$n$ be the number of buyers
and $\mathcal{V}$ be the set of allowable types that map each $f \in \FF$ to a value in $\R$.
\begin{itemize}
    \item \textbf{Multi-Dimensional Revenue Maximization (MRM):} Given as input $n$ distributions $\mathcal{D}_1,\ldots,\mathcal{D}_n$ over types in $\VV$ and a set of feasible outcomes $\mathcal{F}$. Let $T_i$ be the support of $\dD_i$, for every type $t_i\in T_i$, $v_i(t_i,\cdot)$ is a valuation function that maps every outcome to a nonnegative number. The output should be a BIC mechanism $M$ who chooses outcomes from $\mathcal{F}$ with probability $1$ and whose expected revenue is optimal relative to any other, possibly randomized, BIC mechanism when played by $n$ bidders whose types are sampled from $\mathcal{D} = \bigtimes_{i\in[n]} \mathcal{D}_i$.
    \item \textbf{Virtual Welfare Optimization (VWO):} Given as input $n$ functions $C_i(\cdot): T_i\mapsto \mathbb{R}$ and a a set of feasible outcomes $\mathcal{F}$, output should be an outcome $o \in \argmax_{x\in \FF} \sum_i\sum_{t_i\in T_i}C_i(t_i)\cdot v_i(t_i,x)$. We refer to the sum $\sum_i\sum_{t_i\in T_i}C_i(t_i)\cdot v_i(t_i,x)$ as the virtual welfare of outcome $x$.
\end{itemize}

}

\begin{theorem}
\label{thm:dist_rev}
Let $\OO$ be a downward-closed outcome space.
Given \textbf{sample access} to bidders' type distributions $\DD=\bigtimes_i \DD_i$. Let $b$ be an upper bound on the bit complexity of $v_i(t_i,o)$ and $\Pr(t_i)$ for any agent $i$, any type $t_i$, and any outcome $o$, and $\opt$~be the optimal revenue achievable by any BIC and IR mechanism. We further assume that  types are normalized, that is, for each agent $i$, type $t_i$ and outcome $o$,
$v_i(t_i,o)\in [0,1]$.

 Given oracle access to an $\alpha$-approximation algorithm $G$ for VWO
with running time $rt_G(x)$, where $x$ is the bit complexity of the input,
there is an algorithm that terminates in $\poly\left(n, T,\frac{1}{\varepsilon},b,rt_G\left(\poly\left( n, T,\frac{1}{\varepsilon},b \right)\right)\right)$ time, and outputs an \textbf{exactly BIC} and IR mechanism with expected revenue $$
\rev(\MM,\DD) \geq \alpha\cdot OPT - O\left(n  \sqrt{\varepsilon} \right),$$

where $T=\max_{i\in [n]} |\TT_i|$. On any input bid, $\MM$ computes the outcome and payments in expected running time $\poly\left(n, T,\frac{1}{\varepsilon},b,rt_G\left(\poly\left( n, T,\frac{1}{\varepsilon},b \right)\right)\right)$.

\notshow{
Given oracle access to Algorithm~G (Theorem~\ref{thm:rev-max}) and only sample access to the bidders' distributions $\dD_1,\ldots,\dD_n$,
we can compute a exactly BIC Mechanism $\MM$ in time $\poly\left(\sum_i{|T_i|},\frac{1}{\varepsilon},b,rt_G\left(\poly\left( \sum_i{|T_i|},\frac{1}{\varepsilon},b \right)\right)\right)$
with expected revenue

$$
\rev(\MM,\DD) \geq \alpha OPT -  O\left(n \sqrt{\varepsilon} \right). 
$$
}


\end{theorem}

\begin{prevproof}{Theorem}{thm:dist_rev}
We can create an empirical distribution $\widetilde{\dD}_i$ for each bidder $i$,
such that $d_{TV}(\dD_{i},\widetilde{\dD}_i) \leq \varepsilon', \forall i$, 
with probability at least $1-\theta$ using $O(\sum_{i}^n\frac{|T_i|^2}{\varepsilon'^2}\ln\frac{\sum_{j=1}^n|T_j|}{\theta})$ samples. 

We first consider the case where $d_{TV}(\dD_{i},\widetilde{\dD}_i) \leq \varepsilon', \forall i$.
Then, $d_{w}(\dD_i,\widetilde{\dD}_i) \leq  d_{TV}(\dD_{i},\widetilde{\dD}_i) \leq  \varepsilon'$, as the highest value for any outcome is at most $1$. Apply Theorem~\ref{thm:rev-max} on $\widetilde{\DD}=\bigtimes_{i=1}^n \widetilde{\DD}_i$ and let $\widetilde{\MM}$ be the produced mechanism, $\widetilde{\opt}$ be the optimal revenue achievable by any BIC and IR mechanism w.r.t. $\widetilde{\DD}$. Clearly, Theorem~\ref{thm:rev-max} guarantees that $\rev(\widetilde{\MM},\widetilde{\DD})\geq \alpha\cdot \widetilde{\opt}-\varepsilon$. According to Corollary~\ref{cor:opt is close with small wass distance is }, $$|\opt-\widetilde{\opt}|\leq O\left(n\sqrt{\varepsilon'}\right).$$

We set $\theta= \varepsilon$ and $\varepsilon' = \varepsilon$, 
 we apply Theorem~\ref{thm:general reduction_2}  to $\widetilde{\MM}$, that is, the replicas are sampled from $\DD$ and the surrogates are sampled from $\widetilde{\DD}$. Let $\MM$ be the constructed mechanism, and Theorem~\ref{thm:general reduction_2} guarantees that 
 $$\rev(\MM,\DD) \geq \rev(\widetilde{\MM},\widetilde{\DD})- O\left(n \sqrt{\varepsilon} \right)\geq  \alpha\cdot \widetilde{\opt} - O\left(n \sqrt{\varepsilon} \right)\geq \alpha\cdot {\opt} - O\left(n \sqrt{\varepsilon} \right),$$ if $\widetilde{\MM}$ is a $\varepsilon$-BIC and IR mechanism. Otherwise, we know that $\widetilde{\MM}$ is a $1$-BIC and IR mechanism, and this happens with exponentially small probability according to Theorem~\ref{thm:rev-max}, so we can absorb the loss from this case in $O\left(n \sqrt{\varepsilon} \right)$. To sum up, if $d_{TV}(\dD_{i},\widetilde{\dD}_i) \leq \varepsilon', \forall i$, then $\rev(\MM,\DD) \geq \alpha\cdot \opt -O\left(n \sqrt{\varepsilon} \right)$.
 
With probability $\varepsilon$ we may get unlucky and $d_{TV}(\dD_{i},\widetilde{\dD}_i)$ may be larger than  $\varepsilon$ for some $i$. In that case we still construct $\MM$ in the same way, and we can apply Theorem~\ref{thm:general reduction_2}  by upper bounding $\wass(\DD,\widetilde{\DD})$ by $n$ and treating $\widetilde{\MM}$ as a $1$-BIC and IR mechanism, which shows $\rev(\MM,\DD)\geq -O(n)$.

Therefore, in expectation of the randomness of the samples used to estimate $\widetilde{\DD}$, $$\rev(\MM,\DD)\geq (1-\varepsilon)\cdot \left(\alpha\cdot \opt -O\left(n \sqrt{\varepsilon} \right)\right)-O(n\varepsilon)=\alpha\cdot \opt - O(n\sqrt{\varepsilon}),$$ as $\opt\leq n$. Note that even though mechanism $\MM$ depends on $\widetilde{\DD}$ and $\widetilde{\MM}$, it is always BIC and IR w.r.t. $\DD$. The time complexity follows from Theorem~\ref{thm:general reduction_2} and~\ref{thm:rev-max}.
\end{prevproof}
\section{General Outcome Space: Regularized Replica-Surrogate Fractional Assignment}\label{sec:fractional}

In this section we consider general outcome space $\mathcal{O}$, removing the assumption that $\mathcal{O}$ is downward-closed. We first demonstrate in Example~\ref{example:small_prob} that only having sample access to the agents' type distributions is not sufficient to transform an $\epsilon$-BIC mechanism to exactly BIC, with a small revenue loss.

\begin{example}\label{example:small_prob}

Consider the following $\epsilon$-BIC to BIC transformation instance for a single buyer.
Let $$\TT =\{t_H, t_{(l,1)},\ldots,t_{(l,m)},t_{(s,1)},\ldots,t_{(s,m)}\} \in \mathbb{N},$$ be the set of the buyer's types,
$\DD = \{\DD_{(i,j)}\}_{i, j \in [m]}$, where $\DD_{(i,j)}$ is the distribution with support $\{t_H, t_{(l,i)},t_{(s,j)}\}$ and point-mass probability:
\begin{align*}
    \Pr_{t' \sim D_{(i,j)}}[t'=t]=
    \begin{cases}
    1-2\sigma \quad & \text{if $t=t_H$} \\
    \sigma & \text{if $t=t_{(l,i)}$} \\
    \sigma & \text{if $t=t_{(s,j)}$}
    \end{cases}
\end{align*}

The output space is $\OO =\{o_H,o_{(l,1)},\ldots,o_{(l,m)},o_{(s,1)},\ldots,o_{(s,m)}\}$ such that:
\begin{align*}
    v(t_H,o)=&
    \begin{cases}
        1 \quad & \text{if $o=o_H$}\\
        0 & \text{o.w.}
    \end{cases}
    \\
    v(t_{(l,i)},o)=&
    \begin{cases}
        \epsilon \quad &\text{if $o=o_{(s,j')}$, for any $j' \in [m]$}\\
        0 &\text{o.w.}
    \end{cases}
    \\
    v(t_{(s,j)},o)=&
    \begin{cases}
    \epsilon \quad &\text{if $o=o_{(l,i')}$, for any $i' \in [m]$}\\
    0 &\text{o.w.}
    \end{cases}
\end{align*}

For each distribution $\DD_{(i,j)}\in \DD$ we consider the mechanism $\MM_{(i,j)}$:
if the buyer reports $ t_H$, the mechanism outputs $o_H$ and charges the buyer $1$,
if the buyer reports $t_{(l,i)}$, the mechanism outputs $o_{(l,i)}$ and charges the buyer $0$,
if the buyer reports $t_{(s,j)}$, the mechanism outputs $o_{(s,j)}$ and charges the buyer $0$.
If the buyer reports anything else, the mechanism outputs $o_H$ and charges the buyer $1$. Mechanism $\MM_{(i,j)}$ is $\epsilon$-BIC and IR for distribution $D_{(i,j)}$. The revenue of $\MM$ is $1-2\sigma$.
\end{example}

The mechanism designer has oracle access to set $\TT$ and faces the following problem:
There is an arbitrary distribution from the whole space $\DD$, which is unknown to the designer. The buyer's valuation distribution is realized to be some $\DD_{(i,j)} \in \DD$,
the designer is given sample access to $\DD_{(i,j)}$, oracle access to $\MM_{(i,j)}$\footnote{In other words the mechanism designer does not know the outcome space $\mathcal{O}$. It can only put a input type into the mechanism and outputs the returned outcome.} and she needs to output a truly BIC and IR mechanism $\MM'$ w.r.t. the buyer's realized distribution that approximately preserves the revenue. 
Note that $|\TT| = 2 m +1$. Therefore the size of each $t \in \TT$ is at most $O(\log(m))$.

In Lemma~\ref{lem:general-example}, we prove that with proper choice of $\sigma$, any BIC and IR mechanism with $poly(\log(m))$ samples from the type distribution $\DD_{(i,j)}$ and queries from $\MM_{(i,j)}$ has revenue far less than $\MM_{(i,j)}$. 

\begin{lemma}\label{lem:general-example}
For any $\delta>0$ choose $\sigma<\frac{\delta}{2N}$ and $m\geq \frac{2N}{\delta}$.
For any $i,j$, let $\MM'$ be any BIC, IR mechanism w.r.t. $\DD_{(i,j)}$, which only uses $N=\poly(\log(m))$ samples from $\DD_{(i,j)}$ and queries from $\MM_{(i,j)}$ (notice that $\log(m)$ is the size of the input). Then $\rev(\MM')<2\delta+3\delta\cdot\epsilon$. As $\delta$ goes to 0, $\rev(\MM')$ goes to 0 while $\rev(\MM)=1-2\sigma\geq 1-\frac{\delta}{N}$ goes to 1.
\end{lemma}

\begin{proof}
Denote $\MM'(t)$ the output of mechanism $\MM'$ if the agent reports type $t$ and $p(t)$ the expected payment if the agent reports $t$. Notice that the outcome of $\MM_{(i,j)}$ can only be $o_H, o_{(l,i)}$, or $o_{(s,j)}$. Thus the output of mechanism $\MM'$ is also one of these three outcomes. Let $\MM'(t)$ for inputs $t_{(l,i)}, t_{(s,j)}$ be the following distributions
\begin{align*}
    \MM'(t_{(l,i)})=&
    \begin{cases}
    o_H \quad & \text{with probability $1-p-q$} \\
    o_{(l,i)} & \text{with probability  $p$} \\
    o_{(s,j)} & \text{with probability  $q$} 
    \end{cases} \\
    \MM'(t_{(s,j)})=&
    \begin{cases}
    o_H \quad & \text{with probability $1-p'-q'$} \\
    o_{(s,j)} & \text{with probability  $p'$} \\
    o_{(l,i)} & \text{with probability  $q'$} 
    \end{cases}
\end{align*}

Then by BIC constraint we have that:
\begin{align*}
    q \cdot \epsilon - p(t_{(l,i)}) \geq p' \cdot \epsilon - p(t_{(s,j)}) & \quad \text{By BIC constraint if agent's true type is $t_{(l,i)}$ but reports $t_{(s,j)}$}. \\
    q' \cdot \epsilon - p(t_{(s,j)}) \geq p \cdot \epsilon - p(t_{(l,i)}) & \quad \text{By BIC constraint if agent's true type is $t_{(s,j)}$ but reports $t_{(l,i)}$}.
\end{align*}

By the previous inequalities we can infer that:

\begin{align*}
    q-p' \geq \frac{p(t_{(l,i)}) - p(t_{(s,j)}) }{\epsilon} \geq p - q'
\end{align*}

Notice that either $p \leq q'$ or $p' \leq q$.
The claim holds because if $p(t_{(l,i)}) - p(t_{(s,j)})$ is non-negative, then $q \geq p'$ and if $p(t_{(l,i)}) - p(t_{(s,j)})$ is negative we have that $p\leq q'$. Without loss of generality assume $p'\leq q$.

Now we are going to bound $q$. 
Note that since the mechanism designer does not have access to the outcome space, the only way to find outcome $o_{(s,j)}$ is to use $t_{(s,j)}$ as input to the mechanism $\MM$. To do that, the designer must sample $t_{(s,j)}$ either directly from the type space $\TT$ or the realized distribution $\DD_{(i,j)}$\footnote{Note that if the designer does not sample uniformly at random from the set $\{t_{(s,j)}\}_{j \in [m]}$, since the distribution over $\DD$ is unknown, the adversary can make it more difficult for her to sample the right type by picking a different distribution}.

With at most $N$ samples from the buyer's distribution $\DD_{(i,j)}$, with probability at least $(1-\sigma)^N\geq 1-N\sigma$, none of the sampled types is $t_{(s,j)}$.
By assumption, $\sigma \leq \frac{\delta}{2N}$ which implies that with probability at least $1-\frac{\delta}{2}$ we are not be able to sample type $t_{(s,j)}$,
therefore we cannot find output $o_{(s,j)}$ by this method.

Moreover, notice that the output of $\MM$ is $o_H$ for all types $\{t_{s,j'}\}_{j \neq j'\in [m]}$. Thus by querying mechanism $\MM$ with at most $N$ types $t_{s,j'}$ chosen uniformly at random from $\{t_{(s,j)}\}_{j \in [m]}$, with probability at least $(1-1/m)^N\geq 1-\frac{N}{m}$ , none of the returned outcomes is $o_{(s,j)}$.
By assumption $m \geq \frac{2N}{\delta}$ implies that with probability at least $1-\frac{\delta}{2}$ none of the returned outcomes is $o_{(s,j)}$.

Hence, for any $\delta>0$,
by taking Union Bound over the two events described above,
we cannot identify output $o_{(s,j)}$ with probability at least $1- \delta$.
This implies $p'\leq \delta$. Similarly we can prove that $q'\leq \delta$.


By IR constraints for type $t_{(s,j)}$, we have $q' \cdot \epsilon - p(t_{(s,j)}) \geq 0$. Thus $-p(t_{(s,j)}) \geq -q' \cdot \epsilon$. By $u(t_H)$ we denote the utility of type $t_H$ in mechanism $\MM'$ when she reports $t_H$.
Consider the BIC constraint for type $t_H$ when she reports type $t_{(s,j)}$. We have
$$ u(t_H) \geq (1 - p' - q') - p(t_{(s,j)})
                \geq 1 - p' - q' - q' \cdot \epsilon
\geq 1-2\delta-\delta\cdot\epsilon$$

Thus $p(t_H)\leq 2\delta+\delta\cdot\epsilon$. Notice that both $p(t_{(s,j)})$ and $p(t_{(l,i)})$ are at most $\epsilon$ since the agent's value under these two types are at most $\epsilon$. Thus $\rev(\MM')\leq 2\delta+3\delta\cdot\epsilon$.
\end{proof}

Throughout this section, we will assume that we have full access to the agents' type distributions: Assume every agent's type distribution is discrete with finite support. We denote the support of the distribution of the $k$-th agent as $\{t_k^{(i)}\}_{i\in[m_k]}$ and the probability that the $k$-th agent has type $t_k^{(i)}$ as $F_i^k$.  When there is no confusion about the agent that we are referring to,
we may drop the superscript $k$ in $F_i^k$. The constructed mechanism knows $t_k^{(i)}$ and the exact value $F_i^k$ for every $i\in [m_k]$.

We first present the result in the ideal model, where the edge weights are known exactly. Then we show that if we only estimate the edge weights approximately, we can use an approach similar to \cite{hartline2015bayesian} to still guarantee the mechanism is BIC and IR.

\subsection{Regularized Replica-Surrogate Fractional Assignment}\label{sec:appx_rrsf}

We now present the \textbf{Regularized Replica–Surrogate Fractional Assignment (RRSF)} introduced by \cite{hartline2015bayesian}. There are two main differences between RRSF and the replica-surrogate matching: (i) the replicas and surrogates in RRSF are no longer samples from the agent's type distribution; for each type of the agent, there is exactly one replica and one surrogate of that type in RRSF; (ii) RRSF finds the optimal fractional assignment rather than a maximum weight matching.

\begin{definition}[Definition~4.6 of \cite{hartline2015bayesian}]\label{def:rrsf mechanism}
Let $\phi(\mathbf{x})=\frac{1}{2}\gamma||\mathbf{x}||_2^2$.
For agent $k\in [n]$ and $i,j\in [m_k]$, let $W_{ij}^k$ be any value in $[-1,1]$. We may drop the superscript when agent $k$ is fixed and clear from context. Consider the following convex program $(P^3)$ with coefficients $\{W_{i,j}\}_{i,j\in [m_k]}$, where the decision variables are $\{q_{i,j}\}_{i, j \in [m_k]}$.
\vspace{-.05in}
\begin{equation}\label{MP:regularized frac}
\begin{split}
\max &~~\textstyle\sum_{i,j\in [m_k]} F_i \left(W_{i,j} q_{i,j} - \phi(\bm{q}_i)\right) \\
\text{subject to} &~~\textstyle \sum_j q_{i,j} = 1,~~~~~\forall i\in [m_k]\\
                                         &~~\textstyle \sum_i F_i q_{i,j}=F_j,~~~~~\forall j\in [m_k]\\
                                          &~~q_{i,j} \geq 0, ~~~~\forall i,j\in [m_k]\end{split}
\end{equation}
Let $\bm{q^*_k}=\{q^*_{i,j}\}_{i,j\in[m_k]}$ be an optimal solution to the primal problem and $\bm{\lambda^*_k}=\{\lambda_{i,j}^*\}_{i,j\in [m_k]}, \bm{\mu^*_k}=\{\mu^*_i\}_{i\in [m_k]}, \bm{\pi_k^*}=\{\pi^*_j\}_{j\in [m_k]}$ be any Lagrange multipliers satisfying the KKT condition
\begin{enumerate}
    \item $F_i\left(W_{i,j} - \frac{\partial \phi(\bm{q}^*_i)}{\partial q^*_{i,j}}\right) = \lambda^*_{i,j} + \mu_i^* + F_i\pi_j^*, ~~ \forall i,j$
    \item $\lambda^*_{i,j} \leq 0, ~~ \forall i,j$
    \item $\lambda^*_{i,j} q^*_{i,j} = 0, ~~ \forall i,j$
\end{enumerate}
\end{definition}



\subsection{Ideal Model}\label{sec:ideal model RRSF}

\paragraph{RRSF Mechanism with parameters $\left(\bm{q^*}=\{\bm{q^*_k}\}_{k\in[n]}, \bm{\mu^*}=\{\bm{\mu^*_k}\}_{k\in[n]}, \bm{\pi^*}=\{\bm{\pi^*_k}\}_{k\in[n]}\right)$:} every agent $k$ reports her type $t_k=t_k^{(i)}$. Let $s_k$ be a random surrogate type such that $\Pr[s_k=t_k^{(j)}]=q^*_{i,j}$ for every $j\in [m_k]$. 

The RRSF mechanism $\MM'$ runs mechanism $\MM$ under the random input $s=(s_1,\ldots, s_n)$. The outcome for $\MM'$ is the random outcome $o$ generated from $x(s)$ and agent $k$'s expected payment is $p_k(s)$ plus some extra payment $\hat{p}_k$, where $\hat{p}_k(t_k) = \sum_{j} \pi_j^* q^*_{i,j} + \phi(\bm{q}^*_i) - \phi(\bm{0}) + \min_\ell \frac{\mu_\ell^*}{F_\ell}$. We call $\MM'$ the RRSF mechanism with respect to $\textbf{W}=\{W_{i,j}^k\}_{k\in [n],i,j\in [m_k]}$ if for every $k$, the parameters $(\bm{q^*_k},\bm{\mu^*_k},\bm{\pi^*_k})$ are obtained by the convex program $(P^3)$ with coefficients $\{W_{i,j}^k\}_{i,j\in [m_k]}$.



Similar to~\cite{hartline2015bayesian}, in our proof each $W_{i,j}^k$ will be chosen to be the expected utility of agent $k$ for the outcome of the mechanism $\MM=(x,p)$ when her type is $t_k^{(i)}$ and she is matched to the surrogate of type $t_k^{(j)}$. Formally,
$$W_{i,j}^k= \E_{t_{-k}\sim \DD_{-k}}\left[v_k(t_k^{(i)},x(t_k^{(j)}, t_{-k}))\right]-\E_{t_{-k}\sim \DD_{-k}}\left[p_k(t_k^{(j)},t_{-k})\right]$$
They proved that the RRSF mechanism with respect to the $\textbf{W}$ defined above is BIC and IR.

\begin{theorem}[\cite{hartline2015bayesian}]\label{thm:RRSF is BIC and IR}
Given any mechanism $\MM$. For every agent $k$ and $i,j\in [m_k]$, let
$$W_{i,j}^k= \E_{t_{-k}\sim \DD_{-k}}\left[v_k(t_k^{(i)},x(t_k^{(j)}, t_{-k}))\right]-\E_{t_{-k}\sim \DD_{-k}}\left[p_k(t_k^{(j)},t_{-k})\right]$$

Then RRSF mechanism $\MM'$ defined in Definition~\ref{def:rrsf mechanism} is BIC and IR.
\end{theorem} 

As Hartline et al.~\cite{hartline2015bayesian} only care about the welfare of $\MM'$, Theorem~\ref{thm:RRSF is BIC and IR} suffices. We care about the revenue of $\MM'$, so we need to argue that running $\MM'$ does not cause the designer to lose a large fraction of the revenue. $\rev(\MM')$ contains two parts: (i) the expected revenue from payments $\{p_k(\cdot)\}_{k\in [n]}$, which is exactly the same as $\rev(\MM)$; (ii) the expected payments from $\{\hat{p}_k(\cdot)\}_{k\in [n]}$. To prove that $\rev(\MM')$ is not too much smaller than $\rev(\MM)$, we need to prove that the expected payments from $\{\hat{p}_k(\cdot)\}_{k\in [n]}$ are not too negative. We prove this claim with the following sequence of lemmas.

We will prove a more general result for any RRSF mechanism with respect to $\textbf{W}$ that satisfies: $W_{i,i}^k\geq \max\{W_{i,j}^k-\varepsilon',-\varepsilon''\},\forall k\in[n],i,j\in [m_k]$, for some $\varepsilon',\varepsilon''\geq 0$. Note that when $\textbf{W}$ is chosen as in Theorem~\ref{thm:RRSF is BIC and IR} and $\MM$ is $\varepsilon$-BIC and IR, we will have $\varepsilon'=\varepsilon$ and $\varepsilon''=0$. The more general result is useful in the proof of the non-ideal model in Section~\ref{sec:non-ideal model RRSF}.

We first prove a structural result about the optimal solution of the convex program $(P^3)$. 

\begin{lemma}\label{lem: cycles}
Fix any agent $k$. Suppose for all $i,j\in [m_k]$, $W_{i,i}^k\geq \max\{W_{i,j}^k-\varepsilon',-\varepsilon''\}$ holds for some $\varepsilon',\varepsilon''\geq 0$. Let $\{q^*_{i,j}\}_{i,j\in[m_k]}$ be an optimal solution of the convex program $(P^3)$. Then, it holds that  $$q^*_{i,j}>0 \implies W_{i,j} \geq W_{i,i} - m\varepsilon' -  \sqrt{2}m\gamma,~~~~~~\forall i, j \in[m_k].$$
\end{lemma}

\begin{proof}
For any type $t_k^{(i)}\in \supp(\DD_k)$ such that $q^*_{i,i}=1$ the statement holds.

Now assume there are some types $t_k^{(i^*)} \neq t_k^{(j^*)}$ such that $q_{i^*,j^*} >0$.
A sequence of types $S=\{ t_k^{(a_i)} \}_{i=1}^b$ of length $b$, where $t_k^{(a_i)} \in \supp(\DD_k)$,
is called a flow sequence if for each $1\leq i \leq b-1$ it holds that $q_{a_i,a_{i+1}}>0$.
Let $V_{j^*}$ be the set of types (including $t_k^{(j^*)}$) such that there exists some flow sequence that starts with type $t_k^{(j^*)}$ that reaches them.
We are going to prove that $t_k^{(i^*)} \in V_{j^*}$.

Note that for each $t_k^{(r)} \in V_{j^*}$, if $q_{r,\ell}>0$ then $t_k^{(\ell)} \in V_{j^*}$. Thus the set $U_{j^*}=\{t_k^{(\ell)}: \exists t_k^{(r)}\in V_{j^*}\wedge q_{r,\ell}>0\}\subseteq V_{j^*}$. Since $\mathbf{q}^*$ is feasible, we have
$$\sum_{t_k^{(r)}\in V_{j^*}}F_r=\sum_{t_k^{(r)}\in V_{j^*}}\sum_{\ell\in[m_k]} F_rq_{r,l}=\sum_{t_k^{(r)}\in V_{j^*}}\sum_{\ell: q_{r,\ell}>0} F_r q_{r,l}=\sum_{t_k^{(\ell)}\in U_{j^*}}\sum_{t_k^{(r)}\in V_{j^*}}F_r q_{r,\ell}.$$

On the other hand,
$$\sum_{t_k^{(\ell)}\in U_{j^*}}\sum_{t_k^{(r)}\in V_{j^*}}F_r q_{r,\ell}\leq \sum_{t_k^{(\ell)}\in V_{j^*}} \sum_{t_k^{(r)}\in V_{j^*}}F_r q_{r,\ell}\leq \sum_{t_k^{(\ell)}\in V_{j^*}}\sum_{r\in [m_k]}F_rq_{r,\ell}=\sum_{t_k^{(\ell)}\in V_{j^*}}F_\ell.$$

Therefore, we must have $U_{j^*}=V_{j^*}$ and for every $t_{k}^{(\ell)}\in V_{j^*}, t_k^{(r)}\not\in V_{j^*}$, $q_{r,l}=0$. Since $q_{i^*,j^*} >0$, we have $t_k^{(i^*)}\in V_{j^*}$. Let  $S_{i^*,j^*}=\{t_k^{a_1},\ldots,t_k^{a_{b}}\}$ such that $a_1=j^*$ and $a_b=i^*$ be the shortest flow sequence from $t_k^{(j^*)}$ to $t_k^{(i^*)}$. Then each type appears in $S_{i^*,j^*}$ at most once.
Let $f_* = \min_{k \in [b]} F_{a_k} q^*_{a_k,a_{k+1}}\geq 0$ (here $a_{b+1}=a_1$). We consider the following new solution $\mathbf{\hat{q}}$:

\begin{align*}
    \hat{q}_{i,j} =
    \begin{cases}
    q^*_{i,j} - f_*/F_i, \quad &\text{ when there exists $\ell$ s.t. $a_\ell=i$ and $a_{\ell+1}=j$}\\
    q^*_{i,i} + f_*/F_i &\text{ when $i=j$ and there exists $\ell$ s.t. $a_\ell=i$}\\
    q^*_{i,j} &\text{ o.w.}
    \end{cases}
\end{align*}

We will verify that $\mathbf{\hat{q}}$ is a feasible solution to RRSF. For the first set of constraints: for every $a_\ell$, notice that (i) $\hat{q}_{a_\ell,a_\ell}=q^*_{a_\ell,a_\ell}+f_*/F_{a_\ell}$ and (ii) there exists a unique $j_0$ such that $\hat{q}_{a_\ell,j_0}=q^*_{a_\ell,j_0} - f_*/F_{a_\ell}$. Thus $\sum_j\hat{q}_{i,j}=\sum_j q^*_{i,j}=1$. For the second set of constraints: for every $j=a_{\ell+1}$ for some $\ell\in [b]$, notice that there exists a unique $i_0=a_\ell$ such that $\hat{q}_{i_0,j}=q^*_{i_0,j} - f_*/F_{i_0}$, we have $$\sum_iF_i\hat{q}_{i,j}=\sum_iF_iq^*_{i,j}-F_{i_0}\cdot \frac{f_*}{F_{i_0}}+F_j\cdot \frac{f_*}{F_j}=\sum_iF_iq^*_{i,j}=F_j$$   
For the third set of constraints: for every pair $(i, j)$, either $\hat{q}_{i,j}\geq q^*_{i,j}\geq 0$ (the second and third case), or $\hat{q}_{i,j}=q^*_{i,j} - \frac{f_*}{F_{i}}\geq 0$ as $f_*\leq F_iq^*_{i,j}$.

Therefore $\mathbf{\hat{q}}$ is a feasible solution to RRSF.
Moreover since $\bm{q}^*$ is the optimal feasible solution,
its objective should be at least the objective value of $\mathbf{\hat{q}}$. Observe that for every $i\in [m_k]$ and every edge $(i,j)$ with $j\neq i$ the value $\hat{q}_{i,j}$ is not equal to $q^*_{i,j}$ only when there is some $\ell$ with $a_\ell = i, a_{\ell+1} = j$. The value of this edge drops by $f_*/F_i$, whereas the value of $q_{i,i}$ increases by $f_*/F_i$.
Denote $V(\bm{q}^*), V(\bm{\hat{q}})$ the value of the objectives for these two feasible solutions respectively. In the flow sequence $S=S_{i^*,j^*}$, for every $i=a_\ell$ for some $\ell\in [b]$, denote $i'=a_{\ell+1}$. We have

\begin{align*}
    0 \leq& V(\bm{q}^*) -  V(\bm{\hat{q}})\\ 
    = &\sum_{i} F_i (\sum_j W_{i,j} (q^*_{i,j} - \hat{q}_{i,j}) + (\phi(\bm{\hat{q}}_i) - \phi(\bm{q}^*_i)) \\
    = & \sum_{i: t_k^{(i)} \in S_{i^*,j^*}} F_i \left(W_{i,i'} \frac{f_*}{F_i} - W_{i,i}\frac{f_*}{F_i} +  (\phi(\hat{\bm{q}}_i) - \phi(\bm{q}^*_i)) \right)\qquad(\hat{\bm{q}}_i=\bm{q}^*_i,\text{if }t_k^{(i)} \notin S_{i^*,j^*}) \\
    = & f_*\sum_{i: t_k^{(i)} \in S_{i^*,j^*}}  \left(W_{i,i'} - W_{i,i} +  \frac{F_i}{f_*}(\phi(\hat{\bm{q}}_i) - \phi(\bm{q}^*_i)) \right) \\
    \leq & f_*\sum_{i: t_k^{(i)} \in S_{i^*,j^*}}  \left(W_{i,i'} - W_{i,i} +  \frac{F_i}{f_*}\nabla  \phi(\hat{\bm{q}}_i)^T(  \hat{\bm{q}}_i-{\bm{q}^*_i}) \right)\qquad(\text{$\phi(\cdot)$ is convex}) \\
    \leq & f_*\sum_{i: t_k^{(i)} \in S_{i^*,j^*}}  \left(W_{i,i'} - W_{i,i} +  \frac{F_i}{f_*} ||\nabla  \phi(\hat{\bm{q}}_i)||_2 ||{\bm{q}}^*_i - \hat{\bm{q}}_i||_2 \right)\qquad(\text{Cauchy-Schwarz Inequality}) \\
    = & f_*\sum_{i: t_k^{(i)} \in S_{i^*,j^*}}  \left(W_{i,i'} - W_{i,i} +  \sqrt{2} ||\nabla  \phi(\hat{\bm{q}}_i)||_2  \right)\qquad(||{\bm{q}}^*_i - \hat{\bm{q}}_i||_2^2=2\cdot \left(\frac{f_*}{F_i}\right)^2)\\
    = & f_*\sum_{i: t_k^{(i)} \in S_{i^*,j^*}}  \left(W_{i,i'} - W_{i,i} +  \sqrt{2} \gamma||\hat{\bm{q}}_i||_2  \right) \qquad(\text{the definition of $\phi(\cdot)$})\\
    \leq & f_*\sum_{i: t_k^{(i)} \in S_{i^*,j^*}}  \left(W_{i,i'} - W_{i,i} +  \sqrt{2} \gamma  \right)\qquad(\hat{\bm{q}}_i\text{ is a distribution}). \\
\end{align*}
Since $f_*\geq 0$ we have $$\sum_{i: t_k^{(i)} \in S_{i^*,j^*}}  \left(W_{i,i'} - W_{i,i} +  \sqrt{2} \gamma  \right) \geq 0.$$ Thus
\begin{align*}
W_{i^*,j^*} \geq & W_{i^*,i^*} - \sqrt{2}\gamma -\sum_{i: t_k^{(i)} \in S_{i^*,j^*}-\left\{t_k^{(i^*)}\right\}}  \left(W_{i,i'} - W_{i,i} +  \sqrt{2} \gamma  \right) \\
    \geq & W_{i^*,i^*} - \sqrt{2}\gamma -\sum_{i: t_k^{(i)} \in S_{i^*,j^*}-\left\{t_k^{(i^*)}\right\}} \left(\varepsilon' +  \sqrt{2} \gamma  \right) \qquad(\text{$W_{i,i} \geq W_{i,i'} - \varepsilon'$ by assumption})\\
    \geq & W_{i^*,i^*} - m\varepsilon' -  \sqrt{2}m\gamma \\
\end{align*}
\end{proof}

Equipped with Lemma~\ref{lem: cycles}, we proceed to show that the convex program $(P^3)$ has a set of optimal Lagrange multipliers that will make the total expected payments from $\{\hat{p}_k\}_{k\in[n]}$ not too negative. First, we construct a new convex program $(P^4)$ that ``removes'' all the very negative edge weights from $(P^3)$, that is, change every edge $(i,j)$'s weight to  $\max\{W_{i,j},-m(\varepsilon'+2\gamma)-\varepsilon''\}$ (see Definition~\ref{def:hat_W}). Then, we argue that convex program $(P^3)$ and $(P^4)$ are essentially equivalent. Namely, any optimal solution of $(P^4)$ is also an optimal solution for $(P^3)$, and there is a straightforward mapping that transforms any optimal Lagrange multipliers of $(P^4)$ to a set of optimal Lagrange multipliers for $(P^3)$ (see Lemma~\ref{lem:P vs P'}). Finally, we show that if we compute $\hat{p}_k(\cdot)$ using any set of the optimal Lagrange multipliers derived from convex program $(P^4)$, then $\hat{p}_k(\cdot)$ is not too negative (see Lemma~\ref{lem:lower bound hat-p}) The main reason is that using the KKT condition, we can relate $\hat{p}_k(\cdot)$ to the edge weights, and the edge weights in $(P^4)$ are not too small.

\begin{definition}\label{def:hat_W}
Fix an agent $k$ and $\varepsilon'\geq\varepsilon''\geq 0$. For every $i,j$, define

\begin{align*}
    \hat{W}_{i,j}=
    \begin{cases}
    W_{i,j} \quad & \text{if $W_{i,j} \geq -m(\varepsilon'+2\gamma)-\epsilon''$} \\
    -m(\varepsilon'+2\gamma)-\epsilon'' & \text{if $W_{i,j} < -m(\varepsilon'+2\gamma)-\epsilon''$} \\
    \end{cases}
\end{align*}

We solve the following convex program $(P^4)$ where the decision variables are $\{q_{i,j}\}_{i, j\in [m_k]}$.
\vspace{-.05in}
\begin{equation}\label{MP:regularized frac rel}
\begin{split}
\max &~~\textstyle\sum_{i,j} F_i \left(\hat{W}_{i,j} q_{i,j} - \phi(\bm{q}_i)\right) \\
\text{subject to} &~~\textstyle \sum_j q_{i,j} = 1,~~~~~~~~~~\forall i\\
                                         &~~\textstyle \sum_i F_i q_{i,j}=F_j,~~~~~~\forall j\\
                                          &~~q_{i,j} \geq 0, ~~~~~~\forall i,j\end{split}
\end{equation}
Let $\bm{\hat{q}^*}$ be an optimal solution of $(P^4)$ and $\bm{\hat{{\lambda}}^*}, \bm{\hat{\mu}^*}, \bm{\hat{\pi}^*}$ be any Lagrange multipliers satisfying the KKT conditions:

\begin{enumerate}
    \item $F_i\left(\hat{W}_{i,j} - \frac{\partial \phi(\hat{\bm{q}}^*_i)}{\partial q^*_{i,j}}\right) = \hat{\lambda}^*_{i,j} + \hat{\mu}_i^* + F_i\hat{\pi}_j^*, ~~ \forall i,j$
    \item $\hat{\lambda}^*_{i,j} \leq 0, ~~ \forall i,j$
    \item $\hat{\lambda}^*_{i,j} \hat{q}^*_{i,j} = 0, ~~ \forall i,j$
\end{enumerate}
\end{definition}

\begin{lemma}\label{lem:P vs P'}
Fix any agent $k$. Suppose for all $i,j\in [m_k]$, $W_{i,i}^k\geq \max\{W_{i,j}^k-\varepsilon',-\varepsilon''\}$ holds for some $\varepsilon'\geq\varepsilon''\geq 0$.
Let $\bm{\hat{q}^*}=\{\hat{q}^*_{i,j}\}_{i,j\in [m_k]}$ be any optimal solution of the convex program $(P^4)$. Then $\bm{\hat{q}^*}$ is also an optimal solution of the convex program $(P^3)$. Moreover let $\bm{\hat{\lambda}^*}$, $\bm{\hat{\mu}^*}$, $\bm{\hat{\pi}^*}$ be any Lagrange multipliers satisfying the KKT conditions w.r.t. the convex program $(P^4)$. Then there exists $\bm{\lambda^*}=\{\lambda^*_{i,j}\}_{i, j\in[m_k]}$ such that $(\bm{\lambda^*}$,$\bm{\mu^*}=\bm{\hat{\mu}^*}$, $\bm{\pi^*}=\bm{\hat{\pi}^*})$ satisfies the KKT conditions w.r.t. the convex program $(P^3)$.
\end{lemma}

\begin{proof}
We first prove that $\bm{\hat{q}^*}$ is also an optimal solution of the convex program $(P^3)$. Clearly $\bm{\hat{q}^*}$ is a feasible solution of the convex program $(P^3)$. To argue $\bm{\hat{q}^*}$ is an optimal solution, we first prove that $\hat{W}_{i,i} \geq \hat{W}_{i,j} - \varepsilon', \forall i,j \in [m_k]$. For every $i$, Since $W_{i,i} \geq -\varepsilon''$, $\hat{W}_{i,i}=W_{i,i} \geq -\varepsilon''$. If $W_{i,j}\geq -m(\varepsilon'+2\gamma)-\epsilon''$, then $\hat{W}_{i,j}=W_{i,j}$ and we have $\hat{W}_{i,i} \geq \hat{W}_{i,j} - \varepsilon'$ since $W_{i,i} \geq W_{i,j} - \varepsilon'$. On the other hand, suppose $W_{i,j} < -m(\varepsilon'+2\gamma)-\epsilon''$. Then $\hat{W}_{i,j}=-m(\varepsilon'+2\gamma)-\epsilon''\leq 0$. $\hat{W}_{i,i}-\hat{W}_{i,j}\geq W_{i,i}\geq -\varepsilon''\geq -\varepsilon'$.

We apply Lemma~\ref{lem: cycles} to convex program $(P^4)$. For any $i,j$ such that $\hat{q}^*_{i,j}>0$, $\hat{W}_{i,j} \geq \hat{W}_{i,i} -m(\varepsilon'+\sqrt{2}\gamma)\geq -m(\varepsilon'+\sqrt{2}\gamma)-\varepsilon''$. Thus it must hold that $\hat{W}_{i,j}=W_{i,j}$ and the solution $\bm{\hat{q}^*}$ has the same objective value in $(P^3)$ and $(P^4)$. Moreover since $\hat{W}_{i,j}\geq W_{i,j}$ for every $i,j$, the optimal objective value of $(P^3)$ is no larger than the optimal objective value of $(P^4)$. Thus $\bm{\hat{q}^*}$ is an optimal solution of the convex program $(P^3)$.

For the second part of the statement, consider the following dual variables $\{\lambda^*_{i,j}\}_{i,j\in[m_k]}$:

\begin{align*}
    \lambda^*_{i,j} =
    \begin{cases}
    \hat{\lambda}^*_{i,j} & \text{if $W_{i,j} \geq -m(\varepsilon'+2\gamma)-\epsilon''$} \\
    \hat{\lambda}^*_{i,j} + F_i(W_{i,j} +m(\varepsilon'+2\gamma)+\epsilon'') \quad & \text{o.w.}
    \end{cases}
\end{align*}

We now verify that $(\bm{\lambda^*},\bm{\mu^*=\hat{\mu}^*},\bm{ \pi^*=\hat{\pi}^*})$ satisfies the KKT conditions w.r.t. the convex program $(P^3)$.
When $\hat{W}_{i,j} \geq -m(\varepsilon'+2\gamma)-\epsilon''$ we have that $\lambda^*_{i,j} = \hat{\lambda}^*_{i,j}$ and $\hat{W}_{i,i}=W_{i,i}$.
This implies that:

\begin{enumerate}
    \item $F_i\left(\hat{W}_{i,j} - \frac{\partial \phi(\hat{\bm{q}}^*_i)}{\partial \hat{q}^*_{i,j}}\right) = \hat{\lambda}^*_{i,j} + \hat{\mu}_i^* + F_i\hat{\pi}_j^* \Leftrightarrow F_i\left(W_{i,j} - \frac{\partial \phi(\hat{\bm{q}}^*_i)}{\partial \hat{q}^*_{i,j}}\right) = \lambda^*_{i,j} + \hat{\mu}_i^* + F_i\hat{\pi}_j^*$
    \item $\hat{\lambda}^*_{i,j} \leq 0$
    \item $\hat{\lambda}^*_{i,j} \hat{q}^*_{i,j} = 0$
\end{enumerate}

Now consider $i,j$ such that $W_{i,j} <-m(\varepsilon'+2\gamma)-\epsilon''$. In this case $\hat{W}_{i,j} = -m(\varepsilon'+2\gamma)-\epsilon''$. By Lemma~\ref{lem: cycles}, if $\hat{q}^*_{i,j} > 0$,
then $\hat{W}_{i,j} \geq \hat{W}_{i,i} -m(\varepsilon'+\sqrt{2}\gamma)\geq -m(\varepsilon'+\sqrt{2}\gamma)-\epsilon''$. This contradicts with the fact that $W_{i,j} <-m(\varepsilon'+2\gamma)-\epsilon''$, so $\hat{q}^*_{i,j}=0$. We have

\begin{enumerate}
    \item \begin{align*}& F_i\left(\hat{W}_{i,j} - \frac{\partial \phi(\hat{\bm{q}}^*_i)}{\partial \hat{q}^*_{i,j}}\right) = \hat{\lambda}^*_{i,j} + \hat{\mu}_i^* + F_i\hat{\pi}_j^* \\
    \Leftrightarrow &
    F_i\left(-m(\varepsilon'+2\gamma)-\epsilon''  -\frac{\partial \phi(\hat{\bm{q}}^*_i)}{\partial \hat{q}^*_{i,j}}\right) 
    = \lambda^*_{i,j} - F_i(W_{i,j}+m(\varepsilon'+2\gamma)+\epsilon'') + \hat{\mu}_i^* + F_i\hat{\pi}_j^* \\
    \Leftrightarrow &
    F_i\left( W_{i,j} -\frac{\partial \phi(\hat{\bm{q}}^*_i)}{\partial \hat{q}^*_{i,j}}\right) 
    = {\lambda}^*_{i,j}  + \hat{\mu}_i^* + F_i\hat{\pi}_j^*
    \end{align*}
    \item $\hat{\lambda}^*_{i,j} \leq 0$
    \item $\hat{\lambda}^*_{i,j} \hat{q}^*_{i,j} = 0$
\end{enumerate}


Thus $(\bm{\lambda^*},\bm{\mu^*=\hat{\mu}^*}, \bm{\pi^*=\hat{\pi}^*})$ satisfies the KKT conditions w.r.t. the convex program $(P^3)$.
\end{proof}

As illustrated in the following example an arbitrary set of optimal dual variables that satisfy the KKT conditions can cause a big revenue loss.

\begin{example}\label{example:negative. payments}
We consider the following instance.
There is one agent with two possible types $t_L$ and $t_H$ such that the agent has type $t_H$ with probability $F_{t_H}=1-p$ and $t_L$ with probability $F_{t_L}=p$ for some sufficiently small $p>0$. The outcome space $O=\{o_L,o_H\}$. $v(t_L,o_L)=v(t_H,o_H)=1$ and $v(t_L,o_H)=v(t_H,o_L)=0$. 
The mechanism $\MM$ is defined as follows: it returns outcome $o_H$ with input $t_H$ and returns outcome $o_L$ with input $t_L$. The payment of the mechanism is always $\frac{1}{2}$. It's a BIC mechanism and it holds that $W_{t_L,t_L}=W_{t_H,t_H}=1/2$ and $W_{t_L,t_H}=W_{t_H,t_L}=-1/2$.

We are going to prove through KKT conditions that for this instance, the following $q^*$ is the optimal solution for Convex Program $(P^3)$: $q^*_{t_L,t_L}=q^*_{t_H,t_H}=1$ and $q^*_{t_H,t_L}=q^*_{t_H,t_L}=0$. Consider the following set of dual variables: $\mu^*_{t_L} = p(-1+\gamma)$, $\mu^*_{t_H}=0$,
$\pi^*_{t_L}=3/2-2\gamma$, $\pi^*_{t_H}= 1/2- \gamma$ and $\lambda^*(t_H, t_L) = (1-p)(-2 + 2\gamma), \lambda^*(t_H,t_H)=\lambda^*(t_L, t_L) = \lambda^*(t_L, t_H) = 0$.

Note that $\gamma<1$, and $\frac{\partial \phi(\hat{\bm{q}}^*_{t_L})}{\partial q^*_{t_L,t_L}}=\frac{\partial\phi(\hat{\bm{q}}^*_{t_H})}{\partial q^*_{t_H,t_H}}=\gamma$, $\frac{\partial \phi(\hat{\bm{q}}^*_{t_L})}{\partial q^*_{t_L,t_H}}=\frac{\partial\phi(\hat{\bm{q}}^*_{t_H})}{\partial q^*_{t_H,t_L}}=0$. One can verify that the dual variables $(\mu^*,\pi^*,\lambda^*)$ satisfy the KKT conditions with respect to $q^*$, which implies that $q^*$ is an optimal solution for $(P^3)$.

Note that $\min\{\mu_{t_L}^*/F_{t_L},\mu_{t_H}^*/F_{t_H}\}=-1+\gamma$. According to the payment rule, the payment charged to the first agent if she has type $t_L$ is $3/2-2\gamma + \gamma - 0 +(-1 + \gamma) = 1/2$ and the payment charged if she has type $t_H$ is $1/2 - \gamma + \gamma - 0 + (-1 + \gamma) = -1/2 + \gamma$.

Therefore the expected revenue of the RRSF mechanism with parameter $(\mu^*,\pi^*,\lambda^*)$ is $p\cdot 1/2 + (1-p) \cdot (-1/2 + \gamma) = -1/2 + \gamma  + p ( 1- \gamma)$.
When both $p$ and $\gamma$ go to $0$, the expected revenue of the mechanism is close to $-1/2$, which is far from the revenue of $\MM$. 

\end{example}

\begin{lemma}\label{lem:lower bound hat-p}
Given any mechanism $\MM$. Suppose for every agent $k$ and every $i,j\in [m_k]$, $W_{i,i}^k\geq \max\{W_{i,j}^k-\varepsilon',-\varepsilon''\}$ holds for some $\varepsilon'\geq\varepsilon''\geq 0$. Let $\bm{\hat{q}}^*=\{\hat{q}^*_{i,j}\}_{i,j\in[m_k]}$ be any optimal solution of the convex program $(P^4)$ and $\bm{\hat{\lambda}^*}, \bm{\hat{\mu}^*}, \bm{\hat{\pi}^*}$ be any Lagrange multipliers satisfying the KKT conditions w.r.t. the convex program $(P^4)$. Let $(\bm{\lambda^*},\bm{\mu^*=\hat{\mu}^*}, \bm{\pi^*=\hat{\pi}^*})$ be the Lagrange multipliers stated in Lemma~\ref{lem:P vs P'}, which satisfy the KKT conditions w.r.t. the convex program $(P^3)$. Consider the RRSF mechanism $\MM'$ in Definition~\ref{def:rrsf mechanism} with optimal solution $\bm{q^*}=\bm{\hat{q}^*}$ and Lagrange multipliers $(\bm{\lambda^*},\bm{\mu^*=\hat{\mu}^*}, \bm{\pi^*=\hat{\pi}^*})$. Then for each agent $k$ and $i\in [m_k]$, $\hat{p}_k(t_k^{(i)}) = \sum_{j} \pi_j^* q^*_{i,j} + \phi(\bm{q}^*_i) - \phi(\bm{0}) + \min_\ell \frac{\mu_\ell^*}{F_\ell}\geq -m(\varepsilon'+2\gamma)-\epsilon''-\gamma$.
    ~This implies that $$\rev(\MM') \geq \rev(\MM) - n \left(m(\varepsilon'+2\gamma)+\epsilon''+\gamma\right)$$.
\end{lemma}

\begin{proof}
Assume $\ell^*= \argmin_\ell \frac{\hat{\mu}^*_\ell}{F_\ell}$.
Note that $\sum_j \hat{q}^*_{i,j} = 1$.
If agent $k$ reports type $t_k^{(i)}$ then :

\begin{align*}
     \hat{p}_k(t_k^{(i)}) = &\sum_{j} \pi_j^* q^*_{i,j} + \phi(\bm{q}^*_i) - \phi(\bm{0}) + \min_\ell \frac{\mu_\ell^*}{F_\ell} \\
    = & \sum_{j} \hat{\pi}_j^* \hat{q}^*_{i,j} + \phi(\bm{\hat{q}^*}_i) - \phi(\bm{0}) + \sum_{j} \frac{\hat{\mu}^*_{\ell^*}}{F_{\ell^*}} \hat{q}^*_{i,j} \\
    = & \sum_{j}  \hat{q}^*_{i,j} \left(\frac{\hat{\mu}^*_{\ell^*}}{F_{\ell^*}}  +  \hat{\pi}_j^*\right) + \phi(\bm{\hat{q}^*}_i) - \phi(\bm{0})  \\
    \geq & \sum_{j}  \hat{q}^*_{i,j} \left(\frac{\hat{\mu}^*_{\ell^*}}{F_{\ell^*}}  +  \hat{\pi}_j^*\right)
\end{align*}

By the KKT condition w.r.t to the convex program $(P^4)$ we have that:

\begin{align*}
    \sum_{j} \hat{q}^*_{i,j}\left(\frac{\hat{\mu}_{\ell^*}}{F_{\ell^*}}  +  \hat{\pi}_j^*\right) =&\sum_{j} \hat{q}^*_{i,j}\left(-\frac{\hat{\lambda}^*_{\ell^*,j}}{F_{\ell^*}} +  \hat{W}_{\ell^*,j} -\frac{\partial \phi(\hat{\bm{q}}^*_{\ell^*})}{\partial \hat{q}^*_{\ell^*,j}} \right)\\
    \geq & \sum_{j} \hat{q}^*_{i,j}\left( \hat{W}_{\ell^*,j} -\frac{\partial \phi(\hat{\bm{q}}^*_{\ell^*})}{\partial \hat{q}^*_{\ell^*,j}}\right)\qquad(\hat{\lambda}^*_{\ell^*,j}\leq 0)  \\
    \geq & \sum_{j} \hat{q}^*_{i,j}\left( -m(\varepsilon'+2\gamma)-\epsilon''-\frac{\partial \phi(\hat{\bm{q}}^*_{\ell^*})}{\partial \hat{q}^*_{\ell^*,j}}\right)\qquad(\text{Definition~\ref{def:hat_W}}) \\
    \geq  &  -m(\varepsilon'+2\gamma)-\epsilon''-\gamma\qquad(\text{the definition of $\phi(\cdot)$ and $\hat{q}^*_{i,j}\in[0,1]$
    })
\end{align*}

Finally, note that in $\MM'$, the expected revenue from payments $\{p_k(\cdot)\}_{k\in [n]}$ is exactly the same as $\rev(\MM)$. Thus we have
$$\rev(\MM') \geq \rev(\MM) - n \left(m(\varepsilon'+2\gamma)+\epsilon''+\gamma\right).$$

\end{proof}

We summarize the result for the ideal model in the following theorem, by choosing $\textbf{W}$ as in Theorem~\ref{thm:RRSF is BIC and IR}.

\begin{theorem}\label{thm:RRSF eps-BIC to BIC}
Let $\MM$ be an $\varepsilon$-BIC and IR mechanism. Fix any $\gamma>0$. Fix any agent $k$. For every $i,j\in [m_k]$, define $W_{i,j}^k$ as in Theorem~\ref{thm:RRSF is BIC and IR}. Let $\bm{q^*_k}=\{q^*_{i,j}\}_{i,j\in[m_k]}$ be the optimal solution of the convex program $(P^4)$ (described in Definition~\ref{def:hat_W}) and $\bm{\mu^*_k}=\{\hat{\mu}^*_{i}\}_{i \in [m_k]}$, $\bm{\pi^*_k}= \{\hat{\pi}^*_{j}\}_{j \in [m_k]}$ be the corresponding optimal Lagrange multipliers. Let mechanism $\MM'$ be the RRSF mechanism with parameters $\bm{q^*}=\{\bm{q^*_k}\}_{k\in[n]}, \bm{\mu^*}=\{\bm{\mu^*_k}\}_{k\in[n]}, \bm{\pi^*}=\{\bm{\pi^*_k}\}_{k\in[n]}$ as described in Definition~\ref{def:rrsf mechanism}. Then $\MM'$ is BIC, IR, and $$\rev(\MM') \geq \rev(\MM) - n \left(m(\varepsilon+2\gamma)+\gamma\right).$$
\end{theorem}

\begin{proof}
According to Lemma~\ref{lem:P vs P'}, $\bm{q^*_k}$ is also an optimal solution for the convex program $(P^3)$ and $\bm{\mu^*_k}$ and $\bm{\pi^*_k}$ are optimal Lagrange multipliers for $(P^3)$ as well. Hence, $\MM'$ is BIC and IR by Theorem~\ref{thm:RRSF is BIC and IR}. 

Since $\MM$ is $\epsilon$-BIC and IR, we have $W_{i,i}^k\geq \max\{W_{i,j}^k-\varepsilon, 0\}$, for every $k\in [n]$ and $i,j\in [m_k]$. By applying Lemma~\ref{lem:lower bound hat-p} with $\epsilon'=\epsilon$ and $\epsilon''=0$, we have that the revenue of $\MM'$ can be at most $n \left(m(\varepsilon+2\gamma)+\gamma\right)$ smaller than the revenue of $\MM$.

\end{proof}

\notshow{
\begin{corollary}\label{cor: not neg payments}
By Lemma~\ref{lem:P vs P'}, if we use the Lagrange multipliers $\{\hat{\mu}^*_{i}\}_{i \in [m_k]}, \{\hat{\pi}^*_{j}\}_{j \in [m_k]}$ from the convex program $(P^4)$ to define the payment function $\hat{p}_k(\cdot)$ for each agent $k$, the resulting mechanism is BIC, IR and the payment for each participating agent is at least $-m(\varepsilon+2\gamma)-\gamma$. 
\end{corollary}

\begin{remark}\cite{HartlineKM11}\label{rem: rev pres}
Since the first and second constraints of the program $(P^3)$ are satisfied, the distribution of the types that the agent participates in the $\varepsilon$-BIC mechanism $\MM$ with, is the same as her original distribution. Hence, the second phase revenue is $\rev(\MM)$.
\end{remark}

\begin{theorem}\cite{HartlineKM11}\label{thm: rev frac ideal}
Let $\MM$ be an $\varepsilon$-BIC mechanism and $\MM'$ be the mechanism described in \ref{def:rrsf mechanism}. By combining Corollary~\ref{cor: not neg payments}, Remark~\ref{rem: rev pres} we have that $\rev{\MM'} \geq \rev{\MM} - n \left(m(\varepsilon+2\gamma)+\gamma\right)$, where $n$ is the number of agents. 
Moreover $\MM'$ is BIC and IR.
\end{theorem}
}

At the end of Section~\ref{sec:ideal model RRSF}, we prove a lemma about RRSF mechanism using the KKT condition of the convex program $(P^3)$. It's useful in the proof of Lemma~\ref{lem:rev_RERSF}.

\begin{lemma}\label{lem:RRSF-IR}
Let $\MM$ be any RRSF mechanism with respect to $\textbf{W}=\{W_{i,j}^k\}_{k\in [n],i,j\in [m_k]}$. It has parameter $(\bm{q^*}, \bm{\mu^*}, \bm{\pi^*})$. Then for each agent $k$ and $i\in [m_k]$, $\sum_{j\in [m_k]}q_{i,j}^*W_{i,j}^k-\hat{p}_k(t_k^{(i)})\geq 0$.
\end{lemma}

\begin{proof}
Fix any agent $k$. Note that $(\bm{q_k^*}, \bm{\mu_k^*}, \bm{\pi_k^*})$ satisfies the KKT conditions of $(P^3)$. 
By the definition of $\hat{p}_k(t_k^{(i)})$, for every $i\in [m_k]$, we have
\begin{align*}
     \sum_{j\in [m_k]}q_{i,j}^*W_{i,j}^k-\hat{p}_k(t_k^{(i)}) = &\sum_{j} q^*_{i,j}(W_{i,j}^k-\pi_j^*)+ \phi(\bm{0}) - \phi(\bm{q}^*_i) - \min_\ell \frac{\mu_\ell^*}{F_\ell}\\ \geq& \sum_{j} q^*_{i,j}(W_{i,j}^k-\pi_j^*-\frac{\mu_i^*}{F_i})+ \phi(\bm{0}) - \phi(\bm{q}^*_i)~~~~~(\sum_jq_{i,j}^*=1)\\
    =&\sum_{j} q^*_{i,j}(\frac{\lambda_{i,j}^*}{F_i}+\frac{\partial \phi(q_i^*)}{\partial q_{i,j}^*})+ \phi(\bm{0}) - \phi(\bm{q}^*_i)~~~~~(\text{KKT condition 1})\\
    =&\sum_{j} q^*_{i,j}\frac{\partial \phi(q_i^*)}{\partial q_{i,j}^*}+ \phi(\bm{0}) - \phi(\bm{q}^*_i)~~~~~(\text{KKT condition 3})\\
    =& \nabla \phi(\bm{q}_i^*)^T\cdot\bm{q}_i^*+\phi(\bm{0}) - \phi(\bm{q}^*_i)\geq 0~~~~~(\text{$\phi(\cdot)$ is convex})
\end{align*}
\end{proof}

\subsection{Non-Ideal Model}\label{sec:non-ideal model RRSF}
In the previous section we showed that, in the ideal model, the RRSF mechanism $\MM'$ has only a small revenue loss, if the original mechanism $\MM$ is $\varepsilon$-BIC and IR. However, with only sample access to the distributions of the edge weights, these weights can not be computed exactly. In this section, we show that a BIC and IR mechanism with small revenue loss can be constructed, using estimates for the edge weights. The approach is similar to~\cite{hartline2015bayesian}. For any fixed agent $k$, let $\WW_{i,j}^k$ be the utility of agent $k$ on an execution of the mechanism $\MM=(x,p)$ when she reports $t_k^{(j)}$ and her true type is $t_k^{(i)}$.
More formally,
\begin{equation}\label{equ:rv-WW}
\WW_{i,j}= v_k(t_k^{(i)},x(t_k^{(j)}, t_{-k})) - p_k(t_k^{(j)},t_{-k})
\end{equation}
Note that $\WW_{i,j}$ is a random variable over $t_{-k}\sim \DD_{-k}$ and the randomness of $\MM$, whose expectation is $W_{i,j}$ as defined in Definition~\ref{def:rrsf mechanism}.

\begin{definition}[Definition~4.9 from \cite{hartline2015bayesian}]\label{def:rersf}
The Regularized Estimated Replica–Surrogate Fractional Assignment (RERSF) with parameter $L$,
is defined as follows:
\begin{enumerate}
    \item Fix any agent $k$. For every pair of types $t_k^{(i)}, t_k^{(j)} \in \supp(\DD_k)$,
    we define $\widetilde{W}_{i,j}^k$ as the empirical mean with $L$ samples of $\WW_{i,j}$.  A sample of $\WW_{i,j}$ is obtained by drawing a sample $t_{-k}\sim \DD_{-k}$, running mechanism $\MM$ with input $(t_k^{(j)}, t_{-k})$, and computing the utility in Equation~\eqref{equ:rv-WW} based on the output outcome and payment. Let $\OO_k$ be the set of output outcomes from $\MM$ among all samples and pairs $(i,j)$. In total there are $m_k\cdot L$ number of outcomes.
    \item Run the RRSF mechanism $\MM'$ in Definition~\ref{def:rrsf mechanism} with respect to  $\widetilde{\textbf{W}}=\{\widetilde{W}_{i,j}^k\}_{k\in [n], i,j\in [m_k]}$.
\end{enumerate}
\end{definition}

The following lemma shows that if all agents report truthfully, the RERSF mechanism has a small revenue loss compared to $\MM$. 

\begin{lemma}
\label{lem:rev_RERSF}
For any $\eta\in (0,1)$, suppose $L \geq \frac{2}{\eta^2}\ln\left( \frac{2nm^2}{\eta} \right)$. Then the revenue of RERSF, when all agents report truthfully, is at least $\rev(\MM) - nm(\epsilon+6\gamma+5\eta)$. Moreover, the expected utility for each agent when she reports truthfully is at least $-3\eta$.
\end{lemma}

\begin{proof}


Note that $\WW_{i,j}^k\in [-1,1]$. By the Hoeffding bound, when $L\geq \frac{2}{\eta^2}\ln\left( \frac{2nm^2}{\eta}\right)$, $\Pr\left[|\widetilde{W}^k_{i,j}-W_{i,j}^k|>\eta\right]\leq\frac{\eta}{nm^2}$. By taking the union bound over all agents, replicas, and surrogates, we have that with probability at least $1-\eta$, $|\widetilde{W}^k_{i,j}-W_{i,j}^k| \leq \eta$ holds for all agent $k$ and types $t_k^{(i)},t_k^{(j)}\in \supp(\DD_k)$. We refer to the event that the above inequalities hold as a ``good'' event.

For the first statement, consider the case when a ``good'' event happens. Since $\MM$ is $\varepsilon$-BIC and IR, we have $W_{i,i}^k\geq W_{i,j}^k-\varepsilon$ and $W_{i,i}^k\geq 0$, for every $k\in [n]$ and $i,j\in [m_k]$. Thus $$\widetilde{W}_{i,i}^k\geq \widetilde{W}_{i,j}^k-\varepsilon-|\widetilde{W}_{i,i}^k-W_{i,i}^k|-|\widetilde{W}_{i,j}^k-W_{i,j}^k|\geq \widetilde{W}_{i,j}^k-\varepsilon-2\eta$$ 

Moreover, $\widetilde{W}_{i,i}^k\geq W_{i,i}^k-|\widetilde{W}_{i,i}^k-W_{i,i}^k|\geq W_{i,i}^k-\eta$. By applying Lemma~\ref{lem:lower bound hat-p} with $\varepsilon'=\epsilon+2\eta$ and $\varepsilon''=\eta$, we have that the expected revenue of RERSF is at most $n \left(m(\varepsilon+2\gamma+2\eta)+\gamma+\eta\right)$ smaller than the revenue of $\MM$. When a ``good'' event does not happen, we can apply Lemma~\ref{lem:lower bound hat-p} with the trivial guarantee $\varepsilon'=\varepsilon''=1$. The revenue loss is at most $n \left(m(1+2\gamma)+1+\gamma\right)\leq nm(2+3\gamma)$. Thus the expected revenue loss is at most
$$n \left(m(\varepsilon+2\gamma+2\eta)+\gamma+\eta\right)+nm(2+3\gamma)\cdot \eta\leq nm(\epsilon+6\gamma+5\eta)$$

We now prove the second statement. For every agent $k$, let $t_k^{(i)}$ be her true type. For every $j\in [m_k]$, let $q_{i,j}^*$ be the probability that she is represented by type $t_k^{(j)}$ in the RRSF mechanism (step 2 of the RERSF mechanism). Then her expected utility by reporting truthfully is $\sum_jW_{i,j}^k\cdot q_{i,j}^*-\hat{p}_k(t_k^{(i)})$, where $\hat{p}_k$ is the extra payment in the RRSF mechanism with respect to $\widetilde{\textbf{W}}$. By Lemma~\ref{lem:RRSF-IR}, we have
$\sum_j\widetilde{W}_{i,j}^k\cdot q_{i,j}^*-\hat{p}_k(t_k^{(i)})\geq 0$. We have
$$\sum_jW_{i,j}^k\cdot q_{i,j}^*-\hat{p}_k(t_k^{(i)})\geq \sum_j\widetilde{W}_{i,j}^k\cdot q_{i,j}^*-\hat{p}_k(t_k^{(i)})-\sum_{j}q_{i,j^*}|\widetilde{W}_{i,j}^k-W_{i,j}^k|\geq -\sum_{j}q_{i,j^*}|\widetilde{W}_{i,j}^k-W_{i,j}^k|$$

If a ``good'' event happens, $-\sum_{j}q_{i,j^*}|\widetilde{W}_{i,j}^k-W_{i,j}^k|\geq -\eta$. If a ``good'' event does not happen, we can apply the trivial guarantee $|\widetilde{W}_{i,j}^k-W_{i,j}^k|\leq 2$ (since $\WW_{i,j}^k\in [-1,1]$) to get $-\sum_{j}q_{i,j^*}|\widetilde{W}_{i,j}^k-W_{i,j}^k|\geq -2$. Hence the expected utility for agent $k$ when she reports truthfully is at least $(-\eta)\cdot 1+(-2)\cdot \eta=-3\eta$.


\end{proof}

Note that for every agent $k$ with type $t_k^{(i)}$, her expected utility when represented by a surrogate with type $t_k^{(i)}$, is $W_{ij}^k$ instead of $\widetilde{W}_{i,j}^k$. Thus the RERSF mechanism is not BIC. However, with the number of samples $L$ sufficiently large, one can show that the mechanism is $\epsilon'$-BIC for some small $\epsilon'>0$.
To formally prove this claim, we first introduce the \emph{distinguishability} of two types.

\begin{definition}[Definition~3.3 from \cite{hartline2015bayesian}]\label{dfn: dist outcome pair}
Let $v_k(\cdot, \cdot)$ be the valuation function of agent $k$. We define the \emph{swap-disutility} for the types $t_k^{(i)}, t_k^{(j)} \in \supp(\DD_k)$ and outcomes $o, o' \in \OO$ to be $d(t_k^{(i)}, t_k^{(j)},o, o') = v(t_k^{(i)}, o) - v(t_k^{(i)}, o') + v(t_k^{(j)}, o') - v(t_k^{(j)}, o')$. The \emph{distinguishability} of types $t_k^{(i)}, t_k^{(j)}$ is defined as
$$d(t_k^{(i)}, t_k^{(j)}) = \max_{o,o' \in \OO} d(t_k^{(i)}, t_k^{(j)}, o, o')$$
The pair $o, o' \in \OO$ that achieves the maximum is called the \emph{distinguishing outcome pair} for $t_k^{(i)}, t_k^{(j)}$. We call $t_k^{(i)}, t_k^{(j)}$ distinguishable w.r.t. an outcome set $\OO$ if they have non-zero distinguishability for some outcome in that set. For every pair of types $t_k^{(i)},t_k^{(j)}\in \supp(\DD_k)$, define the \emph{sampled outcome distinguishability} $\widetilde{d}(t_k^{(i)}, t_k^{(j)})$ to be the distinguishability of types $t_k^{(i)}, t_k^{(j)}$ with respect to the output outcome set $\OO_k$.
\end{definition}



\begin{lemma}[Lemma~4.7 from \cite{hartline2015bayesian}]
\label{lem:deviate_RERSF}
For any $\gamma>0$, $\zeta \in (0, 1/2)$, consider the RERSF mechanism with parameter $L > \frac{1}{2}\zeta^{-2} \ln(2m/\zeta)$.
Then any agent $k$ with type $t_k^{(i)}\in \supp(\DD_k)$ that participates in the mechanism can gain at most $\frac{2\zeta m}{\gamma}\E[\widetilde{d}(t_k^{(i)},t_k^{(j)})]$ utility by reporting type $t_k^{(j)}\in \supp(\DD_k)$,
where $\widetilde{d}(\cdot,\cdot)$ is the sampled outcome distinguishability of $t_k^{(i)},t_k^{(j)}$. The expectation is taken over the randomness of $\OO_k$.

\end{lemma}

In order to make our mechanism BIC, we will combine the above RERSF mechanism with a strictly IC mechanism defined below. In a strictly IC mechanism, for every agent, the utility by reporting her true type is strictly larger the utility by reporting any other types, by some positive value.

\begin{definition}[\cite{hartline2015bayesian}]\label{def:strict-ic-mech}
\emph{(Strict IC mechanism)} Fix any agent $k$. 
The strictly IC mechanism for agent $k$, denoted by $SIC^k$, consists of the following steps:


\begin{enumerate}
    \item Chooses uniformly at random a pair of distinguishable types $t_k^{(i)},t_k^{(j)}\in \supp(\DD_k)$. Let $(o,o')$ be the distinguishing outcome pair for $(t_k^{(i)},t_k^{(j)})$, where $o, o' \in \OO_k$.
    \item Let $v = v_k(t_k^{(i)},o) - v_k(t_k^{(i)},o')$ and $v' = v_k(t_k^{(j)},o) - v_k(t_k^{(j)},o')$ and define the price $p =\frac{v+v'}{2}$. Without loss of generality, assume $p \geq 0$, otherwise we can swap $o$ and $o'$.
    \item The mechanism lets agent $k$ select between the following two options: 1) The mechanism outputs an outcome $o$ with payment $p$ for agent $k$; 2) The mechanism outputs an outcome $o'$ with payment $0$ for agent $k$.
\end{enumerate}

\end{definition}

\begin{lemma}[Lemma~4.4 from \cite{hartline2015bayesian}]
\label{lem:deviate_SIC}
The $SIC^k$ mechanism for agent $k$ is strictly $IC$: The utility of agent $k$ by reporting her true type $t_k^{(i)}$ is at least $\frac{\E[\widetilde{d}(t_k^{(i)},t_k^{(j)})]}{2M}$ larger than the utility by reporting another type $t_k^{(j)}$. Here $M$ is the number of distinguishable pairs of types in $\supp(\DD_k)$, $\widetilde{d}(\cdot,\cdot)$ is the sampled outcome distinguishability of $(t_k^{(i)},t_k^{(j)})$ w.r.t. $\OO_k$, and the expectation is taken over the randomness of $\OO_k$.
\end{lemma}


Now we are ready to construct our mechanism for the non-ideal model.

\begin{definition}\label{def:mech-non-ideal}
For any $\delta>0$, define mechanism $\MM'$ as follows. The mechanism has parameter $\gamma,\eta,L$ (from the RERSF mechanism) and $\delta,C>0$. With probability $1-\delta$, it runs the RERSF mechanism (with parameter $\gamma,\eta,L$); With probability $\delta$, it picks an agent $k$ uniformly at random and runs $SIC^k$ for agent $k$. The payment for agent $k$ is defined in Definition~\ref{def:strict-ic-mech}. Other agents pays 0. At last, the mechanism subsidizes every agent $C$.
\end{definition}

We wrap everything up in Theorem~\ref{thm:non-ideal}. We prove that with a proper choice of each parameter, mechanism $\MM'$ in Definition~\ref{def:mech-non-ideal} is BIC and IR. Also the revenue loss of $\MM'$ is small.


\begin{theorem}\label{thm:non-ideal}
There exist choices for the parameters $\gamma = \varepsilon, \delta = \varepsilon, \eta = \varepsilon, C=4\varepsilon$ and $L=\poly(n,m,\frac{1}{\varepsilon})$
such that the mechanism $\MM'$ in Definition~\ref{def:mech-non-ideal} is BIC, IR and $\rev(\MM')\geq \rev(\MM) - O(nm\varepsilon)$. Mechanism $\MM'$ has expected running time $\poly \left(n, m, 1/\varepsilon,b\right)$ and makes {in expectation} at most {$\poly\left(n,m,1/\varepsilon\right)$}~queries~to~$\MM$. Here $b$ is an upper bound on the bit complexity of $v_k(t_k^{(i)}, o), F_i^k, o$ for every $k \in [n], i \in [m_k], o \in \OO$.
\end{theorem}

\begin{proof}

Let $\zeta= \frac{\varepsilon^2}{4nm^3}$ and $L = \bigg\lceil\max \bigg(\frac{2}{\varepsilon^2}\ln\left( \frac{2nm^2}{\varepsilon} \right), {\frac{8m^6n^2}{\varepsilon^4}\ln\left(\frac{8m^4n}{\varepsilon^2}\right)}\bigg)\bigg\rceil$. We first show that the mechanism is BIC. Fix any agent $k \in [n]$ and her type $t_k^{(i)} \in \supp(\DD_k)$. We are going to bound the difference between the utility of reporting truthfully and the utility of reporting type $t_k^{(j)}\in \supp(\DD_k)$, for any $j\not=i$.
Note that $L > \frac{1}{2}\zeta^{-2} \ln(2m/\zeta)$. By Lemma~\ref{lem:deviate_RERSF}, when $\MM'$ executes the RERSF mechanism (with probability $1-\delta$), the difference is at least $-\frac{2\zeta m}{\gamma}\E[\widetilde{d}(t_k^{(i)},t_k^{(j)})]$. By Lemma~\ref{lem:deviate_SIC}, when $\MM'$ executes $SIC^k$ (with probability $\delta/n$), the difference is at least $\frac{\E[\widetilde{d}(t_k^{(i)},t_k^{(j)})]}{2m^2}$, since the number of distinguishable pairs of types w.r.t. $\OO_k$ is at most $m^2$. When $\MM'$ executes $SIC^{\ell}$ for some $\ell\not=k$, the difference is 0 as agent $k$'s reported type won't affect the output outcome and her payment. Thus the difference between the expected utility when truthfully reporting and misreporting is at least
$$-\frac{2\zeta m}{\gamma}\E[\widetilde{d}(t_k^{(i)},t_k^{(j)})]\cdot (1-\delta)+\frac{\E[\widetilde{d}(t_k^{(i)},t_k^{(j)})]}{2m^2}\cdot \frac{\delta}{n}$$

The above value is non-negative when $\delta=\gamma=\varepsilon$ and $\zeta= \frac{\varepsilon^2}{4nm^3}$. Thus the mechanism is BIC.



Furthermore, since $L\geq \frac{2}{\eta^2}\ln\left( \frac{2nm^2}{\eta} \right)$, by Lemma~\ref{lem:rev_RERSF}, the expected utility of agent $k$, given the fact that $\MM'$ executes the RERSF mechanism, is at least $C-3\varepsilon$. When $\MM'$ executes $SIC^{k}$, the utility for agent $k$ is at least $-1+C$. This is because by Definition~\ref{def:strict-ic-mech}, the payemnt of the strict IC mechanism is at least -1 and the agent's value for the outcome is non-negative.
When $\MM'$ executes $SIC^{\ell}$ for some $\ell\not=k$, agent $k$'s utility is at least $C$, since her value for the outcome is non-negative. Thus agent $k$'s overall expected utility by reporting her true type $t_{k}^{(i)}$ is at least

$$(C-3\varepsilon)\cdot (1-\delta)+(-1+C)\cdot\frac{\delta}{n}+C\cdot\frac{\delta(n-1)}{n}\geq 0$$ 

The above inequality holds when $\delta=\varepsilon$ and $C=4\varepsilon$. Thus $\MM'$ is IR. To bound the revenue loss, note that the payment in mechanism $SIC^k$ is always non-negative. By Lemma~\ref{lem:rev_RERSF}, we have

$$    \rev(\MM')\geq (1-\delta)\cdot\left(\rev(\MM)-nm(\varepsilon+6\gamma+5\eta)\right) \geq \rev(\MM) -n(12m+1)\varepsilon
$$

The last inequality follows from the fact that $\rev(\MM)\leq n$.

Finally we discuss the running time of $\MM'$ and number of queries to $\MM$. For a given input, in step 1 of the RERSF mechanism (see Definition~\ref{def:rersf}), the mechanism makes $L$ queries to $\MM$ for every $k\in [n]$ and $i,j\in [m_k]$. In step 2 of RERSF, it makes one query to obtain the output outcome for a given input. The strictly IC mechanism does not make queries to $\MM$. Overall, $\MM'$ makes at most $O(mL)=\poly(n,m,1/\varepsilon)$ queries to $\MM$.  

For the running time of $\MM'$, first consider the RERSF mechanism. By~\cite{KOZLOV1980223}, the convex program $(P^3)$ and its dual problem can be solved in $poly(n,m,L,b)$ time. Thus all the parameters in the RRSF is determind in $poly(n,m,L,b)$ time. One can easily verify that other procedures of RERSF runs in time $poly(n,m,L,b)$. Thus RERSF runs in $poly(n,m,1/\varepsilon,b)$ time.

Now consider the running time of the strictly IC mechanism $SIC^k$. Note that finding all the distinguishable types with respect to $\OO_k$ takes time at most $O(m^4L^2)$, by going over all pairs of types $(t_k^{(i)},t_k^{(j)})$ and pairs of outcomes $(o,o')$. Thus $SIC^k$ runs in time $poly(n,m,1/\varepsilon)$.
Altogether, $\MM'$ runs in time $poly(n,m,1/\varepsilon,b)$.
\end{proof}



\vspace{.2in}

Similar to Section~\ref{sec:apps}, we can apply Theorem~\ref{thm:non-ideal} to the multi-dimensional revenue maximization problem and derive the following theorem. The proof of Theorem~\ref{thm:vwo_reduction_general} is analogous to the proof of Theorem~\ref{thm:app2}, following from Theorem~\ref{thm:non-ideal} and Theorem~\ref{thm:rev-max}.

\begin{theorem}
\label{thm:vwo_reduction_general}
Let $\OO$ be a general outcome space.
Given \textbf{full access} to agents' type distributions $\DD=\bigtimes_i \DD_i$. Let $b$ be an upper bound on the bit complexity of $v_i(t_i,o)$ and $\Pr(t_i)$ for any agent $i$, any type $t_i$, and any outcome $o$, and $\opt$~be the optimal revenue achievable by any BIC and IR mechanism. We further assume that  types are normalized, that is, for each agent $i$, type $t_i$ and outcome $o$,
$v_i(t_i,o)\in [0,1]$.

 Given oracle access to an $\alpha$-approximation algorithm $G$ for VWO
with running time $rt_G(x)$, where $x$ is the bit complexity of the input,
there is an algorithm that terminates in $\poly\left(n, m,\frac{1}{\varepsilon},b,rt_G\left(\poly\left( n, m,\frac{1}{\varepsilon},b \right)\right)\right)$ time, and outputs an \textbf{exactly BIC} and IR mechanism with expected revenue $$
\rev(\MM,\DD) \geq \alpha\cdot OPT - O\left(n  m \varepsilon \right),$$

where $m = \max_{i\in [n]} |\supp(\DD_i)|$. On any input bid, $\MM$ computes the outcome and payments in expected running time $\poly\left(n, m,\frac{1}{\varepsilon},b,rt_G\left(\poly\left( n, m,\frac{1}{\varepsilon},b \right)\right)\right)$.

\notshow{
Given oracle access to Algorithm~G (Theorem~\ref{thm:rev-max}) and only sample access to the bidders' distributions $\dD_1,\ldots,\dD_n$,
we can compute a exactly BIC Mechanism $\MM$ in time $\poly\left(\sum_i{|T_i|},\frac{1}{\varepsilon},b,rt_G\left(\poly\left( \sum_i{|T_i|},\frac{1}{\varepsilon},b \right)\right)\right)$
with expected revenue

$$
\rev(\MM,\DD) \geq \alpha OPT -  O\left(n \sqrt{\varepsilon} \right). 
$$
}


\end{theorem}

\bibliographystyle{plain}
\bibliography{Yang}
\newpage
\appendix

\section{Estimating $\gamma$: Approximating the Offline Optimum of the Regularized Matching}\label{appx:offline optimal approximation}

In this section, we show how to estimate the parameter $\gamma$ so that it is a constant factor approximation to optimum of program $(P'')$ (see Definition~\ref{def:augmented entropy regularized matching}) on the replica-surrogate matching in Mechanism~\ref{alg:mout} with high probability. Importantly, the estimate is completely independent from the agent's reported type. Here is the basic idea. 
We sample the edge weights between $\Br'$ and $\Bs$, and use the empirical mean to compute the optimal solution of program $(P'')$. We show that with polynomially many samples, the optimum of $(P'')$ computed based on the empirical means is a constant approximation to the optimum of $(P'')$ on the true edge weights 
 with probability almost $1$.


\notshow{
We use $\omega_{jk}(\Br)$ to denote $W_{i}(\rj,\sk)$ and $\hw_{jk}(\Br)$ to denote its empirical 

we consider the surrogate profile fixed.
Let $\bm{r'}$ be a new independent draw of $d \ell$ replicas.
We consider $\omega_{jk}(\bm{r})$ to be the expected weight from replica $i$  to surrogate $j$ under replica profile $\bm{r}$ and $\hat{\omega}_{jk}(\bm{r})$ its empirical estimation with $N$ samples. Let $\bm{\omega}(\bm{r})$ and $\bm{\hat{\omega}}(\bm{r})$ be the vectors of the expected weights and their empirical estimates respectively, under replica profile $\bm{r}$.
We prove that the optimal solution $OPT(\bm{\omega}\bm{(r)})$ under surrogate profile $\bm{r}$ and the true weights is close to the optimal solution $OPT(\bm{\hat{\omega}}\bm{(r')})$ under surrogate profile $\bm{r'}$ and the estimated weights. We do that in two steps, similarly to \cite{dughmi2017bernoulli}.
}

\begin{lemma}\label{lem:find gamma}
For any agent $i$, given parameters $\ell$, $\delta$, $\eta'$, and $d \geq 32 \frac{\log(8 \eta'^{-1})}{\delta^2 \ell \log^2(\ell)}$, fix $\Bs$ to be the $\ell$ surrogates, first draw $d\ell$ fresh samples from $\DD'_i$, which we denote using   $\Br'$. We use $\opt ({\omega})$ to denote the optimum of $(P'')$ when the edge weight between the $j$-th replica/LHS-node and the $k$-th normal surrogate/normal RHS-node is $\omega_{jk}$.
 There exists a randomized algorithm based  only on $\Br'$ and $\Bs$ that computes a $\gamma$ that lies in $\left[ \frac{2\cdot \opt ({\omega}(\Br'))}{d}, \frac{24\cdot \opt ({\omega}(\Br'))}{d}\right]$ with probability at least $1-\eta'/2$, where  
$\omega_{jk}(\Br')=W_i(r'^{(j)},\sk)$ as defined in Mechanism~\ref{alg:mout}. Moreover, $\gamma$ is at most $O\left(\max\{\ell,\delta\ell\log (\ell)\}\right)$ and the algorithm has $\poly(d, \ell, 1/\eta', 1/\delta)$ running time and makes $\poly(d, \ell, 1/\eta', 1/\delta)$ queries to mechanism $\MM$.

Furthermore, if $\Br$ are $d\ell$ i.i.d. samples from $\DD'_i$, then $\opt(\omega(\Br'))$ lies in $\left[ {\opt({\omega}(\Br))\over 2 },  {3\cdot \opt({\omega}(\Br))\over 2 }\right]$ with probability at least $1-\eta'/2$ over the randomness of $\Br$ and $\Br'$, 
where $\omega_{jk}(\Br)=W_i(\rj,\sk)$. In this case, $\gamma$ also lies in $\left[ {\opt({\omega}(\Br)) \over d},  {36\cdot \opt({\omega}(\Br))\over d} \right]$ with probability at least $1-\eta'$.
\end{lemma}

\begin{prevproof}{Lemma}{lem:find gamma}
We prove our statement in two steps. In the first step, we show that if we take polynomially many samples, we can obtain a sufficiently accurate estimate of $\omega(\Br')_{jk}$ for each edge $(j,k)$. We prove that the optimum of $(P'')$ on the estimated weights is close to $\opt({\omega}(\Br'))$. We use $\FF_{jk}$ to denote the distribution of the random variable $v_i(r'^{(j)},x(\sk, t_{-i}))-(1-\eta)\cdot p_i(\sk,t_{-i})$, where $t_{-i}$ is distributed according to $\DD_{-i}$.

\begin{lemma}[adapted from \cite{dughmi2017bernoulli}]
\label{lemma:h1}
For each edge $(j,k)$ between the $j$-th replica $r'^{(j)}$ and the $k$-th normal surrogate $\sk$, if we take $N\geq \frac{2 \log(4  \ell^2 d \eta'^{-1})}{\delta^2 \log^2(\ell)}$ samples from $\FF_{jk}$, and use $\hw_{jk}(\Br')$ to denote the empirical mean of these $N$ samples, then with probability at least $1-\frac{\eta'}{2}$,

\begin{equation*}
    \frac{\opt({\omega}(\bm{r'}))}{2} \leq \opt({\hat{\omega}}(\bm{r'})) \leq 2 \opt({\omega}(\bm{r'})).
\end{equation*}

\end{lemma}

\begin{proof}
By the Chernoff bound, we know that $\Pr\left[ |\omega_{jk}(\bm{r}) - \hw_{jk}(\bm{r})| \geq \frac{\delta\log(\ell)}{2}\right]\leq \frac{\eta'}{2d \ell^2}$ for each edge $(j,k)$. Since there are $d\ell^2$ many edges, by the union bound, we have that with probability at least $1-\frac{\eta'}{2}$ for each edge $(j,k)$:

\begin{equation*}
    |\omega_{jk}(\bm{r}) - \hw_{jk}(\bm{r})| \leq  \frac{\delta\log(\ell)}{2}
\end{equation*}

\notshow{
\begin{equation*}
    2e^{-\frac{N \delta^2 \log^2(\ell)}{2}} \leq 2e^{-\log\left( 4d \ell^2\eta'^{-1}) \right)} = \frac{\eta'}{2d \ell^2}
\end{equation*}

We have that for a specific $\omega_{jk}(\bm{r})$ when $s_k$ is not a $0$-surrogate :

\begin{equation*}
    |\omega_{jk}(\bm{r}) - \hat{\omega_{jk}}(\bm{r})| \geq \frac{\delta\log(\ell)}{2}
\end{equation*}
}

Let $({x}^*,{y}^*)$ be the optimal solution of the $(P'')$ with edge weights $(\omega_{jk}(\Br'))_{jk}$ and $({x}^{**},{y}^{**})$ be the optimal solution of the $(P'')$ with edge weights $(\hw_{jk}(\Br'))_{jk}$. Then

\begin{align*}
\opt({\omega}(\bm{r}')) & = \sum_{j, k} \left( x_{jk}^* {\omega_{jk}}(\bm{r}') - \delta\cdot \left( x_{jk}^*\log (x_{jk}^*) + y_{jk}^*\log (y_{jk}^*) \right) \right)\\
&\geq \sum_{j, k} \left( x_{jk}^{**} {\omega_{jk}}(\bm{r}') - \delta\cdot \left( x_{jk}^{**}\log (x_{jk}^{**}) + y_{jk}^{**}\log (y_{jk}^{**}) \right) \right)\\
& \geq \opt(\bm{\hat{\omega}}(\bm{r'})) - \frac{d \ell\delta\log(\ell)}{2}\\
& \geq \frac{\opt(\bm{\hat{\omega}}(\bm{r'}))}{2}
\end{align*}

The last inequality holds since a valid assignment is to set $y_{jk}=1/\ell$ and $x_{jk}=0$ for each $j$, $k$, which has objective value $\delta d \ell  \log(\ell)$.

The other direction can be proved similarly.
\end{proof}

Let $A$ be the total weight of the maximum weight matching with edge weights $(\hw_{jk}(\Br'))_{jk}$. It is clear that $A$ lies in $\left[ \opt(\hw(\Br'))-\delta d\ell\log(2\ell),\opt(\hw(\Br'))\right]$. Note that if we set $y_{jk}=1/\ell$ and $x_{jk}=0$ for each $j$, $k$, the objective of $(P'')$ has value $\delta d \ell  \log(\ell)$. Hence, $\max\{A,\delta d \ell  \log(\ell)\}$ is guaranteed to lie in  $\left[{\opt(\hw(\Br'))\over 3}, \opt(\hw(\Br'))\right]$. If we choose $ \gamma$ to be $\frac{12\cdot\max\{A,\delta d \ell  \log(\ell)\}}{d}$, $\gamma$ is guaranteed to lie in  $\left[{4\cdot\opt(\hw(\Br'))\over d}, {12\cdot \opt(\hw(\Br'))\over d}\right]$. Due to Lemma~\ref{lemma:h1}, $\gamma$ lies in $\left[{2\cdot\opt(\omega(\Br'))\over d}, {24\cdot\opt(\omega(\Br'))\over d}\right]$ with probability at least $1-\eta'/2$. As $A$ can be computed in time $\poly(d,\ell,1/\eta',1/\delta)$, $\gamma$ can also be computed in time $\poly(d,\ell,1/\eta',1/\delta)$.

In the second step of the proof, we show that $\opt(\omega(\Br))$ and $\opt(\omega(\Br'))$ are close with high probability. We first need the following Lemma to prove $\opt(\omega(r))$ has bounded difference of $2$.

\begin{lemma}For any $j\in[d\ell]$, any type $r'^{(j)}$ and replica profile $\bm{r}$,
\begin{equation*}
   \left \lvert \opt\left({\omega}(\Br)\right) - \opt\left({\omega}\left(r'^{(j)},\bm{r^{(-j)}}\right)\right)\right \rvert \leq 2,
\end{equation*}
where $\omega_{jk}\left(r'^{(j)},\bm{r^{(-j)}}\right)=W_i(r'^{(j)},\sk)$ and $\omega_{j'k}\left(r'^{(j)},\bm{r^{(-j)}}\right)=\omega_{j'k}(\Br)$ for any $j'\neq j$.
\end{lemma}

\begin{proof}
Let $({x}^*,{y}^*),({x}^{**},{y}^{**})$  be the optimal solutions under replica profile $\Br$ and $\left(r'^{(j)},\bm{r^{(-j)}}\right)$ for $(P'')$ respectively. Then
\begin{align*}
\opt(\bm{\omega}(\Br)) 
&= \sum_{j'\in[d\ell], k\in[\ell]} \left( x_{j'k}^* {\omega_{j'k}}(\bm{r}) - \delta\cdot \left( x_{j'k}^*\log (x_{j'k}^*) + y_{j'k}^*\log (y_{j'k}^*) \right) \right)\\ 
& \geq \sum_{j'\in[d\ell], k\in[\ell]} \left( x_{j'k}^{**} {\omega_{j'k}}(\bm{r}) - \delta\cdot \left( x_{j'k}^{**}\log (x_{j'k}^{**}) + y_{j'k}^{**}\log (y_{j'k}^{**}) \right) \right)\\
& \geq \sum_{j'\neq j, k\in[\ell]} \left( x_{j'k}^{**} {\omega_{j'k}}(\bm{r}) - \delta\cdot \left( x_{j'k}^{**}\log (x_{j'k}^{**}) + y_{j'k}^{**}\log (y_{j'k}^{**}) \right) \right)\\
&~~~~~~~~~~~~~~~~~~~~+\sum_{k\in[\ell]} \left( x_{jk}^{**} \left({\omega_{jk}}(r'^{(j)},r^{(-j)})-2\right) - \delta\cdot \left( x_{jk}^{**}\log (x_{jk}^{**}) + y_{jk}^{**}\log (y_{jk}^{**})\right)\right)\\
 & = \opt({\omega}(r'^{(j)},r^{(-j)})) - 2\\
\end{align*}

The last inequality is because both $\omega_{jk}(\Br)$ and ${\omega}(r'^{(j)},r^{(-j)})$ lie in $[-1,1]$. 
The other direction follows similarly.
\end{proof}

Next, we apply McDiarmid's inequality to the function $\opt(\omega(r))$.

\begin{lemma}
\label{lemma:h2}
When $d \geq 32 \frac{\log(8 \eta'^{-1})}{\delta^2 \ell \log^2(\ell)}$, if both $\Br$ and $\Br'$ are collections of $d\ell$ i.i.d. samples from $\DD'_i$, then with probability at least $1-\frac{\eta'}{2}$,

\begin{equation*}
\frac{1}{2}\opt({\omega}(\bm{r}))\leq \opt({\omega}(\bm{r}'))\leq \frac{3}{2} \opt({\omega}(\bm{r})).
\end{equation*}

The probability is over the randomness of both $\Br$ and $\Br'$.
\end{lemma}

\begin{proof}
Due to Lemma~\ref{lemma:h2}, we can apply McDiarmid's  inequality on the function $\opt({\omega}(\bm{r}))$, and we have

\begin{equation*}
   \Pr_{\Br}\left[ \left\lvert \opt({\omega}(\bm{r}))  - \E_{\Br}[\opt({\omega}(\bm{r})] \right\rvert \geq \frac{\delta d \ell \log(\ell)}{4}\right]\leq \frac{\eta'}{4}
\end{equation*}

Similarly, we have

\begin{equation*}
     \Pr_{\Br'}\left[ \left\lvert \opt({\omega}(\bm{r'}))  - \E_{\Br}[\opt({\omega}(\bm{r})] \right\rvert \geq \frac{\delta d \ell \log(\ell)}{4}\right]\leq \frac{\eta'}{4}\end{equation*}

Hence, with probability at least $1-\frac{\eta'}{2}$,

\begin{equation*}
\left \lvert \opt({\omega}(\bm{r})) - \opt ({\omega}(\bm{r}'))\right \rvert \leq \frac{\delta d \ell \log(\ell)}{2}
\end{equation*}

Since $\delta d \ell \log(\ell)$ is a lower bound on both $OPT(\bm{\omega}(\bm{r}'))$ and $OPT(\bm{\omega}(\bm{r}))$ (by setting $y_{jk}=1/\ell$ and $x_{jk}=0$ for each $j$, $k$) we have that with probability at least $1-\frac{\eta'}{2}$,

\begin{equation*}
\frac{1}{2}\opt({\omega}(\bm{r}))\leq \opt({\omega}(\bm{r}'))\leq \frac{3}{2} \opt({\omega}(\bm{r})).
\end{equation*}

\end{proof}

Lemma~\ref{lem:find gamma} follows from Lemma~\ref{lemma:h1} and \ref{lemma:h2}.
\end{prevproof}

\section{Missing Details from Section~\ref{sec:mechanism}}\label{sec:appx_sec_5}

\begin{prevproof}{Lemma}{lem:M' BIC}
\vspace{-.1in}
\paragraph{$\MM'$ is BIC:}
We prove the Bayesian Incentive Compatibility in two parts.
The first part is similar to the proof of Lemma~\ref{lemma:IDEAL_BIC_IR}.
We argue that the distribution of the normal surrogate $s_i$ that represents agent $i$ in Phase 2, when the agent $i$ reports truthfully, is $\dD_i$. Note that for any matching Algorithm~\ref{alg:new online matching} produces, the $k$-th normal surrogate and the $k$-th $0$-surrogate together are matched to exactly $d$ replicas for every $k\in[\ell]$. As the $d\ell-1$ replicas and the agent's type are all drawn from the same distribution $\DD'_i$, we can simply treat all of them as replicas and uniformly choose one to be the agent reported type \emph{after} Algorithm~\ref{alg:new online matching} terminates. Therefore, the surrogate $s_i$ that represents the agent, will also be chosen uniformly at random between all the normal surrogates. Thus, the distribution of $s_i$ is $\dD_i$.

If all the other agents report truthfully, agent $i$'s value for winning a surrogate $s$ is exactly $W_i(t_i, s)$ if her true type is $t_i$. In other words, under the assumption that all other agents report truthfully, Mechanism~\ref{alg:mout} for agent $i$ is equivalent to a competition among replicas to win surrogates, and the edge weight between a replica and a surrogate is exactly the replica's value for winning the surrogate. To show that Mechanism~\ref{alg:mout} is BIC, it suffices to prove that at any position $\pi$,  
$$u_i(t_i,\AA_\pi(W(t'_i)))-\E_{\lambda}[q_i(t'_i,\lambda)]$$
 is maximized when the reported type $t'_i$ equals to the true type $t_i$. Here $W(t'_i)$ is simply the collection of the edge weights when $r^{(\pi)}=t'_i$, and the function $u_i(\cdot)$ is defined in Definition~\ref{def:implicit payment}. A result by Rochet~\cite{rochet1987necessary} implies that this is indeed the case. Interested readers can find a modern restatement of the result in Theorem 2.1 of~\cite{babaioff2013multi}~\footnote{To apply Theorem 2.1 of~\cite{babaioff2013multi} to our setting, one should think of each surrogate as an outcome, and the corresponding edge weight as the value for the outcome. In other words, a replica's type is the weights on the incident edges. As the matching is computed by a maximal-in-range algorithm, we can  allow the edge weights to be arbitrary numbers, and the induced allocation rule will still be implementable in an incentive compatible way. As a result, we can apply Theorem 2.1 of~\cite{babaioff2013multi} to our setting. Note that the  incentive compatible payment rule it gives is  off by an absolute constant compared to our payment rule in Definition~\ref{def:implicit payment}.}. 

\paragraph{$\MM'$ is IR:} The expected utility for agent $i$ with type $t_i$ at position $\pi$ is 
\begin{equation}\label{eq:utility}
	u_i(t_i,\AA_\pi(W))-\E_{\lambda}[q_i(t_i,\lambda)]=\int_0^1 u_i\left(t_i,\AA_{\pi}\left(\lambda W_{\pi},W_{-\pi}\right)\right)d 
	\lambda+\sqrt{\delta}(\log 2\ell +1),
\end{equation}

where $W$ is the collection of weights in $G_i$ when agent $i$ reports truthfully. We will first prove that for any $\lambda\in[0,1]$, $u_i\left(t_i,\AA_{\pi}\left(\lambda W_{\pi},W_{-\pi}\right)\right)$ is at least $-{\delta \log(2\ell)\over \lambda}$.  Denote $H(x,y)=-\sum_k(x_{k}\log(x_{k})+y_{k}\log(y_{k}))$ as the entropy for distribution $(x,y)$.

Let $(x_{\pi}'',y_\pi'')$ be $\mathcal{A}_\pi(\lambda W_{\pi},W_{-\pi})$. By Observation~\ref{obs:exp weight with dummy-every replica},
$$(x_\pi'',y_\pi'')=\argmax_{(x_\pi,y_\pi)}\sum_{k}x_{\pi k}\cdot \lambda W_{\pi k}+\delta\cdot H(x_\pi,y_\pi)-\sum_k\gamma\alpha_k^{(\pi)}\cdot (x_{\pi k}+y_{\pi k})$$

By considering an alternative solution $(0,x_\pi''+y_\pi'')$, we have
\begin{align*}
	\sum_{k}x''_{\pi k}\cdot \lambda W_{\pi k}+&\delta\cdot H(x_\pi'',y_\pi'')-\sum_k\gamma\alpha_k^{(\pi)}\cdot (x''_{\pi k}+y''_{\pi k})\\&\geq 0-\delta\cdot \sum_{k}(x_{\pi k}''+y_{\pi k}'')\log(x_{\pi k}''+y_{\pi k}'')-\sum_k\gamma\alpha_k^{(\pi)}\cdot (x_{\pi k}''+y_{\pi k}'')
\end{align*}

Since $-\sum_{k}(x_{\pi k}''+y_{\pi k}'')\log(x_{\pi k}''+y_{\pi k}'')\geq 0 $,

$$u_i\left(t_i,\AA_{\pi}\left(\lambda W_{\pi},W_{-\pi}\right)\right)=\sum_{k}x''_{\pi k}\cdot  W_{\pi k}\geq -\frac{\delta}{\lambda}
H(x_\pi'',y_\pi'')\geq {-\delta\log (2\ell)\over\lambda}.$$  
Another lower bound for $u_i\left(t_i,\AA_{\pi}\left(\lambda W_{\pi},W_{-\pi}\right)\right)$ is $-1$, as the mechanism $\MM$ is IR, by definition of the edge weight $W_{\pi k}\geq -1$ for all $k\in[\ell]$.
Now, we are ready to lower bound the utility
\begin{align*}
	\text{RHS of Equation~\eqref{eq:utility}}&\geq \int_{\sqrt{\delta}}^1 {-\delta\log (2\ell)\over\lambda} d\lambda
+ \int^{\sqrt{\delta}}_0 -1 d\lambda + \sqrt{\delta}(\log 2\ell +1)\\
&=-\delta \log (1/\sqrt{\delta})\log (2\ell)-\sqrt{\delta}+\sqrt{\delta}(\log (2\ell) +1)\\
&\geq -\sqrt{\delta}\log (2\ell)+\sqrt{\delta}\log (2\ell)\\
&=0
\end{align*}

Next, we prove that $\MM'$ does not lose too much revenue by subsidizing the agents in Phase 1.

\paragraph{$\E_{\lambda}[q_i(t'_i,\lambda)]$ is at least $-\sqrt{\delta}(\log (2\ell)+1)$:} It suffices to show that 
$$u_i\left(t_i,\AA_{\pi}\left(W \right)\right) \geq u_i\left(t_i,\AA_{\pi}\left(\lambda W_{\pi},W_{-\pi}\right)\right)$$
 for any $\lambda\in[0,1)$. We still use $(x_{\pi}'',y_\pi'')$ to denote $\mathcal{A}_\pi(\lambda W_{\pi},W_{-\pi})$ and  $(\hx_\pi,\hy_\pi)$  to denote $\AA_{\pi}\left(W \right)$.

By Observation~\ref{obs:exp weight with dummy-every replica}, both allocations are maximal-in-range for the same dual variables $\alpha^{(\pi)}$. Hence, the following two inequalities are true.

$$\sum_k \hx_{\pi k}W_{\pi k}+\delta\cdot H(\hx_\pi,\hy_\pi)-\sum_k\gamma\alpha_k^{(\pi)}\cdot (x_{\pi k}'+y_{\pi k}')\geq \sum_kx_{\pi k}''W_{\pi k}+\delta\cdot H(x_\pi'',y_\pi'')-\sum_k\gamma\alpha_k^{(\pi)}\cdot (x_{\pi k}''+y_{\pi k}'')$$

$$\sum_kx_{\pi k}''\lambda W_{\pi k}+\delta\cdot H(x_\pi'',y_\pi'')-\sum_k\gamma\alpha_k^{(\pi)}\cdot (x_{\pi k}''+y_{\pi k}'')\geq \sum_k \hx_{\pi k}\lambda W_{\pi k}+\delta\cdot H(\hx_\pi,\hy_\pi)-\sum_k\gamma\alpha_k^{(\pi)}\cdot (\hx_{\pi k}+\hy_{\pi k})$$

Summing up the two inequalities together, we have
$$\sum_k(\hx_{\pi k}-x_{\pi k}'')W_{\pi k}(1-\lambda)\geq 0.$$

Since $\lambda\in [0,1)$,
$$u_i\left(t_i,\AA_{\pi}\left(W \right)\right) - u_i\left(t_i,\AA_{\pi}\left(\lambda W_{\pi},W_{-\pi}\right)\right)=\sum_k(\hx_{\pi k}-x_{\pi k}'')W_{\pi k}\geq 0.$$
 
Finally, we analyze the time and query complexity of the mechanism.

\paragraph{Time and Query Complexity:} All steps except Step 3,6, and 8 clearly has $\poly(d,\ell)$ time and query complexity. According to Lemma~\ref{lem:find gamma}, Step 3 has $\poly(d,\ell,1/\eta',1/\delta)$ time and query complexity. Since $\gamma$ is guaranteed to be at most $\max\{\ell, \delta\ell \log \ell\}$, Algorithm~\ref{alg:new online matching} in Step 6 has time and query complexity $\poly(d,\ell,1/\delta)$ according to Theorem~\ref{thm:arbitrary-weight-algorithm}. From Definition~\ref{def:implicit payment}, it is clear that Step 8 also has time and query complexity at most $\poly(d,\ell,1/\delta)$. Hence, the mechanism $\MM'$ has time and query complexity $\poly(d,\ell,1/\eta',1/\delta)$.
\end{prevproof}

\subsection{Bounding the Revenue Loss}\label{sec:rev M'}

Now we give the proof of Theorem~\ref{thm:general reduction_2}. It suffices to  lower bound the revenue of $\MM'$ from the second phase due to  Lemma~\ref{lem:M' BIC}. 
To bound the revenue loss, we define some new notations.  

\vspace{-.1in}

\paragraph{New Notations:}
For every agent $i$, and the corresponding bipartite graph $G_i$, we define a new bipartite graph $\hG_i$ whose left hand side nodes are the replicas/LHS-nodes of $G_i$. For each normal surrogate/RHS-node of $G_i$, we duplicate it $d$ times to form the set of right hand side nodes of $\hG_i$. For the $k$-th surrogate in $G_i$, the $(a\ell+k)$-th surrogate in $\hG_i$ is one of its copies for all $0\leq a\leq d-1$. We do not copy the $0$-surrogates to $\hG_i$. The edge weights in $\hG_i$  are still defined using $W_i(\cdot,\cdot)$. Clearly, every $d$-to-1 matching in $G_i$ corresponds to a $1$-to-$1$ matching in $\hG_i$. If replica $r$ is matched to a surrogate $s$ in $G_i$, simply match $r$ to the first available copy of $s$ in $\hG_i$. We use $\ell'$ to denote $d\ell$, and $\hG_i$ has $2\ell'$ nodes. When we say the matching in $\hG_i$ produced by Algorithm~\ref{alg:new online matching}, we mean the matching in $\hG_i$ that corresponds to the matching produced by Algorithm~\ref{alg:new online matching} in $G_i$. We follow the convention to use $\rj$ to denote the type of the $j$-th replica and $\sk$ to denote the type of the $k$-th surrogate in $\hG_i$. We further simplify the notation and use $p_i(t_i)$ to denote $\E_{t_{-i}\sim \DD_{-i}}[p_i(t_i,t_{-i})]$ for any type $t_i\in \TT_i$.

Given the replica profile $\Br$ and surrogate profile $\Bs$, for any matching $L(\Br,\Bs)$ in $\hG_i$, we slightly abuse notation to use  $W_i(L(\Br,\Bs))$ to denote $\sum_{(r,s)\in L(\Br,\Bs)}W_i(r,s)$. When the replica profile $\Br$ and surrogate profile $\Bs$ are clear from context, we simply use $W_i(L)$ to denote the total weight of the matching $L$. Since the analysis mainly concerns the set of surrogates that are matched in a matching, we use $s\in L(\Br,\Bs)$ to denote that the surrogate $s$ is matched in $L(\Br,\Bs)$. Let $O(\Br,\Bs)$ be the (randomized) matching obtained by Algorithm~\ref{alg:new online matching} on $\hG_i$, $V(\Br,\Bs)$ be the maximum weight matching  in  $\hG_i$. 

We first provide a Lemma that relates the expected revenue of $\MM'$ to the size of the matchings.
\begin{lemma}\label{lem:mout lower bound}  Let  $\revsecond_i(\MM',\DD')$ be the expected revenue of $\MM'$ from agent $i$ in Phase 2, $\rev_i(\MM,\DD)$ be the expected revenue of $\MM$ from agent $i$,\\
$$\textstyle \revsecond_i(\mout,\DD')\geq (1-\eta)\cdot\E_{\Br,\Bs}\left[\sum_{\sk\in O(\Br,\Bs)}p_i(\sk)/\ell'\right],$$ and
$$\textstyle \rev_i(\MM,\DD)=\E_{\Bs}\left[\sum_{k\in [\ell']}p_i(\sk)/\ell'\right]$$
\end{lemma}
\begin{proof}
For every agent $i$, 
 only when the agent $i$ is matched to a surrogate in $ O(\Br,\Bs)$, she pays the surrogate price. We can again first sample  $\Br$ and $\Bs$, and run Algorithm ~\ref{alg:new online matching} on the corresponding graph $\hG_i$ to find the matching $O(\Br,\Bs)$, then choose a replica uniformly at random to be agent $i$. Since each replica has exactly probability $1/\ell'$ to be agent $i$, each surrogate in $O(\Br,\Bs)$ is selected with probability $1/\ell'$, the expected revenue paid by agent $i$ is exactly $(1-\eta)\cdot\E_{\Br,\Bs}\left[\sum_{\sk\in O(\Br,\Bs)} \frac{ p_i(\sk)}{\ell'}\right]$.  
The expected payment from agent $i$ in $\MM$ is $\E_{t_i\sim \DD_i}\left[p_i(t_i)\right]$. Since each $\sk$ is drawn from $\DD_i$, this is exactly the same as $\E_{\Bs}\left[\sum_{k\in [\ell']}p_i(\sk)/\ell'\right]$.
\end{proof}

In Lemma~\ref{lem:payment K to V simplified}, we bound the gap between  $\revsecond_i(\MM',\DD')$ and $\rev_i(\MM,\DD)$. Indeed, we prove a stronger result that holds for any matching $K(\Br,\Bs)$ that has close to maximum total weight.

\begin{lemma}\label{lem:payment K to V simplified}
Recall that $V(\Br,\Bs)$ is the maximum weight matching  in  $\hG_i$. Let  $$\E_{\Br,\Bs}\left[W_i\left(K\left( \Br,\Bs\right)\right)\right]= \E_{\Br,\Bs}\left[W_i\left(V\left( \Br,\Bs\right)\right)\right]-\Delta.$$ We have
$$(1-\eta)\E_{\Br,\Bs}\left[\sum_{\sk\in K(\Br,\Bs)}p_i(\sk)/\ell'\right]\geq \E_{\Bs}\left[\sum_{k\in [\ell']}p_i(\sk)/\ell'\right]-n\left[\eta+\sqrt{\frac{|\TT_i'|}{\ell}}+\frac{\varepsilon}{\eta}+\frac{\Delta}{d\ell\eta}\right]-\frac{2}{\eta}\wass(\mathcal{D}_i,\mathcal{D}_i').$$
\end{lemma}

To prove Theorem~\ref{thm:general reduction_2}, one only needs to choose $K(\Br,\Bs)$ to be the matching $O(\Br,\Bs)$ produced Algorithm~\ref{alg:new online matching}, and combine the guarantees in Lemma~\ref{lem:M' BIC} and ~\ref{lem:mout lower bound}.

Instead of proving Lemma~\ref{lem:payment K to V simplified}, we prove the following strengthened version of the statement.

\begin{lemma}\label{lem:payment K to V}
\notshow{Let $K$ be any randomized matching between replicas and surrogates. Suppose for all $\Br,\Bs$, $\E[W_i(K)]\geq W_i(V)-\Delta$, then we have
$$(1-\eta)\sum_i\E_{\Br,\Bs}\left[\sum_{\sk\in K(\Br,\Bs)}p_i(\sk)/\ell'\right]\geq \sum_i\E_{\Bs}\left[\sum_{k\in [\ell']}p_i(\sk)/\ell'\right]-n\left[\eta+\sqrt{\frac{T}{\ell}}+\frac{\varepsilon}{\eta}+\frac{\Delta}{d\ell\eta}\right]-\frac{2}{\eta}\wass(\mathcal{D},\mathcal{D}')$$
where $\wass(\mathcal{D},\mathcal{D}')$ is the sum of the Wasserstein Distance between every $\mathcal{D}_i$ and $\mathcal{D}_i'$. }

Let  $\E_{\Br,\Bs}\left[W_i\left(K\left( \Br,\Bs\right)\right)\right]= \E_{\Br,\Bs}\left[W_i\left(V\left( \Br,\Bs\right)\right)\right]-\Delta$. We have
$$(1-\eta)\E_{\Br,\Bs}\left[\sum_{\sk\in K(\Br,\Bs)}p_i(\sk)/\ell'\right]\geq \E_{\Bs}\left[\sum_{k\in [\ell']}p_i(\sk)/\ell'\right]-\left(\eta+\sqrt{\frac{|\TT_i'|}{\ell}}+\frac{\varepsilon}{\eta}+\frac{\Delta}{d\ell\eta}\right)-\frac{2}{\eta}\wass(\mathcal{D}_i,\mathcal{D}_i').$$

Moreover, for any coupling $c_i(\cdot)$ that $v_i$ is non-increasing w.r.t. $c_i(\cdot)$, the last term can be improved to 
$$-\frac{1}{\eta \ell'}\E_{\Br,\Bs}\left[\sum_{(\rj,\sk)\in K(\Br,\Bs)} \E_{c_i(\Br),~t_{-i}\sim \DD_{-i}}\left[v_i(\rj,x(\sk,t_{-i}))-v_i(c_i(\rj),x(\sk,t_{-i}))\right]\right].\footnote{For the rest of the proof, when we use the notation $\E_{c_i(\Br)}[\cdot]$, we are taking the expectation over the randomness of the coupling. The $c_i(\Br)=\left (c_i(\rj)\right )_{j\in[\ell']}$ inside the expectation is the realized type after coupling.}$$ 
\end{lemma}

\begin{prevproof}{Lemma}{lem:payment K to V}
To prove the statement, we consider an arbitrary coupling $c_i(\cdot,\cdot)$ of distribution $\mathcal{D}'_i$ and $\mathcal{D}_i$ 
. For every replica $r\in \TT_i'$, $c_i(r)$ is a random type from  $\mathcal{T}_i$. For every realization of the types $c_i(\Br)=\left(c_i(\rj)\right)_{j\in [\ell']}$, we consider the maximal matching that  matches a replica $\rj$ with a surrogate $\sk$ only if $c_i(\rj)=\sk$. We denote the matching as $L(c_i(\Br),\Bs)$ and refer to it as the \emph{maximal coupled same-type matching}. In the next Lemma, we argue that in expectation of $\Br$, $\Bs$ and the realization of $c_i(\Br)$, the expected size of $L(c_i(\Br),\Bs)$ is close to $\ell'$.

\begin{lemma}\label{lem:same type matching}
	For any $\Br$, $\Bs$, and realization of $c_i(\Br)=\left (c_i(\rj)\right )_{j\in [\ell']}$, let $L(c_i(\Br),\Bs)$ be a maximal coupled same-type matching, then $$\E_{\Br,\Bs,c_i(\Br)}\left[\left\lvert L(c_i(\Br),\Bs)\right\rvert\right]\geq \ell'-\sqrt{d\left\lvert \TT_i'\right\rvert\cdot \ell'}.$$
\end{lemma}

\begin{proof}
To prove the result, we first invoke the following Lemma.
	\begin{lemma}[Adapted from \cite{HartlineKM11}]\label{lem:equal type matching}
Let $\Br'$ be $N$ replicas drawn i.i.d. from distribution $\DD'_i$, and $\Bs'$ be $N$ surrogates  drawn i.i.d. from distribution $\DD_i$. For any coupling $c_i(\cdot)$ between $\DD_i'$ and $\DD_i$, the expected cardinality of a maximal matching that only matches a replica $r$ and a surrogate $s$ when $c_i(r)=s$  is at least $N-\sqrt{|\mathcal{T}_i'|\cdot N}$. The expectation is over the randomness of $\Br'$, $\Bs'$, and the coupling $c_i(\Br')$. \end{lemma}
Although we have $\ell'$ replicas and $\ell'$ surrogates, we cannot directly apply Lemma~\ref{lem:equal type matching}, as the surrogates are not i.i.d. samples from $\DD_i$. Instead, we partition $\hG_i$ into $d$ subgraphs. The $\alpha$-th subgraph contains all replicas $\rj$ and surrogates $\sk$ with $j$ and $k$ lie in $[\alpha \ell+1,(\alpha+1)\ell]$. If we only consider the $\alpha$-th subgraph, due to our construction of $\hG_i$, the replicas are all sampled i.i.d. from $\DD_i'$ and the surrogates are also sampled i.i.d. from $\DD_i$. Therefore, Lemma~\ref{lem:equal type matching} implies that a maximal coupled same-type matching in the $\alpha$-th subgraph has expected size at least $\ell-\sqrt{|\TT_i'|\ell}$. Since there are $d$ subgraphs, so the expected size of a maximal coupled same-type matching is at least $\ell'-\sqrt{d\left\lvert \TT_i'\right\rvert\cdot \ell'}$.
\end{proof}

Now, it suffices to argue that the total payment from surrogates that are in $L(c_i(\Br),\Bs)$ but not in $K(\Br,\Bs)$ is small. Indeed, when $K(\Br,\Bs)$ is the maximum weight matching, one can directly prove the claim. However, $K(\Br,\Bs)$ only has approximately maximum weight, and it appears to be difficult to directly compare $K(\Br,\Bs)$ with $L(c_i(\Br),\Bs)$. Instead, we construct an auxiliary matching based on both $K(\Br,\Bs)$ and $L(c_i(\Br),\Bs)$. For any $\Br,\Bs$ and realization of types $\left (c_i(\rj)\right)_{j\in [\ell']}$, we decompose the union of these two matchings into a set of disjoint alternating paths and cycles. Every surrogate that appears in $L(c_i(\Br),\Bs)$ but not in $K(\Br,\Bs)$ must be an endpoint of some alternating path. These alternating paths have the following two forms:

\begin{enumerate}[label=(\alph*).]
    \item It starts with a surrogate in $L(c_i(\Br),\Bs)\backslash K(\Br,\Bs)$ and ends with a surrogate in $K(\Br,\Bs)\backslash L(c_i(\Br),\Bs)$ with the form $\left(s_{(1)},r_{(1)},s_{(2)},r_{(2)},...,r_{(a)},s_{(a+1)}\right)$.
    \item It starts with a surrogate in $L(c_i(\Br),\Bs)\backslash K(\Br,\Bs)$ and ends with a replica with the form\\ $\left(s_{(1)},r_{(1)},s_{(2)},r_{(2)},...,s_{(a)},r_{(a)}\right)$.
\end{enumerate}

We use \textbf{$\PP$ to denote the set of all alternating paths of form (a) and (b)}. We construct a new matching $K'(c_i(\Br),\Bs)$ as follows: start with the matching $K(\Br,\Bs)$, for any alternating path $P$ of form (a) and (b), swap the edges in $K(\Br,\Bs)$ with the ones in $L(c_i(\Br),\Bs)$, that is, replace all edges in $P\cap K(\Br,\Bs)$ with edges in $P\cap L(c_i(\Br),\Bs)$. Since all the alternating paths are disjoint,  $K'(c_i(\Br),\Bs)$ is indeed a matching.

\begin{corollary}\label{cor:bound revenue by L}
$$\E_{\Br,\Bs,c_i(\Br)}\left[\sum_{\sk\in K'(c_i(\Br),\Bs)} p_i(\sk)/\ell'\right]\geq \E_{\Bs}\left[\sum_{k\in[\ell']}  p_i(\sk)/\ell'\right]- \sqrt{\frac{|\mathcal{T}_i'|}{\ell}}.$$
\end{corollary}
\begin{proof}
Fix  $\Br,\Bs$ and types $\{c_i(\rj)\}_{j\in [\ell']}$. For any alternating path $P$ of form either (a) or (b), $P\cap L(c_i(\Br),\Bs)$ is the same as $P\cap K'(c_i(\Br),\Bs)$. For other alternating paths, the matched surrogate in $P\cap L(c_i(\Br),\Bs)$ is a subset of $P\cap K(\Br,\Bs)$. Thus the number of the matched surrogates in $K'(c_i(\Br),\Bs)$ is at least $\lvert L(c_i(\Br),\Bs)\rvert$. By Lemma~\ref{lem:same type matching}, $\E_{\Br,\Bs,c_i(\Br)}[|\{k:\sk\not\in K'(c_i(\Br),\Bs)\}|]\leq \sqrt{d|\mathcal{T}_i'|\cdot \ell'}$. As $\MM$ is IR, $p_i(s)\leq 1$ for any surrogate $s\in \TT_i$. Therefore, $$\E_{\Br,\Bs,c_i(\Br)}\left[\sum_{\sk\in K'(c_i(\Br),\Bs)} p_i(\sk)/\ell'\right]\geq \E_{\Bs}\left[\sum_{k\in[\ell']}  p_i(\sk)/\ell'\right]- \sqrt{\frac{|\mathcal{T}_i'|}{\ell}}.$$\end{proof}

Equipped with Corollary~\ref{cor:bound revenue by L}, we only need to compare $K(\Br,\Bs)$ with $K'(c_i(\Br),\Bs)$.
\begin{lemma}\label{lem:payment K to L}
$$\E_{\Br,\Bs}\left[\sum_{\sk\in K(\Br,\Bs)} p_i(\sk)/\ell'\right]\geq \E_{\Br,\Bs,c_i(\Br)}\left[\sum_{\sk\in K'(c_i(\Br),\Bs)} p_i(\sk)/\ell'\right]-\frac{1}{\eta}(\varepsilon+\frac{\Delta}{\ell'})-\frac{2}{\eta}\wass(\mathcal{D}_i,\mathcal{D}_i')$$
Moreover, for any coupling $c_i$ such that $v_i$ is non-increasing w.r.t. $c_i$, the last term can be improved to\\
$-\frac{1}{\eta}\E_{t\sim \mathcal{D}'}[v_i(t_i,x'(t))-v_i(c_i(t_i),x'(t))]$.
\end{lemma}
\begin{proof}

Fix any $\Br,\Bs$ and realization of $c_i(\Br)$. Observe that if we decompose the union of $K(\Br,\Bs)$ and $K'(c_i(\Br),\Bs)$ into alternating path and cycles, we will end up with many length $2$ cycles and all the alternating paths in $\PP$. Hence, we only need to consider the paths in $\PP$.

Consider any $k\in [a]$ if the path has form (a) (or $k\in [a-1]$ if the path has form (b)), note that $c_i(r_{(k)})=s_{(k)}$, as this is also an edge in the matching $L(c_i(\Br),\Bs)$. Since $\MM$ is $\varepsilon$-BIC, we have

$$\E_{t_{-i}\sim \DD_{-i}}[v_i(c_i(r_{(k)}),x(s_{(k)},t_{-i}))]-p_i(s_{(k)})\geq \E_{t_{-i}\sim \DD_{-i}}[v_i(c_i(r_{(k)}),x(s_{(k+1)},t_{-i}))]-p_i(s_{(k+1)})-\varepsilon,$$

which is equivalent to
\begin{equation}\label{equ:epsilon-bic}
\begin{aligned}
&W_i(r_{(k)},s_{(k)})-W_i(r_{(k)},s_{(k+1)})\\
&\geq-\varepsilon-\eta\cdot (p_i(s_{(k+1)})-p_i(s_{(k)}))+\Delta_{i,c_i}(r_{(k)},s_{(k)})-\Delta_{i,c_i}(r_{(k)},s_{(k+1)})
\end{aligned}
\end{equation}

where $\Delta_{i,c_i}(r,s)=\E_{t_{-i}\sim \DD_{-i}}[v_i(r,x(s,t_{-i}))-v_i(c_i(r),x(s,t_{-i}))]$.

By summing up Inequality~\eqref{equ:epsilon-bic} for each $k$, we are able to relate the difference of the total weight between $ K(\Br,\Bs)$ and $K'(c_i(\Br),\Bs)$ with the total payment from surrogates in $ K(\Br,\Bs)$ and $K'(c_i(\Br),\Bs)$.
\begin{itemize}
    \item For any form (a) path,
    \begin{align*}
       & \sum_{k=1}^a \left(W_i(r_{(k)},s_{(k+1)})-W_i(r_{(k)},s_{(k)})\right )\\
        \leq & a\cdot\varepsilon+\eta\cdot(p_i(s_{(a+1)})-p_i(s_{(1)}))-\sum_{k=1}^a\left (\Delta_{i,c_i}(r_{(k)},s_{(k)})-\Delta_{i,c_i}(r_{(k)},s_{(k+1)})\right)
    \end{align*}
    \item For any form (b) path
    \begin{align*}
        &\sum_{k=1}^{a-1} \left(W_i(r_{(k)},s_{(k+1)})-W_i(r_{(k)},s_{(k)})\right)-W_i(r_{(a)},s_{(a)})\\
        \leq& (a-1)\cdot\varepsilon+\eta\cdot\sum_{k=1}^{a-1}\left(p_i(s_{(k+1)})-p_i(s_{(k)})\right)-W_i(r_{(a)},s_{(a)})-\sum_{k=1}^{a-1}\left(\Delta_{i,c_i}(r_{(k)},s_{(k)})-\Delta_{i,c_i}(r_{(k)},s_{(k+1)})\right)\\
        \leq&  
        (a-1)\cdot\varepsilon-\eta\cdot p_i(s_{(1)})-\sum_{k=1}^{a}\Delta_{i,c_i}(r_{(k)},s_{(k)})+\sum_{k=1}^{a-1}\Delta_{i,c_i}(r_{(k)},s_{(k+1)})
    \end{align*}
   The last inequality is because 
$\eta\cdot p_i(s_{(a)})-W_i(r_{(a)},s_{(a)})\leq -\Delta_{i,c_i}(r_{(a)},s_{(a)})$, which is implied by the fact that $\MM$ is IR.
\end{itemize}

To sum up, for any alternating path $P\in \PP$,\begin{equation} \label{equ:diff}\begin{aligned}
&\sum_{(\rj,\sk)\in P \cap K(\Br,\Bs)}W_i(\rj,\sk)-\sum_{(\rj,\sk)\in P\cap K'(c_i(\Br),\Bs)}W_i(\rj,\sk)\\
&\leq \lvert P\cap K(\Br,\Bs)\rvert\cdot\varepsilon+\eta\cdot \left[\sum_{\sk\in P\cap K(\Br,\Bs)}p_i(\sk)-\sum_{\sk\in P\cap K'(c_i(\Br),\Bs)}p_i(\sk)\right]+\diff(P),
\end{aligned}
\end{equation}

where $\diff(P)=\sum_{(\rj,\sk)\in P\cap K(\Br,\Bs)}\Delta_{i,c_i}(\rj,\sk)-\sum_{(\rj,\sk)\in P\cap K'(c_i(\Br),\Bs)}\Delta_{i,c_i}(\rj,\sk)$.

Since $V(\Br,\Bs)$ is the maximum weight matching, we have

\begin{equation}\label{equ:weight K vs V}
\begin{aligned}
&\sum_{P \in \PP}\left[\sum_{(\rj,\sk)\in P \cap K(\Br,\Bs)}W_i(\rj,\sk)-\sum_{(\rj,\sk)\in P\cap K'(c_i(\Br),\Bs)}W_i(\rj,\sk) \right]\\
=&W_i(K(\Br,\Bs))-W_i(K'(c_i(\Br),\Bs))\\
 \geq & W_i(K(\Br,\Bs))-W_i(V(\Br,\Bs))
\end{aligned}
\end{equation}

Note that if we are using the matching $L(c_i(\Br),\Bs)$ instead of $K'(c_i(\Br),\Bs)$, we can no longer prove Inequality~\eqref{equ:weight K vs V}. The reason is quite subtle. It is possible that $L(c_i(\Br),\Bs)$ has much higher weight than $K(\Br,\Bs)$ on paths in $\PP$,  but much smaller weight on the rest alternating path and cycles. In that case, the first equal sign will be replaced by a less equal sign, which makes the inequality meaningless. By comparing to $K'(c_i(\Br),\Bs)$, we can avoid this issue.

Combining Inequality~\eqref{equ:diff} and~\eqref{equ:weight K vs V} , we have
\begin{align*}
    &W_i(K(\Br,\Bs))-W_i(V(\Br,\Bs))\\
    \leq& 
    \sum_{P \in \PP}\left[\sum_{(\rj,\sk)\in P \cap K(\Br,\Bs)}W_i(\rj,\sk)-\sum_{(\rj,\sk)\in P\cap K'(c_i(\Br),\Bs)}W_i(\rj,\sk) \right]\\
    \leq&
     \sum_{P \in \PP}\left[\lvert P\cap K(\Br,\Bs)\rvert\cdot\varepsilon+\eta\cdot \left[\sum_{\sk\in P\cap K(\Br,\Bs)}p_i(\sk)-\sum_{\sk\in P\cap K'(c_i(\Br),\Bs)}p_i(\sk)\right]+\diff(P)\right]\\
    \leq& \ell'\cdot \varepsilon+\eta\cdot \left[\sum_{\sk\in K(\Br,\Bs)}p_i(\sk)-\sum_{\sk\in K'(c_i(\Br),\Bs)}p_i(\sk)]\right]+\sum_{P \in \PP}\diff(P)
\end{align*}


Finally, we take expectation over $\Br$, $\Bs$, and  $c_i(\Br)$. \begin{align*}
&\E_{\Br,\Bs}\left[\sum_{\sk\in K(\Br,\Bs)} p_i(\sk)\right]-\E_{\Br,\Bs,c_i(\Br)}\left[\sum_{\sk\in K'(c_i(\Br),\Bs)} p_i(\sk)\right]\\
\geq& \frac{1}{\eta}\left(-\ell'\cdot\varepsilon+\E_{\Br,\Bs}\left[W_i\left(K\left( \Br,\Bs\right)\right)\right]- \E_{\Br,\Bs}\left[W_i\left(V\left( \Br,\Bs\right)\right)\right]-\E_{\Br,\Bs,c_i(\Br)}\left[ \sum_{P \in \PP}\diff(P)\right]\right)\\
\geq& -\frac{1}{\eta}(\ell'\cdot\varepsilon+\Delta)-\frac{1}{\eta}\E_{\Br,\Bs,c_i(\Br)}\left[ \sum_{P \in \PP}\diff(P)\right]\end{align*}

For every type $r$, $s$, and realized type $c_i(r)$, $\Delta_{i,c_i}(r,s)=\E_{t_{-i}\sim \DD_{-i}}[v_i(r,x(s,t_{-i}))-v_i(c_i(r),x(s,t_{-i}))]\in [-\dist_i(r,c_i(r)),\dist_i(r,c_i(r))]$ (recall that $\dist_i(r,c_i(r))=\max_{o\in \cO}|v_i(r,o)-v_i(c_i(r),o)|$). Thus $$\sum_{P \in \PP}\diff(P)
\leq 2\sum_{j=1}^{\ell'}\dist_i(\rj,c_i(\rj)),$$ and $$\E_{\Br,\Bs,c_i(\Br)}\left[ \sum_{P \in \PP}\diff(P)\right]\leq 2\ell' \wass(\DD_i,\DD_i') .$$
Therefore, $$\E_{\Br,\Bs}\left[\sum_{\sk\in K(\Br,\Bs)} p_i(\sk)/\ell'\right]\geq \E_{\Br,\Bs,c_i(\Br)}\left[\sum_{\sk\in K'(c_i(\Br),\Bs)} p_i(\sk)/\ell'\right]-\frac{1}{\eta}(\varepsilon+\frac{\Delta}{\ell'})-\frac{2}{\eta}\wass(\mathcal{D}_i,\mathcal{D}_i').$$

If $v_i$ is non-increasing w.r.t. $c_i$, then $\Delta_{i,c_i}(\cdot,\cdot)$ is a non-negative function. Then

$$ \sum_{P \in \PP}\diff(P)\leq \sum_{P \in \PP} \sum_{(\rj,\sk)\in P\cap K(\Br,\Bs)}\Delta_{i,c_i}(\rj,\sk)\leq \sum_{(\rj,\sk)\in K(\Br,\Bs)}\Delta_{i,c_i}(\rj,\sk),$$
and

$$\E_{\Br,\Bs,c_i(\Br)}\left[ \sum_{P \in \PP}\diff(P)\right]\leq \E_{\Br,\Bs}\left[\sum_{(\rj,\sk)\in K(\Br,\Bs)} \E_{c_i(\Br),~t_{-i}\sim \DD_{-i}}\left[v_i(\rj,x(\sk,t_{-i}))-v_i(c_i(\rj),x(\sk,t_{-i}))\right]\right].$$

\end{proof}

\vspace{.2in}

Finally, we are ready to prove Lemma~\ref{lem:payment K to V}. Note that for every $\sk$, $p_i(\sk)\leq 1$ since $\MM$ is -IR. We have
$$(1-\eta)\cdot \E_{\Br,\Bs}\left[\sum_{\sk\in K(\Br,\Bs)}p_i(\sk)/\ell'\right]\geq \E_{\Br,\Bs}\left[\sum_{\sk\in K(\Br,\Bs)}p_i(\sk)/\ell'\right]-\eta.$$
The lemma follows from Lemma~\ref{lem:payment K to L} and Corollary~\ref{cor:bound revenue by L}.
\end{prevproof}


\begin{prevproof}{Theorem}{thm:general reduction_2}
First, by Lemma~\ref{lem:mout lower bound}, we can lower bound the revenue of $\MM'$ under $\DD'$ from agent $i$ in Phase 2 $\revsecond_i(\MM',\DD')$ by $(1-\eta)\E_{\Br,\Bs}\left[\sum_{\sk\in O(\Br,\Bs)}p_i(\sk)/\ell'\right]$, where $O(\Br,\Bs)$ is the matching produced by Algorithm~\ref{alg:new online matching} on $\hG_i$. Lemma~\ref{lem:mout lower bound} also provides an equivalent expression for the revenue of $\MM$ under $\DD$ from agent $i$: $\rev_i(\MM,\DD)=\E_{\Bs}\left[\sum_{k\in [\ell']}p_i(\sk)/\ell'\right]$. 

We choose the parameters according to Theorem~\ref{thm:arbitrary-weight-algorithm}, that is, for any $\psi\in(0,1)$, we set $\delta=\Theta({\psi\over \log \ell})$, $\eta'=\Theta(\psi)$ and $d\geq {\ell \log \ell\over \eta'^2}$. Theorem~\ref{thm:arbitrary-weight-algorithm} implies that $\E_{\Br,\Bs}\left[W_i\left(O\left( \Br,\Bs\right)\right)\right]= \E_{\Br,\Bs}\left[W_i\left(V\left( \Br,\Bs\right)\right)\right]-O(d\ell\psi)$, that is in expectation $O\left( \Br,\Bs\right)$ has close to maximum weight. We will specify the choice of the other parameters $\ell$, $\eta$, and $\psi$ later. By Lemma~\ref{lem:payment K to V}, we know that 

\begin{equation}\label{eq:revenue gap between each agent}
    \revsecond_i (\MM',\DD') \geq \rev_i(\MM,\DD) -  \left( \eta + \sqrt{\frac{|\TT_i'|}{\ell}} + \frac{\varepsilon}{\eta} + \frac{O(\psi)}{\eta} \right) - \frac{2 d_w(\dD_i,\dD_i')}{\eta}
\end{equation}

Combining Inequality~\eqref{eq:revenue gap between each agent} with Lemma~\ref{lem:M' BIC}, we can obatain the following lower bound on $\rev(\MM',\DD')$.
\begin{equation}
    \rev(\MM',\DD') \geq \rev(\MM,\DD) - \sum_{i\in[n]} \left( \eta + \sqrt{\frac{T'}{\ell}} + \frac{\varepsilon}{\eta} + \frac{O(\psi)}{\eta} \right) - \frac{2 d_w(\dD,\dD')}{\eta}-n\sqrt{\delta}(\log (2\ell)+1)
\end{equation}


Now we set $\ell = \frac{T'}{\varepsilon}$, $\psi={\varepsilon^2\over \log \ell}$, and we can choose $\eta$ to be $O\left(\sqrt{\varepsilon+{d_w(\dD,\dD' )\over n}}\right)$ so that 

\begin{align*}
    \rev(\MM',\DD') &\geq  \rev(\MM,\DD) - O\left(n \left(\sqrt{\varepsilon}+\sqrt{\varepsilon+{d_w(\dD,\dD' )\over n}} \right)\right)-\Theta(n\varepsilon)\\
    &\geq \rev(\MM,\DD) -  O(n\sqrt{\varepsilon})-O\left(\sqrt{{n\cdot d_w(\dD,\dD' )}}\right) \end{align*}

Plugging in our choice of the parameters to Lemma~\ref{lem:M' BIC}, we can conclude that both the computational and query complexity of $\MM'$ is $\poly(n,T',1/\varepsilon)$.

If $c_i(\cdot)$ that $v_i$ is non-increasing w.r.t. $c_i(\cdot)$, we can replace the last term $- \frac{2d_w(\dD_i,\dD_i')}{\eta}$ in Inequality~\eqref{eq:revenue gap between each agent} by $$-\frac{1}{\eta \ell'}\E_{\Br,\Bs}\left[\sum_{(\rj,\sk)\in O(\Br,\Bs)} \E_{c_i(\Br),~t_{-i}\sim \DD_{-i}}\left[v_i(\rj,x(\sk,t_{-i}))-v_i(c_i(\rj),x(\sk,t_{-i}))\right]\right].$$ Note that this quantity is the same as $$-\frac{1}{\eta}\E_{t\sim \DD'}\left[\E_{c_i(t_i)}\left[v_i\left(t_i,x'(t)\right)-v_i\left(c_i(t_i),x'(t)\right)\right]\right].$$ Hence, for any $\eta\in (0,1)$, we can improve the result to 
$$ \rev(\MM',\DD')\geq \rev(\MM,\DD)-n\sqrt{\varepsilon}-O\left(n\eta+{n\varepsilon\over \eta}\right)-{\sum_{i\in [n]}\E_{t\sim \DD'}\left[\E_{c_i(t_i)}\left[v_i\left(t_i,x'(t)\right)-v_i\left(c_i(t_i),x'(t)\right)\right]\right]\over \eta}.$$

\end{prevproof}

\notshow{

\subsection{BIC, IR, and Implicit Payment Computation}\label{appx:implicit payment computation}\label{sec:BIC and IR}

}

\notshow{

In this section we will describe how to compute implicit first-phase payment for every agent. The procedure for the single-parameter setting has been used by \cite{archer2004approximate} and \cite{HartlineL10}, which makes $n+1$ calls to the mechanism. A different implicit payment computation procedure by Babioff et al. \cite{babaioff2013multi,babaioff2015truthful} only requires one call in expectation, but incurs a $1-\epsilon$ loss in expected welfare. 

Now let's fix agent $i$, and treat every replica (including agent $i$) as a ``bidder'' in an auction and every surrogate as an ``item''. According to Algorithm~\ref{alg:mout}, each replica $j$ has type $r_j\in\mathcal{T}_i'$ and has value $W_i(r_j,s)$ for matching to a surrogate with type $s$. Let $\mathcal{A}$ be the allocation of matching surrogates to replicas in $\mout$ (Algorithm~\ref{alg:mout}). In order to introduce the procedure, we need the following definition.

\begin{definition}
The type space $\mathcal{T}_i'$ has rescalable types (or called star-convex at 0) if: for every $r\in \mathcal{V}_i$, and any scalar $\lambda\in [0,1]$, there exists a type $\lambda r\in \mathcal{T}_i'$ such that for any surrogate with type $s\in \mathcal{T}_i'$, 
$W_i(\lambda r,s)=\lambda W_i(r,s)$.
\end{definition}

}

\notshow{

\begin{definition}
The type space $\mathcal{T}_i'$ has rescalable types (or called star-convex at 0) if: for every $r\in \mathcal{V}_i$, and any scalar $\lambda\in [0,1]$, there exists a type $\lambda r\in \mathcal{T}_i'$ such that for any surrogate with type $s\in \mathcal{T}_i'$, 
$W_i(\lambda r,s)=\lambda W_i(r,s)$.
\end{definition}

The following theorem shows how the truthful payment rule is computed for any implementable allocation rule.

The implicit payment computation procedure assumes that $\mathcal{T}_i'$ has rescalable types. In order to satisfy this assumption, we need to first generalize $\mathcal{A}$ to allow input with scaled types. For every $r\in \mathcal{T}_i'$ and $\lambda\in [0,1)$, define $W_i(\lambda r,s)=\lambda W_i(r,s)$ for every $s\in \mathcal{T}_i'$. Now $\mathcal{A}$ is also well-defined for scaled types: it returns the online matching with respect to the weight $W_i$. Notice that we are not enlarging the type space $\mathcal{T}_i'$ to include all the scaled types. Thus the mechanism stays unchanged. The scaled type is only taken as an input of the oracle $\mathcal{A}$ when computing the payment.  

As we discussed in Section~\ref{sec:mechanism}, agent $i$ is incentive-compatible by defining the first-phase payment as VCG payments for the maximal-in-range allocation. $\mathcal{A}$ is thus implementable. With Theorem~\ref{thm:rochet}, the implicit payment for agent $i$ can be computed as follows: For every input $\textbf{r},\textbf{s}$, assume the true agent is at the $j$-th position with type $r_j$. We sample $\lambda\in [0,1]$ uniformly at random, and charge the agent $W_i(r_j,\mathcal{A}_j(\textbf{r},\textbf{s}))-W_i(r_j,\mathcal{A}_j(\lambda r_j,r_{-j},\textbf{s}))-(\delta\log\ell/\epsilon+\epsilon)$. Then in expectation the payment will be away from the one in Equation~\ref{equ:thm-rochet} by a fixed number. The agent will still be incentive-compatible. 

We remains to prove that with the first-phase payment defined above, the agent is individual rational. And in expectation over all agents' types, the total first-phase payment is close to 0. Lemma~\ref{lem:M' IR/NPT} is implied by the following lemma.
}

\notshow{
\begin{lemma}\label{lem:IR/NPT}
Suppose the agent $i$ is at the $j$-th position with type $t_j$ and let $\mathcal{A}_j$ be the allocation of matching surrogates to replicas $j$ in $\mout$ (Algorithm~\ref{alg:mout}). For every $\Br,\Bs$, define the payment for agent $i$ as

$$\mathcal{P}_j(\textbf{r},\textbf{s})=W_i(r_j,\mathcal{A}_j(\textbf{r},\textbf{s}))-\int_0^1 W_i(r_j,\mathcal{A}_j(\lambda r_j,r_{-j},\textbf{s}))d\lambda-(\delta\log\ell/\epsilon+\epsilon)$$

Moreover, for any true agent's type $r_j\in\mathcal{T}_i'$ and any other replicas' and surrogates' types $(r_{-j},\textbf{s})$, $\mathcal{P}_j(\textbf{r},\textbf{s})\geq -(\delta\log\ell/\epsilon+\epsilon)$. Furthermore,
$W_i(r_j,\mathcal{A}_j(\textbf{r},\textbf{s}))-\mathcal{P}_j(\textbf{r},\textbf{s})\geq 0$. And hence the agent $i$ is individual rational in $\mout$.
\end{lemma}

\begin{proof}

For the first part, it's sufficient to prove that $W_i(r_j,\mathcal{A}_j(\textbf{r},\textbf{s}))\geq W_i(r_j,\mathcal{A}_j(\lambda r_j,r_{-j},\textbf{s}))$ for every $\lambda\in [0,1]$. In fact consider the two allocations for $j$: $\mathcal{A}_j(r_j,r_{-j},\textbf{s})$ and $\mathcal{A}_j(\lambda r_j,r_{-j},\textbf{s})$. Let $x_j'$ be the vector of probabilities that $j$ is matched to each normal RHS-node in $\mathcal{A}_j(r_j,r_{-j},\textbf{s})$ and $y_j'$ for 0-RHS-nodes. Similarly define $(x_j'',y_j'')$ for $\mathcal{A}_j(\lambda r_j,r_{-j},\textbf{s})$. Notice that $\alpha^{(j)}$ only depends on the weights incident to the LHS-nodes $1$ to $j-1$. Thus by Observation~\ref{obs:exp weight with dummy-every replica}, both allocations are maximal-in-range for the same dual variables $\alpha^{(j)}$. Denote $H(x_j,y_j)=\sum_k(x_{jk}\log(x_{jk})+\sum_{k}y_{jk}\log(y_{jk}))$. We have
$$\sum_kx_{jk}'\omega_{jk}-\delta\cdot H(x_j',y_j')-\sum_k\gamma\alpha_k^{(j)}\cdot (x_{jk}'+y_{jk}')\geq \sum_kx_{jk}''\omega_{jk}-\delta\cdot H(x_j'',y_j'')-\sum_k\gamma\alpha_k^{(j)}\cdot (x_{jk}''+y_{jk}'')$$
$$\sum_kx_{jk}''\lambda\omega_{jk}-\delta\cdot H(x_j'',y_j'')-\sum_k\gamma\alpha_k^{(j)}\cdot (x_{jk}''+y_{jk}'')\geq \sum_kx_{jk}'\lambda\omega_{jk}-\delta\cdot H(x_j',y_j')-\sum_k\gamma\alpha_k^{(j)}\cdot (x_{jk}'+y_{jk}')$$

Summing up the two inequality together, we have
$$\sum_k(x_{jk}'-x_{jk}'')\omega_{jk}(1-\lambda)\geq 0$$

Since $\lambda\in [0,1]$ and the fact that $W_i(r_j,\mathcal{A}_j(r_j,r_{-j},\textbf{s}))=\sum_k x_{jk}'\omega_{jk}$, $W_i(r_j,\mathcal{A}_j(\lambda r_j,r_{-j},\textbf{s}))=\sum_k x_{jk}''\omega_{jk}$, we have $W_i(r_j,\mathcal{A}_j(\textbf{r},\textbf{s}))\geq W_i(r_j,\mathcal{A}_j(\lambda r_j,r_{-j},\textbf{s}))$.

For the second part, we will prove that for every $\lambda\in [0,1]$, $W_i(r_j,\mathcal{A}_j(\lambda r_j,r_{-j},\textbf{s}))\geq -\delta\log\ell/\lambda$. With this inequality, we have
\begin{align*}
&W_i(r_j,\mathcal{A}_j(\lambda r_j,r_{-j},\textbf{s}))-\mathcal{P}_j(\textbf{r},\textbf{s})\\
=&\int_0^1 W_i(r_j,\mathcal{A}_j(\lambda r_j,r_{-j},\textbf{s}))d\lambda+\delta\log\ell/\epsilon+\epsilon\\\
=&\int_0^\epsilon W_i(r_j,\mathcal{A}_j(\lambda r_j,r_{-j},\textbf{s}))d\lambda+\int_{\epsilon}^1 W_i(r_j,\mathcal{A}_j(\lambda r_j,r_{-j},\textbf{s}))d\lambda+\delta\log\ell/\epsilon+\epsilon\\
\geq& -\epsilon-\int_{\epsilon}^1\delta\log\ell/\lambda d\lambda+\delta\log\ell/\epsilon+\epsilon\\
=& -\epsilon-\delta\log\ell\cdot\log(\frac{1}{\epsilon})+\delta\log\ell/\epsilon+\epsilon \geq 0
\end{align*}

Using the same notation $(x_j'',y_j'')$ for the first part, we have $W_i(r_j,\mathcal{A}_j(\lambda r_j,r_{-j},\textbf{s}))=\sum_k x_{jk}''\omega_{jk}$. Again by Observation~\ref{obs:exp weight with dummy-every replica},
$$(x_j'',y_j'')=\argmax_{(x,y)}\sum_{k}x_{jk}\omega_{jk}-\delta\sum_{k}x_{jk}\log(x_{jk})-\delta\sum_{k}y_{jk}\log(y_{jk})-\sum_k\gamma\alpha_k^{(j)}\cdot (x_{jk}+y_{jk})$$

By considering an alternating solution $(0,x_j''+y_j'')$, we have
$$\sum_{k}x_{jk}''\lambda\omega_{jk}-\delta\cdot H(x_j'',y_j'')-\sum_k\gamma\alpha_k^{(j)}\cdot (x_{jk}''+y_{jk}'')\geq 0-\delta\cdot \sum_{k}(x_{jk}''+y_{jk}'')\log(x_{jk}''+y_{jk}'')-\sum_k\gamma\alpha_k^{(j)}\cdot (x_{jk}''+y_{jk}'')$$

Since $H(x_j'',y_j'')\geq 0$,

$$\sum_{k}x_{jk}''\omega_{jk}\geq -\frac{\delta}{\lambda}\sum_k(x_{jk}''+y_{jk}'')\log(x_{jk}''+y_{jk}'')\geq \delta\log\ell/\lambda$$

\end{proof}
}
\notshow{
\subsection{Missing Details of Section~\ref{sec:rev M'}}\label{sec:proof}
}
\notshow{
In this section we will bound the revenue loss from $M$ to $M'$. The proof is inspired by~\cite{HartlineKM11,CaiZ17}. For simplicity, we will slightly abuse the notation and denote $p_i(t_i)=\E_{t_{-i}\sim \DD_{-i}}[p_i(t_i,t_{-i})]$ the interim payment of agent $i$ throughout this section.






Fix agent $i$. For any realization of $\Br$ and $\Bs$, we first duplicate each surrogate $d$ times. Now the graph becomes complete bipartite graph between $d\ell$ replicas and $d\ell$ surrogates. Every $d$-to-1 matching in the old graph corresponds to a 1-to-1 matching in the new bipartite graph. Let $\ell'=d\ell$. \yangnote{But it is not a perfect matching...}

Recall $$W_i(\rj,\sk)=\E_{t_{-i}\sim \DD_{-i}}[v_i(\rj,x(\sk,t_{-i}))]-(1-\eta)\cdot p_i(\sk).$$ For any matching $L(\Br,\Bs)$, to simplify the notation, we use $W_i(L(\Br,\Bs))=\E_{\Br,\Bs}[\sum_{(r,s)\in L(\Br,\Bs)}W_i(r,s)]$. and $W_i(L)$ for the same value when $\Br,\Bs$ are fixed.

Let $O(\Br,\Bs)$ be the (randomized) matching obtained by $M'$. It's a distribution of perfect matchings between replicas and surrogates. Denote $V(\Br,\Bs)$ the maximum weight matching with respect to $W_i$. From Section~\ref{sec:arbitrary_weight_matching}, in expectation $O(\Br,\Bs)$ is nearly-optimal. We first give a lower bound of $\rev(M')$ by only considering the revenue from the second phase.

\begin{lemma}\label{lem:mout lower bound}
$$\rev(\mout)\geq (1-\eta)\cdot \sum_i\E_{\Br,\Bs}\left[\sum_{\sk\in O(\Br,\Bs)}p_i(\sk)/\ell'\right]-n(\delta\log\ell/\varepsilon+\varepsilon)$$
$$\rev(M)=\sum_i\E_{\Bs}\left[\sum_{k\in [\ell']}p_i(\sk)/\ell'\right]$$
\end{lemma}
\begin{proof}
The revenue of $\mout$ contains two parts -- the payment for the maximal-in-range matching (first phase payment) and the prices paid by the chosen surrogates (second phase payment). By Lemma~\ref{lem:M' IR/NPT}, the expected first phase payment is at least $-n(\delta\log\ell/\varepsilon+\varepsilon)$. Let's compute the second phase payment. For every agent $i$ and each realization of $\Br$ and $\Bs$, only when the agent $i$ is matched to a (normal) surrogate in $ O(\Br,\Bs)$, she pays the surrogate price. Since each surrogate is selected with probability $1/\ell'$, the expected surrogate price paid by agent $i$ is exactly $(1-\eta)\cdot\E_{\Br,\Bs}\left[\sum_{\sk\in O(\Br,\Bs)} \frac{ p_i(\sk)}{\ell'}\right]$. We have proved our first statement.
 
The expected payment from agent $i$ in $M$ is $\E_{t_i\sim \DD_i}\left[p_i(t_i)\right]$. Since all $\sk$ is drawn from $\DD_i$, this is exactly the same as $\E_{\Bs}\left[\sum_{k\in[\ell']} \frac{ p_i(\sk)}{\ell'}\right]$.
\end{proof}
}

\end{document}